\documentclass[journal]{IEEEtai}
\usepackage{amsmath,amsfonts,amssymb}
\usepackage{array}
\usepackage[caption=false,font=footnotesize,labelfont=rm,textfont=rm,]{subfig}
\usepackage{textcomp}
\usepackage{stfloats}
\usepackage{url}
\usepackage{verbatim}
\usepackage{graphicx}
\usepackage{cite}
\usepackage[table]{xcolor}
 \usepackage{threeparttable}
\usepackage{makecell}
\usepackage{diagbox}
\usepackage{hyperref}
\usepackage{amssymb}
\usepackage{multirow}
\usepackage{ragged2e}
\usepackage{stmaryrd}
\usepackage{booktabs}
\usepackage{epstopdf}
\usepackage{wasysym}
\usepackage{pifont}
\usepackage[ruled,vlined]{algorithm2e}
\usepackage{algorithmic}
\usepackage{amsthm}
\usepackage{graphicx}
\usepackage[font=footnotesize,labelsep=period]{caption}
\usepackage{indentfirst}
\usepackage{soul}
\SetKwFor{A}{Aggregator:}{}{endfor}
\SetKwFor{V}{Verifier:}{}{endfor}

\newtheorem{corollary}{Corollary}
\newtheorem{assumption}{Assumption}
\newtheorem{thm}{\bf Theorem}

\usepackage{color, xcolor}
\usepackage{fancyhdr}
\begin{document}

\title{PPFPL:  Cross-silo Privacy-preserving   Federated Prototype Learning Against Data Poisoning Attacks}

\author{Hongliang Zhang, Jiguo Yu,~\IEEEmembership{Fellow,~IEEE},  Fenghua Xu, Chunqiang Hu, Yongzhao Zhang, Xiaofen Wang, Zhongyuan Yu, Xiaosong Zhang

\thanks{This work was supported by the National Natural Science Foundation of China under Grants 62272256 and 62202250, the Major Program of Shandong Provincial Natural Science Foundation for the Fundamental Research under Grant ZR2022ZD03, the National Science Foundation of Shandong Province under Grant ZR2021QF079, the Colleges and Universities 20 Terms Foundation of Jinan City under Grant 202228093, and the Shandong Province Youth Innovation Team Project underGrant 2024KJH032. ($\textit{Corresponding author}$: $\textit{Jiguo Yu}$.)}
\thanks{H. Zhang is with the School of Computer Science and Technology, Qilu University of Technology, Jinan, 250353, China, Email: b1043123004@stu.qlu.edu.cn.}
\thanks{J. Yu is with  School of Computer Science and Engineering, University of Electronic Science and Technology of China, Chengdu, 611731, China, and also with the Big Data Institute, Qilu University of Technology, Jinan, 250353, China, Email: jiguoyu@sina.com; jiguoyu17@uestc.edu.cn.}
\thanks{F. Xu is with the Cyber Security Institute, University of Science and Technology of China,  Hefei, 230026, China, Email: nstlxfh@gmail.com.}
\thanks{C. Hu is with the School of Big Data and Software Engineering, Chongqing University,  Chongqing, 400044, China, Email:  chu@cqu.edu.cn.}
\thanks{Y. Zhang, X. Wang and X. Zhang are with the School of Computer Science and Engineering, University of Electronic Science and Technology of China,  Chengdu, 611731, China, Email:  zhangyongzhao@uestc.edu.cn, xfwang@uestc.edu.cn, johnsonzxs@uestc.edu.cn.}
\thanks{Z. Yu is with the College of computer science and technology, China University of Petroleum,  Qingdao, 266580, China, Email:  yuzhy24601@gmail.com.}
}
\markboth{}%
{}


\maketitle

\thispagestyle{fancy}

\begin{abstract}
Privacy-Preserving Federated Learning (PPFL)  enables  multiple clients to collaboratively train models  by submitting secreted model updates.
 Nonetheless, PPFL is vulnerable to data poisoning attacks due to its  distributed training paradigm   in cross-silo scenarios.
Existing solutions have struggled to improve the performance of  PPFL under   poisoned Non-Independent and Identically Distributed (Non-IID) data.
To address the issues, this paper proposes a privacy-preserving federated prototype learning framework, named  PPFPL,  which   enhances the cross-silo FL performance  against  poisoned Non-IID data  while protecting client privacy. 
Specifically, we adopt prototypes as client-submitted model updates to eliminate the impact  of poisoned data distributions.
In addition, we design a secure aggregation protocol  utilizing homomorphic encryption to achieve Byzantine-robust aggregation on two servers, significantly  reducing  the impact of malicious clients.
Theoretical analyses confirm the convergence and privacy  of PPFPL.
Experimental  results on public  datasets show that PPFPL effectively resists data poisoning attacks under  Non-IID settings.
\end{abstract}


\section{Introduction}
\label{sec:introduction}
Federated Learning (FL)  is a distributed learning paradigm  where each client shares its model updates instead of  raw training data.
In industrial applications, large volumes of data are distributed across independent organizations governed by  strict privacy regulations  \cite{10518181}.
To  break data silos among organizations without compromising privacy,  cross-silo FL provides a viable solution for industrial scenarios \cite{kalloori2024towards}\cite{jin2023federated}.
Specifically, in cross-silo FL,   clients   are usually large organizations, resulting in a relatively small number of clients with significant  computational capabilities.
However, client-submitted model  updates are vulnerable to  privacy  attacks, threatening the  security of cross-silo FL \cite{rao2024privacy}\cite{10767289}.
To mitigate  the privacy  risks,  Differential Privacy (DP)-based \cite{10400897}\cite{ling2024efficient}  and Homomorphic Encryption (HE)-based \cite{10989521}\cite{10733641}  PPFL approaches are  proposed.
In particular,  DP-based approach is commonly  applied to cross-device FL  due to its low computation overhead, but degrades  the FL performance by injecting noise.
Conversely,  HE-based  approach  provides   higher privacy security without sacrificing  the global model accuracy,   making it more suitable for cross-silo FL.

While  HE-based  approaches  have  demonstrated   their effectiveness   in terms of  privacy preservation, they are  susceptible  to data poisoning attacks \cite{10274102}\cite{9900151}.
Specifically, malicious clients  launch data poisoning attacks by tampering with their raw training data and submitting  model updates derived from the poisoned data, thereby   degrading   the performance of  PPFL \cite{10399805}\cite{zhu2023adfl}.
Moreover, privacy-preserving techniques obscure the model updates from malicious clients, making data poisoning attacks more concealed and  difficult to defend.
To audit secreted  model updates, existing methods \cite{hao2021efficient,9524709,9849010,10531276,10777580} utilize HE  and Secure Multi-party Computing (SMC) to  identify   malicious model updates within  ciphertext.
 However, these works overlook  the data heterogeneity in FL, i.e., the data among clients is typically non-independent and identically distributed (Non-IID). 
The Non-IID data leads to the inconsistency among benign model updates, making it difficult  for these defense methods \cite{hao2021efficient,9524709,9849010,10531276,10777580} to distinguish whether the deviations are caused by data heterogeneity or by malicious  manipulation.
To  distinguish between benign   and malicious  updates under Non-IID data, existing methods \cite{9762272,10250861,liu2025antidotefl} employ    clustering or adaptive aggregation weighting operation to mitigate the inconsistency of model updates during the auditing phase.
However, in  data poisoning attacks, malicious clients manipulate  the features and labels of  their  training data, thereby forming poisoned Non-IID data.
These manipulations distort the optimization direction of local  training,  causing these methods \cite{9762272,10250861,liu2025antidotefl} to optimize toward the  tampered   Non-IID distribution, severely   degrading the performance of FL.
Therefore, a critical challenge remains: how to enhance  the performance  of  PPFL under  poisoned Non-IID data while effectively   resisting data poisoning attacks.

Inspired by  prototype learning \cite{yang2018robust},  several  works \cite{tan2022fedproto,guo2024dynamic,huang2023rethinking,mu2023fedproc} have introduced  prototype aggregation  to address  the Non-IID problem in FL by exchanging prototypes between servers and clients.
Each  prototype represents a class-level feature,   computed as the mean of the feature representations of samples within  the same class.
For instance, the work  in \cite{huang2023rethinking} suggests that in recognizing the  class ``cat,'' different clients have their own unique ``imaginary picture'' or ``prototype'' to represent the concept of ``cat''.
By prototype exchange, clients  gain more knowledge about the concept of ``cat''.
Consequently,   prototypes, as feature representations independent of local data distributions, motivate our work.

In cases of  poisoned Non-IID data,  both tampered data features and    distributions cannot be repaired  since the server cannot control malicious client behavior.
This prompts us to pose a question: \textit{Is it possible to design a PPFL that leverages prototype learning, so that client-submitted model updates are affected only by  tampered data features rather than by  tampered data distribution, while  incorporating   a secure aggregation protocol to  eliminate the impact of tampered data features,  thereby achieving Byzantine-robust results  (i.e., guaranteeing reliability under arbitrary and potentially malicious behaviors in distributed computing systems \cite{yin2018byzantine})?}
To answer this question, we propose a \textbf{P}rivacy-\textbf{P}reserving \textbf{F}ederated \textbf{P}rototype \textbf{L}earning framework, named PPFPL, which is  suitable for  cross-silo scenarios.
The PPFPL framework consists of two ``non-colluding'' servers (assumed to be honest but curious) and multiple clients.
The key novelty of PPFPL lies in: $(\textbf{\text{i}})$ 
We propose a novel local optimization function leveraging prototype learning, where  each client submit prototypes to two  servers,  which mitigates the impact  of tampered data distributions.
Unlike existing prototype-based works \cite{tan2022fedproto}\cite{guo2024dynamic}, our optimization function employs cosine similarity to constrain  the influence of malicious clients (see the proof of  Theorem \ref{14785414}).
$(\textbf{\text{ii}})$
 We design a secure aggregation protocol between  two  servers to aggregate client-submitted prototypes using  HE and SMC techniques.
  This protocol achieves Byzantine-robust results while preserving  the privacy of benign clients.
  
 Our main  contributions   are  summarized as follows.
\begin{itemize}
 \item  
To the best of our knowledge, this is the first work that introduces prototype learning into PPFL to defend against data poisoning attacks.
 By transferring prototypes between clients and  servers, our framework mitigates the impact of  tampered data distribution and enhances the FL performance in poisoned Non-IID data.

  \item We employ HE and SMC techniques to design a secure aggregation protocol across two servers that filters malicious prototypes, thereby ensuring Byzantine-robust results while protecting the privacy of benign clients.
Notably, existing studies have not provided solutions for privacy protection in prototype-based federated learning, and our work fills this important research gap.

    \item  We establish theoretical guarantees on convergence and privacy for PPFPL,  thereby ensuring the framework's feasibility.

  \item Compared to existing methods, the superiority of our framework has been empirically validated in poisoned Non-IID data.
\end{itemize}

The rest of the paper is organized as follows. Section \ref{sec:RELATED WORK} reviews  PPFL works  against poisoning attacks. Section \ref{sec:PRELIMINARIES} introduces prototype learning  and HE. Section \ref{sec:Problem Formulation} formalizes both  the system model and threat model. Our PPFPL is detailed  in Section \ref{sec:Detailed design of PPFPL}. Section \ref{sec:Aggregation and security analysis} provides   theoretical analysis. Section \ref{sec:Experiments} reviews  experimental results. Finally, Section \ref{sec:concludesq} concludes this paper.

\section{RELATED WORK}
\label{sec:RELATED WORK}
\textit{Privacy preservation in FL:}
Although  FL inherently provides a degree of privacy protection, it remains  vulnerable to  privacy attacks, which causes the privacy threat of benign clients.
To resist such  attacks, DP-based and HE-based  approaches have been proposed to preserve client-submitted model updates.
Specifically,  DP-based schemes  \cite{3394533,zhao2020local,9347706,WANG2022446} deploy local differential privacy into model updates, ensuring  privacy without compromising the utility of model updates.
Despite their low computation overheads, these schemes introduce Gaussian noise or Laplace noise into  local model training, which inevitably degrades  the FL performance   to some extent.
In contrast, HE is a commonly adopted cryptographic primitive in across-silo PPFL that provides strong privacy preservation without sacrificing  the global model's accuracy.
Fang and Qian are one of the first scholars to implement PPFL using HE \cite{fi13040094}. They proposed a multi-party machine learning scheme using Paillier \cite{paillier1999public} technique  without compromising clients' privacy.
Considering the heavy communication overhead of Paillier, the work  in \cite{22818}  proposed a privacy-preserving FL  using CKKS (i.e., Cheon-Kim-Kim-Song) that  reduces  computational overhead associated with ciphertexts.
This is because CKKS is more efficient and better suitable to handle large-scale vector and multi-parameter network models compared to Paillier \cite{li2021security}.
However,  the aforementioned schemes   overlook  the threat of data poisoning attacks caused by distributed training.

\textit{Resisting  data poisoning attacks in FL:}
Tolpegin et al. demonstrated that  data poisoning attacks can severely  reduce the classification accuracy of FL, even with a small percentage of malicious clients \cite{10007}.
Additionally, they proposed a defense strategy during aggregation  that    identifies malicious clients  to circumvent data poisoning attacks.
Similarly, the works in \cite{blanchard2017machine,mu2024feddmc,10549523} introduced  detection  mechanisms  to filter  malicious model updates. 
Differently, Doku et al. employed  Support Vector Machine (SVM)  to  audit  client's local training data for excluding malicious clients \cite{9369581}.
However, this work violates the privacy of  clients to some extent.
Furthermore, the threats of  privacy  and data poisoning attacks usually coexist in  practice across-silo FL.
While PPFL approaches employ cryptographic primitives to ensure client privacy,  they mask data poisoning attacks from malicious clients.

\textit{Resisting  data poisoning attacks in PPFL:}
Considering the threats of both privacy and data poisoning attacks,  the works such as \cite{9524709,hao2021efficient,9849010,10531276,10777580} integrate  HE and SMC to detect  anomalous model updates directly within  ciphertexts.
For instance,
the work in \cite{9524709} proposed  a privacy-enhanced FL  framework that adopts HE as the underlying technology and provides  two servers   to punish malicious clients  via   gradient extraction of  logarithmic function.
Similarly,  the work in \cite{hao2021efficient}  designed a validity checking protocol for ciphertexts under  two servers, which protects data privacy and \ adjust the weight of clients' gradients to weaken  data poisoning attacks.
Notably, these works \cite{9524709,hao2021efficient,9849010,10531276,10777580} adopt a two-server  architecture instead of a single server, which  facilitates secure aggregation in  adversarial scenarios.
However, they neglect  the deviations  caused by Non-IID data, making it difficult to distinguish whether the deviations come from malicious manipulation or Non-IID data.


To audit malicious  updates in  PPFL with Non-IID data,  the works in \cite{9762272}\cite{10250861} are proposed to eliminate the deviations caused by Non-IID data during the aggregation stage.
Specifically,  ShieldFL  designed a Byzantine-tolerant aggregation mechanism to  prevent misjudgments  on outliers caused by Non-IID data \cite{9762272}.
Furthermore,  Chen et al.  adopt  clustering  combined with  cosine similarity and median strategies to eliminate  deviations  among model updates during aggregation auditing \cite{10250861}.
These  schemes can only resist model poisoning attacks confronted by federated learning with Non-IID data, but they cannot essentially improve the performance of FL on Non-IID data.
They need to be combined with  specialized   training techniques  designed for  Non-IID data (e.g., FedProx\cite{li2020federated}, FedDyn\cite{acar2021federated}, or FedLC\cite{zhang2022federated}, etc.)   to radically improve performance of FL.
These specialized techniques  introduce   auxiliary terms into  local optimization function to constrain client model updates,  thereby improving  consistency of model updates during local model training.
However, in data poisoning attacks, malicious clients tamper with the features and labels of their  training data,  generating  poisoned Non-IID data. This manipulation misleads specialized   training techniques  to  adjust model updates based on these compromised inputs, which ultimately   cause the global model to converge to a malicious objective.
Therefore, the above works \cite{9762272}\cite{10250861} cannot effectively integrate these specialized techniques  to resist data poisoning attacks while enhancing the FL performance  under poisoned Non-IID data.

\textit{Federated prototype learning:}
Recently, prototype learning  has  been   applied in federated learning to address the Non-IID issue.
Specifically,   the works in \cite{tan2022fedproto,huang2023rethinking,guo2024dynamic,mu2023fedproc} are  one of the first to  propose federated prototype learning using the concept  of prototype learning.
Different from the specialized training techniques (e.g., FedProx, FedDyn, and FedLC),  the core idea enables clients to pull the same-class samples  towards the global prototypes of that class and away from the global prototypes of other classes.
In other words, each class holds  its corresponding   prototype that is independent of other classes.
Consequently, client-submitted prototype   is  affected by the samples but is independent of data distribution among clients.
This insight motivates our work.

Although previous works   \cite{tan2022fedproto,huang2023rethinking,guo2024dynamic,mu2023fedproc} have employed prototype learning to enhance the FL performance  on Non-IID  data,  they are vulnerable to data poisoning attacks in distributed scenarios.
Specifically, malicious clients can compromise the performance of federated training by uploading their poisoned prototypes.
Furthermore, existing work \cite{103612072} designed a dynamic memory  model inversion attack that can recover the private training data by utilizing  client's learned prototypes, which seriously damages the client's privacy.
To this end, our proposed PPFPL   mitigates  the performance degradation caused by data poisoning attacks on Non-IID  data while ensuring client privacy.


\section{PRELIMINARIES}
\label{sec:PRELIMINARIES}
This section introduces  prototypes in federated  learning and CKKS technology.
In addition,   TABLE \ref{notation} provides explanations of the acronyms used in this paper.
\begin{table}[]
 \caption{Summary of main acronyms.}
\centering
\tabcolsep=0.15cm
\scalebox{0.76}{
\begin{tabular}{c|c|c|c}
\hline Notation & Meaning &Notation & Meaning  \\
\hline  FL   & Federated Learning & PPFL   & Privacy-Preserving Federated Learning \\
 \hline
 HE   & Homomorphic Encryption & CKKS & Cheon–Kim–Kim–Song encryption  \\
 \hline
 SMC  & Secure Multi-Party Computation &  DP   & Differential Privacy \\
 \hline
 KGC  & Key Generation Center & SVM  & Support Vector Machine \\
\hline
 \multirow{2}{*}{Non-IID }  &  \multirow{2}{*}{\makecell{Non-Independent and \\ Identically Distributed}}    &  \multirow{2}{*}{RSA}   & \multirow{2}{*}{\makecell{Rivest–Shamir–Adleman cryptosystem}} \\
 &  &    &  \\
 \hline
\end{tabular}}
\label{notation}
\end{table}

\subsection{Prototypes Meet Federated  Learning} \label{sec: Federated Prototype Learning}
In the classification task of prototype learning, the prototype is defined as a feature vector representing a specific  class \cite{Zhu2021CVPR}.
This inherent property ensures that prototypes of the same class  are similar in  FL task.
Consequently, many FL schemes \cite{tan2022fedproto,huang2023rethinking,guo2024dynamic,mu2023fedproc}  employ prototype learning to address Non-IID challenges,   enabling clients to align their local prototypes with global class representations during local model training.

To  understand the prototype  calculation  in federated prototype learning, we  introduce some basic  notations below.
Let $\mathcal{S}$ be the set of clients, where  each client $m \in \mathcal{S}$  owns a   private dataset, denoted as \scalebox{0.8}{$\mathcal{D}_m =\{(\textit{\textbf{x}}_{(i)},y_{(i)})\}^{|\mathcal{D}_m|}$}. Here, $|\mathcal{D}_m|$ is the number of samples in client $m$, and $(\textit{\textbf{x}}_{(i)},y_{(i)})$ denotes sample $i$ in dataset, where $\textit{\textbf{x}}_{(i)}$ and $y_{(i)}$ correspond to the feature vector and class label of sample $i$, respectively.
Meanwhile, let $\mathcal{I}$ be   the set of classes in  classification task, where each class $k$ belongs to $\mathcal{I}$.
In  classification task, the local model  includes a feature extractor and a decision classifier.
Specifically, the  feature extractor transforms   raw sample features into compressed representations, while the decision classifier maps the compressed features to get  classification results.
Formally, let $f_m(\textit{\textbf{r}}_{m,t};\cdot)$ be feature extractor for client $m$, parameterized by $\textit{\textbf{r}}_{m,t}$, where $t$ denotes the $t$-th communication round.
Given the feature $\textit{\textbf{x}}_{(i)}$ of  sample $i$,  it is input into feature extractor to obtain  compressed feature  $\textit{\textbf{u}}_{(i)}= f_m(\textit{\textbf{r}}_{m,t};\textit{\textbf{x}}_{(i)})$.
Let $g_m(\textit{\textbf{z}}_{m,t}; \cdot)$ be  decision classifier for client $m$, parameterized by $\textit{\textbf{z}}_{m,t}$. The  classifier maps the compressed feature $\textit{\textbf{u}}_{(i)}$ to   predict the  class $y^\prime=g_m(\textit{\textbf{z}}_{m,t};\textit{\textbf{u}}_{(i)})$.
Thus, we denote the local model as $\mathcal{F}_m((\textit{\textbf{r}}_{m,t},\textit{\textbf{z}}_{m,t});\cdot)=g_m(\textit{\textbf{z}}_{m,t}; \cdot)\circ f_m(\textit{\textbf{r}}_{m,t}; \cdot)$, where $\circ$ denotes composition operator.
For simplicity, we use $\textit{\textbf{w}}_{m,t}$ to denote $(\textit{\textbf{r}}_{m,t},\textit{\textbf{z}}_{m,t})$, so we have $\mathcal{F}_m((\textit{\textbf{r}}_{m,t},\textit{\textbf{z}}_{m,t}); \cdot)= \mathcal{F}_m(\textit{\textbf{w}}_{m,t};\cdot)$, and $\textit{\textbf{w}}_{m,t}$ is considered as  model parameters  for client $m$.
Next, we  present   calculation process of  prototypes.

In federated prototype learning, prototypes can be  categorized into local prototypes (computed by   clients) and global prototypes (aggregated by the server).
Specifically,  each client's goal  is to  align its  local prototypes with global prototypes during local model training.
Each  client computes its local prototype via its training dataset during local model training.
Formally, let $\textit{\textbf{c}}^k_{m,t}$ be the local prototype of class $k \in \mathcal{I}$ at client $m$ in $t$-th communication round, calculated as:
\begin{equation}\scalebox{1}{$
\begin{aligned}
\textit{\textbf{c}}^{k}_{m,t}= \frac{1}{|\mathcal{D}^k_m|}\sum_{(\textit{\textbf{x}}_{(i)},y_{(i)}) \in \mathcal{D}^k_m} f_m(\textit{\textbf{r}}_{m,t};\textit{\textbf{x}}_{(i)}), \forall k \in \mathcal{I},\label{92611}
\end{aligned}$}
\end{equation}
where $\mathcal{D}^k_m$ denotes the dataset  with class $k$ at client  $m$.
$\textit{\textbf{c}}^k_{m,t}$ can be understood as the mean of  compressed features of  samples belonging to class $k$ at client $m$.
Further, the set of local prototypes for all classes at client $m$ is denoted as \scalebox{0.8}{$\{\textit{\textbf{c}}^k_{m,t}\}_{k\in \mathcal{I}}^{|\mathcal{I}|}$}.
After completing local training, each client submits its  local prototypes to the server for aggregation.

 To calculate  global prototype, the server adopts  an averaging operation on local prototypes submitted by clients.
 Thus,  the global prototype $\textit{\textbf{C}}^k_{t+1}$ for  each class   is calculated as follows:
\begin{equation}\scalebox{1}{$
\begin{aligned}
\textit{\textbf{C}}^k_{t+1}& = \frac{1}{|\mathcal{S}|}\sum_{m \in \mathcal{S}} \textit{\textbf{c}}_{m,t}^k, \forall k \in \mathcal{I},
\end{aligned}$}
\end{equation}
where  $|\mathcal{S}|$ denotes the number of clients.
The set of global prototypes for all classes is denoted as \scalebox{0.8}{$\{\textit{\textbf{C}}^k_{t+1}\}_{k\in \mathcal{I}}^{|\mathcal{I}|}$}.
Subsequently, the server distributes the latest global prototypes to each client  to  further train  their local model.



\subsection{CKKS Technique}
A ciphertext is the encrypted form of plaintext data, generated   through an encryption algorithm.
Homomorphic encryption  allows  operations to be performed directly on ciphertexts without decryption. However,  traditional  HE schemes such as RSA \cite{wiener1990cryptanalysis}, ElGamal \cite{tsiounis1998security}, and Paillier are limited in that they only support either additive or multiplicative operations but not both simultaneously.
In contrast, the CKKS technology provides both additive and multiplicative homomorphic encryption, which is known as full HE \cite{10265238}.
In addition, CKKS is known for its efficiency, especially in terms of encryption/decryption speed when applied to large-scale vectors with varying parameter lengths. Therefore, we employ CKKS to ensure client privacy while maintaining computational efficiency.
The CKKS scheme mainly consists of key generation, encoding, encryption,  addition,  multiplication, decryption, and decoding. Here is a brief description of each operation:
\begin{enumerate}
\item  \textit{Key generation.} $\textit{KeyGen}(\cdot)$: Given a  security parameter $\kappa$, the $\textit{KeyGen}(\cdot)$ generates secret key $Sk$, public key $Pk$ and evaluation key  for ciphertext calculations.
\item  \textit{Encoding}. $\textit{Ecd}(\cdot)$: Given a $(\frac{N}{2})$-dimensional vector \textit{\textbf{t}} and a scaling factor, the $\textit{Ecd}(\cdot)$ encodes the factor $\textit{\textbf{t}}$ into a polynomial $\mu$.
\item \textit{Encryption}. $\textit{Enc}(\cdot)$: Given a polynomial $\mu$, the $\textit{Enc}(\cdot)$ encrypts $\mu$ by  the public key $Pk$ and generates a ciphertext $\psi$.
\item \textit{Addition}. $\textit{Add}(\cdot)$:  Given   a pair of ciphertexts  $\psi_1$ and $\psi_2$, the $\textit{Add}(\cdot)$ outputs  a ciphertext $\hat{\psi} = \psi_1 \oplus \psi_2$, where $\hat{\psi}$ is the ciphertext of the sum of  plaintexts of $\psi_1$ and $\psi_2$.
\item   \textit{Multiplication}. \textit{Mult}$(\psi_1,\psi_2, evk)$: Given  a pair of ciphertexts $\psi_1$ and $\psi_2$,  the \textit{Mult}$(\cdot)$ outputs  a ciphertext  $\overline{\psi} = \psi_1 \otimes \psi_2$, where $\overline{\psi}$ is the ciphertext of the product  of  plaintexts of $\psi_1$ and $\psi_2$.
\item \textit{Decryption}. $\textit{Dec}(\cdot)$:  Given a  ciphertext $\psi$, the $\textit{Dec}(\cdot)$ generates a polynomial $\mu$ by the secret key $Sk$.
\item  \textit{Decoding}. $\textit{Dcd}(\cdot)$ : Given  an input polynomial $\mu$ and a scaling factor $s$, the  $\textit{Dcd}(\cdot)$  outputs the vector \textit{\textbf{t}}.
\end{enumerate}
For more details on the implementation principles and procedures of CKKS, please refer to the work in \cite{cheon2017homomorphic}.

\section{Problem Statement}
\label{sec:Problem Formulation}
In the section, we formalize the  PPFPL  framework, define potential threats, and   design goals.
\subsection{PPFPL framework}
The  framework of PPFPL consists of four entities, each of which has its specific function, as shown in Fig. \ref{fig2}.
The interaction among  these entities composes the  operation of  whole system.  
The roles of each entity are outlined  as follows.
\begin{figure}[!t]
\centerline{\includegraphics[width=0.7\columnwidth]{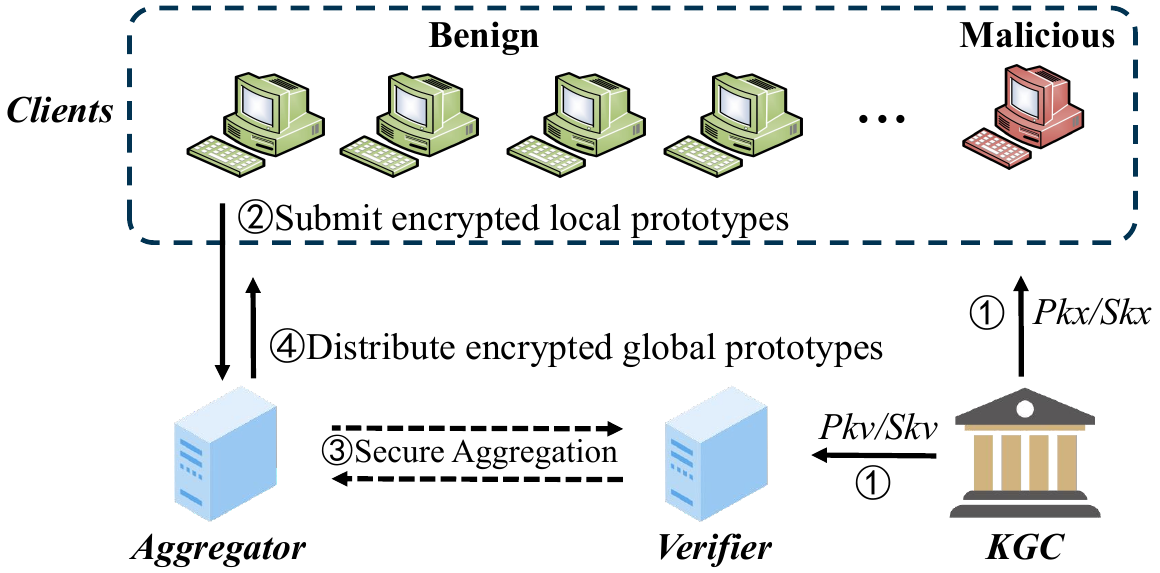}}
\caption{ The  PPFPL framework.
\ding{172}  \textit{KGC} generates $Pkv/Skv$ for  \textit{Verifier} and $Pkx/Skx$ for \textit{Clients}.
\ding{173} After local training is completed,  clients submit  encrypted local prototypes to the \textit{Aggregator}.
\ding{174}  \textit{Verifier} and \textit{Aggregator} perform  a secure  aggregation protocol to get  encrypted global prototypes.
\ding{175}   \textit{Aggregator} distributes the  encrypted global prototypes to \textit{Clients}.
}
\label{fig2}
\end{figure}
\begin{itemize}
\item \textit{Key Generation  Center (KGC)}.
The entity is responsible for generating and managing the keys of both \textit{Clients} and \textit{Verifier}, which are  essential elements to ensure  security of  encryption/decryption process.
\item  \textit{Clients}. The \textit{Clients} are large organizations participating in federated training. The aim of  benign  organizations  is to get a better  model by federated training. They have a pair of   public/secret keys generated by \textit{KGC},  denoted as $Pkx$/$Skx$.
\item \textit{Aggregator}. The \textit{Aggregator} is a central server responsible for aggregating local prototypes submitted by clients.
\item \textit{Verifier}. The \textit{Verifier} is a  non-colluding central server that cooperates with the \textit{Aggregator} to aggregate local prototypes. It has a pair of public/private keys  generated by \textit{KGC},  denoted as $Pkv$/$Skv$.
\end{itemize}

In our framework, we not only define the function of each entity, but also define  potential threats.

\subsection{Potential Threats  for PPFPL}

We discuss  potential threats of   PPFPL    in detail.

\textit{i}) The \textit{KGC} is a trusted institution  (e.g., government, union).

\textit{ii}) The \textit{Aggregator} and \textit{Verifier} are assumed  to be  non-colluding  and  curious but honest.
Specifically,  we assume that  \textit{Aggregator} and \textit{Verifier} do not collude to attack PPFPL.
The assumption  is reasonable in practice, since it is usually impossible for two well-known  service providers to collude with each other due to legal regulations and company reputation \cite{10160091}.
Furthermore, we assume that they  follow the system's protocol but  may attempt  to get sensitive information (i.e., raw training data)  by inferring client-submitted local prototypes.
Notably, although prototypes differ from gradients, traditional inference attack methods cannot recover  client's private training data from  prototype.
Although prototypes differ from gradients, and traditional inference attacks cannot directly reconstruct private training data from them, the work in \cite{103612072} designs a dynamic memory model inversion attack that can recover the private training data by utilizing  client's learned prototypes.
Thus, preserving the privacy of prototypes is  necessary.
Notably, we do not consider security attacks resulting from employee insider threats (e.g.,  compromised employees within the \textit{Aggregator} or \textit{Verifier}).

\textit{iii})
In practice, we cannot guarantee that all clients in the system are honest.
 Therefore, clients can be either benign or malicious.
Specifically, benign clients are ``honest but curious''.
For malicious clients, they can collude  to infer sensitive information about benign clients.
We assume that the proportion of malicious clients is less than 50\%, which is a more realistic threat  in cross-silo FL due to  the  high reputation of large organizations participating.
Furthermore, malicious clients can   launch data poisoning attacks  by submitting   malicious  prototypes derived from  poisoned data.
In federated prototype learning, we define the following two representative types of data poisoning attacks from the perspective of data features and labels.
\begin{itemize}
  \item  \textbf{Feature attacks}.
The purpose of feature attacks is to degrade the performance of federated learning. Specifically, malicious clients tamper with the features  of  their training data to  generate malicious prototypes for uploading, thereby affecting the training results of other clients.
  \item \textbf{Label attacks}.
The malicious clients tamper with the labels of their training samples into other random labels to  form the tampered data distribution, aiming to reduce the performance of the entire federated learning.
\end{itemize}


\subsection{Design Goals of PPFPL}
The goal of our study is to  improve model performance in cross-silo federated learning under poisoned Non-IID data while resisting data poisoning attacks.
Specifically,  we design the PPFPL framework to fulfill the following  goals:
\begin{itemize}
\item
\textbf{Security}. PPFPL should guarantee the robust model training in  the presence of data poisoning attacks with Non-IID data. In other words, the model performance of each benign client is not affected by data poisoning attacks under different data distributions.
\item
\textbf{Privacy}. PPFPL should ensure the privacy and security of  benign clients. For any third entity, they cannot access  the sensitive information about  benign clients.
\item
\textbf{Efficiency}. PPFPL should reduce the number of parameters submitted to two servers compared to other similar schemes, thereby reducing both privacy computation    and communication overheads.
\end{itemize}

\section{Design of  PPFPL}\label{sec:Detailed design of PPFPL}
In this section, we first provide the  overview of PPFPL, and  then describe each step in detail.
\subsection{Overview of  PPFPL}
The execution process of PPFPL  is summarized in  \textbf{Algorithm 1}.
The framework  initializes the number of communication rounds $T$,  the number of local iterations $E$, and  the local model's parameters  $\textit{\textbf{w}}^{\textit{init}}$ for each client. Then,  \textit{Aggregator}, \textit{Verifier}, and \textit{Clients} jointly perform FL training.  Specifically, PPFPL iteratively performs the following two steps:
\begin{itemize}
  \item \textit{Step I. Local Computation:}
  Each client trains its  local model with its  local dataset. Then, the client  normalizes and encrypts its local prototypes before submitting  them to \textit{Aggregator}. 
  \item \textit{Step II. Secure Aggregation Protocol:}
  The two servers  verify the normalization of  encrypted local prototypes submitted by clients, and perform the  secure two-party computation   to get encrypted global prototypes, which are then distributed to each client.
\end{itemize}
The above process  repeats until configured number of communication rounds $T$.
In the following, we describe the process of \textit{Steps I} and   \textit{II}  in detail.
\begin{algorithm}[t]
\begin{small}
    \caption{$\textbf {Overview of  PPFPL}$}
    \label{algorithm:test}
    \LinesNumbered
    \KwIn {$\mathcal{S}$, $\textit{\textbf{w}}^{\textit{init}}$, $E$, $T$.
}
    \KwOut {Model parameter of each client.}
Initialize $T,E,\textit{\textbf{w}}^{\textit{init}}$;\\
 \textit{Aggregator} distributes $T,E,\textit{\textbf{w}}^{\textit{init}}$ to each client;\\
\For{each {\rm  communication round} $t \in \{1,2,\cdots,T\}$}{
// \textbf{Step I}: $\mathtt{ Local \ Computation}$.\\
\For{each {\rm   client } $m \in \mathcal{S}$}{
{\rm Train local model};\\
{\rm Normalize  and encrypt  local prototypes};\\
{\rm Send encrypted  local prototypes to}  \textit{Aggregator};\\
}
// \textbf{Step II}:  $\mathtt{Secure\ Aggregation\ Protocol}$.\\
{\rm Two servers verify  normalization};\\
{\rm Two servers compute global prototypes};\\
\textit{Aggregator} distributes  global prototypes to  each client;\\
}
 \Return{{\rm Model parameter of each client}}
 \end{small}
\end{algorithm}

\subsection{Local Computation}

The local computation step includes two essential stages: \textit{local model training} and \textit{prototype handling} stages. 
The details are outlined in \textbf{Algorithm 2}.

\begin{algorithm}[t]
\begin{small}
    \caption{$\textbf {Local Computation}$}
    \label{algorithm:test}
    \LinesNumbered
    \KwIn {$ \mathcal{S}$, $ \mathcal{D}_m$, $\eta$ ,  $\{\llbracket\textit{\textbf{C}}^{k}_{t}\rrbracket_{Pkx}\}_{k\in \mathcal{I}}^{|\mathcal{I}|}$, $E$.}
    \KwOut {Encrypted local prototypes.}
    \For{each {\rm client} $m$ $\in$ $ \mathcal{S}$}{
	Get  global prototypes $\{\llbracket\textit{\textbf{C}}^{k}_{t}\rrbracket_{Pkx}\}_{k\in \mathcal{I}}^{|\mathcal{I}|}$ from \textit{Aggregator};\\
	Decrypt $\{\llbracket\textit{\textbf{C}}^{k}_{t}\rrbracket_{Pkx}\}_{k\in \mathcal{I}}^{|\mathcal{I}|}$ using its $Skx$;\\
	$ \textit{\textbf{w}}_{m,t}^{(E)} \leftarrow \textbf {Training}$($ \textit{\textbf{w}}_{m,t-1}^{(E)}$, $ \mathcal{D}_m$, $\eta$, $E$, $\{\textit{\textbf{C}}^{k}_{t}\}_{k\in \mathcal{I}}^{|\mathcal{I}|}$);\\
	$\{\llbracket\widetilde{\textit{\textbf{c}}}_{m,t}^{k}\rrbracket_{Pkv}\}_{k\in \mathcal{I}}^{|\mathcal{I}|}\leftarrow\textbf{Handling}(\textit{\textbf{w}}_{m,t}^{(E)}, \mathcal{D}_m)$;\\
	Send \scalebox{0.8}{$\{\llbracket\widetilde{\textit{\textbf{c}}}_{m,t}^{k}\rrbracket_{Pkv}\}_{k\in \mathcal{I}}^{|\mathcal{I}|}$} to $\textit{Aggregator}$;\\
    }
    \Return{$\{\llbracket\widetilde{\textit{\textbf{c}}}_{m,t}^{k}\rrbracket_{Pkv}\}_{k\in \mathcal{I}}^{|\mathcal{I}|}$}
    \end{small}
\end{algorithm}

\subsubsection{Local Model Training}
During local model training, each client aims to minimize its classification loss while aligning its  local prototype close to the global prototype.
To  achieve this, we design  an auxiliary term by  leveraging
 prototype learning  in the local optimization function.
Formally, the local optimization function of client $m$ is defined as:
\begin{equation}\scalebox{1}{$
\begin{aligned}
\mathcal{L}(\textit{\textbf{w}}_m;  \mathcal{D}&_m,\textit{\textbf{c}}^{k}_m,\textit{\textbf{C}}^{k}) =\mathcal{L_S}(\mathcal{F}_m(\textit{\textbf{w}}_m;\textit{\textbf{x}}_{(i)}),y_{(i)})\\&+\lambda \mathcal{L_R}(\textit{\textbf{c}}^{k}_m,\textit{\textbf{C}}^{k}),
\forall (\textit{\textbf{x}}_{(i)},y_{(i)})\in \mathcal{D}_m,\forall k\in \mathcal{I}, \label{925-4}
\end{aligned}$}
\end{equation}
where  $\mathcal{L_S}(\cdot,\cdot)$ is the classification loss function (e.g., cross-entropy loss function), $\lambda$ is the importance weight of the auxiliary term, and $\mathcal{L_R}(\cdot,\cdot)$ is the  auxiliary term, defined as:
\begin{equation}\scalebox{1}{$
\mathcal{L_R}(\textit{\textbf{c}}^{k}_m,\textit{\textbf{C}}^{k}) = \frac{1}{|\mathcal{I}|}\sum_{k\in\mathcal{I}} (1-\textit{sim}(\textit{\textbf{c}}_m^k,\textit{\textbf{C}}^k)),\label{10101}$}
\end{equation}
where $\textit{sim}(\cdot,\cdot)$ denotes cosine similarity between two vectors. The cosine similarity ranges from $-1$ to $1$, where the value closer to $``1"$ indicates similar vector directions, while value closer to $``-1"$ means opposing directions.
The optimization function ensures that each client reduces its classification loss while aligning its local prototype with the global prototype in direction.
Next, we present the process of local model training.

Specifically, in the $t$-th communication round, each   client  $m$ uses its local dataset $ \mathcal{D}_m$ to iteratively train  its local model, as detailed in \textbf{Algorithm 3}.
The superscript $(e)$ indicates the iteration state of variables, where $e \in\{1,\cdots,E\}$.
Each client $m$  uses its local model parameters  $\textit{\textbf{w}}_{m,t-1}^{(E)}$ from the $(t-1)$-th round   as the starting point   in the current $t$-th round. Subsequently, each client iteratively performs the following stages.
\begin{itemize}
  \item  In the $e$-th local iteration, the  client   randomly selects  training data $\mathcal{D}_{m}^{(e)}$ from its local dataset $ \mathcal{D}_m$.
  \item   The client inputs the training data \scalebox{0.8}{$\mathcal{D}_{m}^{(e)}$} into the local model's feature extractor to compute   the local prototype $\textit{\textbf{c}}^{k,(e)}_{m,t}$ via formula (\ref{92611}).
  \item The  client computes the unbiased stochastic gradient by
  		$$\textit{\textbf{g}}_{m,t}^{(e-1)} = \nabla\mathcal{L}(\textit{\textbf{w}}_{m,t}^{(e-1)};  \mathcal{D}_{m}^{(e)},\textit{\textbf{c}}^{k,(e)}_{m,t},\textit{\textbf{C}}^{k}_t),$$
  where $\nabla$ denotes derivation operation. 
  The local model parameters $\textit{\textbf{w}}_{m,t}^{(e-1)}$ are then updated as:
\begin{equation}
\textit{\textbf{w}}_{m,t}^{(e)}  = \textit{\textbf{w}}_{m,t}^{(e-1)} - \eta \textit{\textbf{g}}_{m,t}^{(e-1)},\label{925-3}\nonumber
\end{equation}
where $\eta$ is the local learning rate.
\end{itemize}

After completing $E$ local iterations, each client derives  the updated  local model parameters $\textit{\textbf{w}}_{m,t}^{(E)}$.

\begin{algorithm}[t]
\begin{small}
    \caption{$\textbf {Training}$}
    \label{algorithm:test}
    \LinesNumbered
    \KwIn {$\textit{\textbf{w}}_{m,t-1}^{(E)}$, $ \mathcal{D}_m$, $\eta$, $E$, $\{\textit{\textbf{C}}^{k}_{t}\}_{k\in \mathcal{I}}^{|\mathcal{I}|}$.}
    \KwOut {Local model parameters of each client   $\textit{\textbf{w}}_{m,t}^{(E)}$.}
$\textit{\textbf{w}}_{m,t}^{(0)} = \textit{\textbf{w}}_{m,t-1}^{(E)}$;\\
    \For{each {\rm local iteration} $e$ $\in$ $\{1,2,\cdots,E\}$}{
	    	Randomly sample $\mathcal{D}_{m}^{(e)} \subset \mathcal{D}_m$;\\
 \For{each {\rm class} $k$ $\in$ $\mathcal{I}$}{
Calculate $\textit{\textbf{c}}^{k,(e)}_{m,t}$ from $\mathcal{D}^{k,(e)}_m$ with formula (\ref{92611});\\}
		$\textit{\textbf{g}}_{m,t}^{(e-1)} = \nabla\mathcal{L}(\textit{\textbf{w}}_{m,t}^{(e-1)};  \mathcal{D}_{m}^{(e)},\textit{\textbf{c}}^{k,(e)}_{m,t},\textit{\textbf{C}}^{k}_t)$;\\
		$\textit{\textbf{w}}_{m,t}^{(e)} = \textit{\textbf{w}}_{m,t}^{(e-1)} - \eta \textit{\textbf{g}}_{m,t}^{(e-1)}$;\\
    }
    \Return{$\textit{\textbf{w}}_{m,t}^{(E)}$}
    \end{small}
\end{algorithm}

\subsubsection{Prototype  Handling}
The prototype handling stage consists of three key  phases: \textit{prototype generation}, \textit{normalization}, and \textit{encryption},  as illustrated in \textbf{Algorithm 4}.
\paragraph{Prototype Generation}
Since the parameters of   local model change with each  iteration, the  prototypes generated by each local iteration are different.
Consequently, the local prototypes evolve dynamically during local model training.
To submit more representative prototypes, each client regenerates them via the local model parameters \scalebox{0.8}{$\textit{\textbf{w}}_{m,t}^{(E)}$}.
Formally, let $\textit{\textbf{c}}^{k}_{m,t}$ be the submitted local prototype,  which is computed as:
\begin{equation}
\scalebox{1}{$
\begin{aligned} \textit{\textbf{c}}^{k}_{m,t}= \frac{1}{| \mathcal{D}^k_m|}\sum_{(\textit{\textbf{x}}_{(i)},y_{(i)}) \in  \mathcal{D}^k_m} f_m(\textit{\textbf{r}}_{m,t}^{(E)};\textit{\textbf{x}}_{(i)}), \forall k\in \mathcal{I},
\end{aligned}$}\nonumber
\end{equation}
where \scalebox{0.8}{$\textit{\textbf{r}}_{m,t}^{(E)}$} is  the parameters of  feature extractor.
\paragraph{Normalization}
Considering that malicious clients amplify their submitted local prototypes, the local prototype of each class at  each client  is normalized.
Formally, the normalized local prototype is given by:
\begin{equation}
\widetilde{\textit{\textbf{c}}}^{k}_{m,t}= \textit{\textbf{c}}^{k}_{m,t}/ \Vert\textit{\textbf{c}}^{k}_{m,t}\Vert,\forall m\in \mathcal{S}, k\in \mathcal{I},\label{9271}\nonumber
\end{equation}
where $\widetilde{\textit{\textbf{c}}}^{k}_{m,t}$ is a unit vector.
After normalization, each local prototype is encrypted to ensure privacy protection.

\begin{algorithm}[t]
\begin{small}
    \caption{$\textbf {Handling}$}
    \label{algorithm:test}
    \LinesNumbered
    \KwIn {$\textit{\textbf{r}}_{m,t}^{(E)}$, $\mathcal{D}_m$}
    \KwOut {$\{\llbracket\textit{\textbf{c}}_{m,t}^{k}\rrbracket\}_{k\in \mathcal{I}}^{|\mathcal{I}|}$}
    \For{each {\rm class} $k$ $\in$ $ \mathcal{I}$}{
		$\textit{\textbf{c}}^{k}_{m,t}= \frac{1}{| \mathcal{D}^k_m|}\sum_{(\textit{\textbf{x}}_{(i)},y_{(i)}) \in  \mathcal{D}^k_m}^{| \mathcal{D}^k_m|} f_m(\textit{\textbf{r}}_{m,t}^{(E)};\textit{\textbf{x}}_{(i)})$;\\
		$\widetilde{\textit{\textbf{c}}}^{k}_m= \textit{\textbf{c}}^{k}_{m,t}/ \Vert\textit{\textbf{c}}^{k}_{m,t}\Vert$;\\
Encrypt $\widetilde{\textit{\textbf{c}}}^{k}_{m,t}$ by   \textit{Verifier}'s $Pkv$ to get $\llbracket\widetilde{\textit{\textbf{c}}}^{k}_{m,t}\rrbracket_{Pkv}$;\\
    }
    \Return{$\{\llbracket\textit{\textbf{c}}_{m,t}^{k}\rrbracket\}_{k\in \mathcal{I}}^{|\mathcal{I}|}$ }
    \end{small}
\end{algorithm}

\paragraph{Encryption}
To protect privacy of clients,  they  encrypt their own normalized local prototypes \scalebox{0.8}{$\widetilde{\textit{\textbf{c}}}^{k}_{m,t}$} using the \textit{Verifier}'s public key $Pkv$ to get \scalebox{0.8}{$\llbracket\widetilde{\textit{\textbf{c}}}^{k}_{m,t}\rrbracket_{Pkv}$}, and  sends the \scalebox{0.8}{$\{\llbracket\widetilde{\textit{\textbf{c}}}^{k}_{m,t}\rrbracket_{Pkv}\}_{k\in \mathcal{I}}^{|\mathcal{I}|}$} to  \textit{Aggregator}, where  CKKS technique is used to encrypt.

\subsection{Secure Aggregation  Protocol}
To resist  local prototypes submitted by malicious clients without compromising privacy,  we  design a secure  aggregation protocol across two servers to filter  malicious local prototypes and obtain Byzantine-robust global prototypes.
The protocol consists of \textit{normalization  verification} and \textit{secure two-party computation}.

\subsubsection{Normalization  Verification}
Since malicious clients may amplify the impact of their local prototypes, the two servers need to verify that  encrypted local prototypes are normalized.  
Specifically, \textit{Aggregator}  calculates the inner product \scalebox{0.8}{$\llbracket\widetilde{\textit{\textbf{c}}}^k_{m,t}\rrbracket_{Pkv} \cdot \llbracket\widetilde{\textit{\textbf{c}}}^k_{m,t}\rrbracket_{Pkv}$} for each local prototype, and sends the  results to  \textit{Verifier}, where $\cdot$ denotes the  inner product.
Then, the \textit{Verifier} decrypts  \scalebox{0.8}{$\llbracket\widetilde{\textit{\textbf{c}}}^k_{m,t}\rrbracket_{Pkv}\cdot \llbracket\widetilde{\textit{\textbf{c}}}^k_{m,t}\rrbracket_{Pkv}$} using its secret key $Skv$, and  checks whether the  inner product \scalebox{0.8}{$\Vert\widetilde{\textit{\textbf{c}}}^k_{m,t}\Vert^2$} equals 1.
If the inner product of  prototype from client $m$ is not equal to 1, which indicates that its local prototype is not normalized, client $m$ is removed from the set of clients $\mathcal{S}$.
After validation,  \textit{Verifier} sends the remaining  client set   $\mathcal{S}$ to  \textit{Aggregator}.
\subsubsection{Secure Two-party Computation}
To ensure privacy preservation  and  Byzantine-robust aggregation,  two servers perform secure two-party computation to calculate global prototypes.
The computation of the global prototypes is shown in \textbf{Algorithm 5} and described as follows.

\begin{algorithm}[t]
\begin{small}
    \caption{$\textbf { SecComput}$}
    \label{algorithm:test}
    \LinesNumbered
    \KwIn {$\llbracket\widetilde{\textit{\textbf{c}}}^k_{m,t}\rrbracket_{Pkv}$, $ \mathcal{S}^{}$, $\chi$. }
    \KwOut {$\llbracket \textit{\textbf{C}}^k_{t+1} \rrbracket_{Pkx}$.}
    \A{}{Compute  trusted prototype $\llbracket\textit{\textbf{C}}^{\prime k}_{t+1}\rrbracket_{Pkv}$ by formula (\ref{1447});\\
Compute $\llbracket \textit{sim}^{k}_{m,t}\rrbracket_{Pkv}$ by formula (\ref{925-3});\\
Compute $\llbracket \textit{sim}^{\prime k}_{m,t}\rrbracket_{Pkv}$ and $\llbracket\chi^{\prime}\rrbracket_{Pkv}$;\\
$\llbracket h^{k}_{m,t}\rrbracket_{Pkv} \leftarrow \textbf {OutPut}(\llbracket \textit{sim}^{\prime k}_{m,t}\rrbracket_{Pkv}, \llbracket\chi^\prime\rrbracket_{Pkv})$;\\
Randomly select   a $n$-dimensional vector $ \textit{\textbf{V}}^n$ and a number $p$;\\
Send $p \times\llbracket h^{k}_{m,t}\rrbracket_{Pkv}$and $\textit{\textbf{V}}^n \odot \llbracket\widetilde{\textit{\textbf{c}}}^{k}_{m,t}\rrbracket_{Pkv}$ to \textit{Verifier};
}{}



\V{}{Decrypt $p \times\llbracket h^{k}_{m,t}\rrbracket_{Pkv}$and $\textit{\textbf{V}}^n \odot \llbracket\widetilde{\textit{\textbf{c}}}^{k}_{m,t}\rrbracket_{Pkv}$ with $Skv$;\\
Compute $j^{k}_{m,t}$ by formula (\ref{10051});\\
Compute $\textit{Sum}^k_t$;\\
Encrypt $j^{k}_{m,t}$, $\textit{\textbf{V}}^n \odot \widetilde{\textit{\textbf{c}}}^{k}_{m,t}$ using \textit{Clients}' $Pkx$; \\
Send $\textit{Sum}^k_t$, $\llbracket j^{k}_{m,t}\rrbracket_{Pkx}$, and $\llbracket\textit{\textbf{V}}^n \odot \widetilde{\textit{\textbf{c}}}^{k}_{m,t}\rrbracket_{Pkx}$ to \textit{Aggregator}; \\
}{}
\A{}{Compute $\llbracket \widetilde{\textit{\textbf{c}}}^{k}_{m,t}\rrbracket_{Pkx} \leftarrow\frac{1}{\textit{\textbf{V}}^n} \odot\llbracket \textit{\textbf{V}}^n \odot \widetilde{\textit{\textbf{c}}}^{k}_{m,t}\rrbracket_{Pkx}$;\\
Aggregate global prototype $\llbracket \textit{\textbf{C}}^k_{t+1} \rrbracket_{Pkx}$ by formula (\ref{10111});\\
Distribute global prototype $\llbracket \textit{\textbf{C}}^k_{t+1} \rrbracket_{Pkx}$ to clients;\\}{}
\end{small}
\end{algorithm}

Specifically, malicious clients exploit  poisoned local data to generate  local prototypes with low  credibility.
Conversely, local prototypes submitted by benign clients should have  high credibility.
However, for each class without a trusted prototype direction, it is difficult  to assess  the credibility  of  submitted prototypes.
To this end,  \textit{Aggregator} computes the trusted prototype  via the  formula:
\begin{equation}\scalebox{1}{$
\llbracket\textit{\textbf{C}}^{\prime k}_{t+1}\rrbracket_{Pkv} = \frac{1}{|\mathcal{S}|}\sum_{m \in \mathcal{S}} \llbracket \textit{\textbf{c}}_{m,t}^k \rrbracket_{Pkv}, \forall k \in \mathcal{I},  \label{1447}$}
\end{equation}
where \scalebox{0.8}{$\llbracket\textit{\textbf{C}}^{\prime k}_{t+1}\rrbracket_{Pkv}$} denotes  the trusted prototype for class $k$ in  $(t+1)$-th communication round.
\scalebox{0.8}{$\textit{\textbf{C}}^{\prime k}_{t+1}$} is considered a trusted prototype for two reasons: (i) In terms of magnitude, this is because the normalized malicious prototypes do not affect the magnitude of trusted prototype, only the direction of trusted prototype.
(ii) In terms of direction, the average of all prototypes remains a plausible direction due to the small proportion of malicious clients.
Subsequently,  \textit{Aggregator}  obtains the plaintext
\scalebox{0.8}{$\Vert\textit{\textbf{C}}^{\prime k}_{t+1}\Vert$} for subsequent  cosine similarity computation.
 To achieve the plaintext, it computes the inner product \scalebox{0.8}{$\llbracket\textit{\textbf{C}}^{\prime k}_{t+1}\rrbracket_{Pkv} \cdot \llbracket\textit{\textbf{C}}^{\prime k}_{t+1}\rrbracket_{Pkv}$} and sends it to the \textit{Verifier} for decryption, then receives the decrypted result.
Then, cosine similarity, a widely used  metric for measuring the angle between two vectors, is employed to evaluate the credibility of  local prototype.
If the direction of  local prototype is similar to that of  trusted prototype, it is assigned higher credibility.
Formally,  \textit{Aggregator} computes the credibility of  local prototype  by  the following formula:
\begin{equation}
\scalebox{0.8}{$
\begin{aligned}
\llbracket \textit{sim}^{k}_{m,t}\rrbracket_{Pkv} &=\llbracket\textit{sim}(\widetilde{\textit{\textbf{c}}}^{k}_{m,t},\textit{\textbf{C}}^{\prime k}_{t+1})\rrbracket_{Pkv}= \left\llbracket\frac{\widetilde{\textit{\textbf{c}}}^{k}_{m,t} \cdot \textit{\textbf{C}}^{\prime k}_{t+1}}{\Vert\widetilde{\textit{\textbf{c}}}^{k}_{m,t}\Vert \Vert\textit{\textbf{C}}^{\prime k}_{t+1}\Vert}\right\rrbracket_{Pkv} \\ &=\frac{1}{\Vert\textit{\textbf{C}}^{\prime k}_{t+1}\Vert}\left(\llbracket\widetilde{\textit{\textbf{c}}}^{k}_{m,t}\rrbracket_{Pkv}\cdot\llbracket\textit{\textbf{C}}^{\prime k}_{t+1}\rrbracket_{Pkv}\right),\label{925-3}
\end{aligned}$}
\end{equation}
where $\llbracket \textit{sim}^{k}_{m,t}\rrbracket_{Pkv}$ denotes the credibility of local prototype.
Additionally, we define a detection threshold  $\chi$ and classify  local prototypes with credibility less than $\chi$ as  anomalous local prototypes.
During   global prototype aggregation,  we set the aggregation weight of  anomalous local prototype to 0.
Formally, the aggregation  weight of each local prototype   is calculated as follows:
\begin{equation}
j^{k}_{m,t} =
\begin{cases}
0.& \textit{sim}^{k}_{m,t}<\chi\\
\textit{sim}^{k}_{m,t}.& \textit{sim}^{k}_{m,t}>\chi,
\end{cases}\label{10041}
\end{equation}
where $j^{k}_{m,t}$ denotes the  aggregation weight of local prototype for class $k$ at client $m$ in the  $t$-th communication round.
However,  \textit{Aggregator}  cannot directly compare \scalebox{0.8}{$\llbracket \textit{sim}^{k}_{m,t}\rrbracket_{Pkv}$} with  $\chi$ within the  encryption domain.

To implement  comparison  within  encryption domain, the  work in  \cite{cheon2019numerical} offers a ciphertext comparison method that  outputs  the maximum value of  homomorphic ciphertexts corresponding to two plaintexts in the range (0, 1) without decryption.
However, since   cosine similarity ranges from -1 to 1,   the method cannot be applied directly.
To resolve the range mismatch, \textit{Aggregator} sets \scalebox{0.8}{$\llbracket \textit{sim}^{\prime k}_{m,t}\rrbracket = \frac{1}{2}(\llbracket \textit{sim}^{k}_{m,t}\rrbracket+\llbracket1\rrbracket)$} and \scalebox{0.8}{$\chi^{\prime} = \frac{1}{2}(\chi+1)$} as  input  to the comparison method \cite{cheon2019numerical}. This is because the size relationship between $\llbracket \textit{sim}^{ k}_{m,t}\rrbracket$ and $\chi^{}$  remains unchanged.
Then,   \textit{Aggregator}  employs \textbf{Algorithm 6}  to   calculate  the maximum value $\llbracket h^{k}_{m,t}\rrbracket$ between $\llbracket \textit{sim}^{\prime k}_{m,t}\rrbracket$ and $\llbracket\chi^{\prime}\rrbracket$.
Notably, in the \textbf{Algorithm 6},
$a$ and $b$ are defined as $a = \llbracket \textit{sim}^{\prime k}_{m,t}\rrbracket$ and $b = \llbracket \chi^{\prime}\rrbracket$. From these definitions, $q_1 = \frac{a+b}{2}$ and $q_2 = \frac{a-b}{2}$ are directly computed.
 The parameter $d$ in \textbf{Algorithm 6} is a hyper-parameter. The larger the value of $d$, the more accurate the output. Based on previous work \cite{cheon2019numerical}, setting $d=29$ is sufficient to ensure that the output matches the top-16 bits coincide with those of the true maximum value.
In addition, the design principle process of \textbf{Algorithm 6} is detailed in works \cite{cheon2019numerical}\cite{wilkes1951preparation}.


\begin{algorithm}[t]
\begin{small}
    \caption{$\textbf {OutPut}$}
    \label{algorithm:test}
    \LinesNumbered
    \KwIn {$(\llbracket \textit{sim}^{\prime k}_{m,t}\rrbracket, \llbracket\chi^{\prime}\rrbracket ) \in (\llbracket0\rrbracket,\llbracket1\rrbracket), d \in \mathbb{N}$ }
    \KwOut {an max value of $\llbracket \textit{sim}^{\prime k}_{m,t}\rrbracket$ or $\llbracket\chi^{\prime}\rrbracket$}
Initialize $a = \llbracket \textit{sim}^{\prime k}_{m,t}\rrbracket, b =\llbracket\chi^{\prime}\rrbracket$;\\
 $q_1 = \frac{(a+b)}{2}, q_2 = \frac{(a-b)}{2}$;\\
$a_0 = q_2^2, b_0 = q_2^2-1$;\\
 \For{each $n$ $\in$ $(0,d-1)$}{
		$a_{n+1} = a_{n} (1-\frac{b_n}{2})$;\\
$b_{n+1} = b_{n}^2  (\frac{b_n-3}{4})$;\\
    }
$q_3 = a_d$;\\
  \Return{$ (q_1+q_3) $}
\end{small}
\end{algorithm}

After obtaining  the maximum value $\llbracket h^{k}_{m,t}\rrbracket_{Pkv}$, \textit{Aggregator} collaborates with  \textit{Verifier} to calculate  the aggregation results.
Specifically,  \textit{Aggregator} selects  a random value $p$,  and multiplies it by the maximum value   to get  \scalebox{0.8}{$p \times\llbracket h^{k}_{m,t}\rrbracket_{Pkv}$}.
Additionally,   \textit{Aggregator}   chooses  a random vector $\textit{\textbf{V}}$ with the same dimension as the prototype, and computes its hadamard product with  local prototype to obtain \scalebox{0.8}{$\textit{\textbf{V}} \odot \llbracket\widetilde{\textit{\textbf{c}}}^{k}_{m,t}\rrbracket_{Pkv}$}, where $\odot$ denotes the hadamard product.
Then,  \textit{Aggregator} sends \scalebox{0.8}{$p \times\llbracket h^{k}_{m,t}\rrbracket_{Pkv}$} and \scalebox{0.8}{$\textit{\textbf{V}} \odot \llbracket\widetilde{\textit{\textbf{c}}}^{k}_{m,t}\rrbracket_{Pkv}$} to  \textit{Verifier}.
The \textit{Verifier} then   decrypts these values  by its secret key $Skv$ to get \scalebox{0.8}{$p \times h^{k}_{m,t}$} and \scalebox{0.8}{$\textit{\textbf{V}} \odot \widetilde{\textit{\textbf{c}}}^{k}_{m,t}$}.
Since the random value $p$ and  random vector $\textit{\textbf{V}}$ obfuscate $h^{k}_{m,t}$ and \scalebox{0.8}{$\llbracket\widetilde{\textit{\textbf{c}}}^{k}_{m,t}\rrbracket_{Pkv}$}, respectively,  \textit{Verifier} cannot extract sensitive information about clients from $p \times h^{k}_{m,t}$ and \scalebox{0.8}{$\textit{\textbf{V}} \odot \widetilde{\textit{\textbf{c}}}^{k}_{m,t}$}.
Then, the aggregation weight is rewritten  as  the following formula:
\begin{equation}
 j^{k}_{m,t}=
\begin{cases}
0.&    \textit{Round}(p\times h^{k}_{m,t},6)=\textit{min}^k_t \\
p\times h^{k}_{m,t}&    \textit{Round}(p\times h^{k}_{m,t},6)>\textit{min}^k_t,\label{10051}
\end{cases}
\end{equation}
where $\textit{Round}(p\times h^{k}_{m,t},6)$ denotes that $p\times h^{k}_{m,t}$ is rounded to the 6-th  decimal place.
This is because CKKS   decrypts the ciphertext with an error in the range of $10^{-7}$ \cite{kim2022approximate}.
Moreover, $\textit{min}^k_t$ denotes the smallest value for class $k$ in set \scalebox{0.8}{$\{\textit{Round}(p \times h^{k}_{m,t},6)\}_{m \in \mathcal{S^{}}}^{|\mathcal{S^{}}|}$}.
The condition $\textit{Round}(p\times h^{k}_{m,t},6)=\textit{min}^k_t$ means that the local prototype \scalebox{0.8}{$\textit{\textbf{c}}^{k}_{m,t}$} satisfies $\textit{sim}^{k}_{m,t}<\chi$, thus its aggregation weight $j^{k}_{m,t}$ is 0.
Additionally, \textit{Verifier} calculates \scalebox{0.8}{$\textit{Sum}^k_t=\sum_{m \in \mathcal{S}} j^{k}_{m,t}$} and encrypts $j^{k}_{m,t}$ and \scalebox{0.8}{$\textit{\textbf{V}} \odot \widetilde{\textit{\textbf{c}}}^{k}_{m,t}$} using the client's public key $Pkx$ to get \scalebox{0.8}{$\llbracket j^{k}_{m,t}\rrbracket_{Pkx}$} and \scalebox{0.8}{$\llbracket\textit{\textbf{V}} \odot \widetilde{\textit{\textbf{c}}}^{k}_{m,t}\rrbracket_{Pkx}$}, and sends them to  \textit{Aggregator}.


The \textit{Aggregator} computes \scalebox{0.8}{$\frac{1}{\textit{\textbf{V}}} \odot\llbracket \textit{\textbf{V}} \odot \widetilde{\textit{\textbf{c}}}^{k}_{m,t}\rrbracket_{Pkx}$} to get \scalebox{0.8}{$\llbracket \widetilde{\textit{\textbf{c}}}^{k}_{m,t}\rrbracket_{Pkx}$}, and   aggregates the encrypted local prototypes to get the  encrypted global prototype $\llbracket \textit{\textbf{C}}^k_{t+1} \rrbracket_{Pkx}$ by the following formula:
\begin{equation}
\begin{aligned}
\llbracket \textit{\textbf{C}}^k_{t+1} \rrbracket_{Pkx} &= \frac{1}{\textit{Sum}^k_t}\sum_{m\in  \mathcal{S}} \llbracket j^{k}_{m,t}\rrbracket_{Pkx} \times \llbracket \widetilde{\textit{\textbf{c}}}^{k}_{m,t}\rrbracket_{Pkx}.\label{10111}
\end{aligned}
\end{equation}
Subsequently, the \textit{Aggregator}  distributes the encrypted global prototype to \textit{Clients}.
 After the above process is completed,  FL executes the next communication round until a predefined number of  rounds is reached.

\section{Analysis}
\label{sec:Aggregation and security analysis}
In the section, we provide both convergence analysis and privacy  analysis for PPFPL.
Specifically, we make the following assumptions similar to existing general frameworks  \cite{tan2022fedproto}\cite{wang2020tackling} for the local optimization function (\ref{925-4}).

\begin{assumption}
Each  loss function is $L_1$ Lipschitz smooth, which means that the gradient of  loss function is $L_1$ Lipschitz continuous, we can get
\begin{equation}
\scalebox{0.8}{$
\begin{gathered}
\left\|\nabla \mathcal{L}_{m,t}^{(e_1)}-\nabla \mathcal{L}_{m,t}^{(e_2)}\right\|_2 \leq L_1\left\|\textit{\textbf{w}}_{m,t}^{(e_1)}-\textit{\textbf{w}}_{m,t}^{(e_2)}\right\|_2,\nonumber
\end{gathered}$}
\end{equation}
where  $\mathcal{L}_{m,t}^{(e_1)}$ denotes  loss function at the $(tE+e_1)$-th local iteration in client $m$. This  implies the  quadratic bound:
\begin{equation}
\scalebox{0.8}{$
\begin{aligned}
\mathcal{L}_{m,t}^{(e_1)}&-\mathcal{L}_{m,t}^{(e_2)} \leq\left\langle\nabla \mathcal{L}_{m,t}^{(e_2)},\left(\textit{\textbf{w}}_{m,t}^{(e_1)}-\textit{\textbf{w}}_{m,t}^{(e_2)}\right)\right\rangle +\frac{L_1}{2}\left\|\textit{\textbf{w}}_{m,t}^{(e_1)}-\textit{\textbf{w}}_{m,t}^{(e_2)}\right\|_2^2. \nonumber
\end{aligned}$}
\end{equation}
\end{assumption}

\begin{assumption}
The stochastic gradient \scalebox{0.8}{$\textit{\textbf{g}}_{m, t}^{(e)}=\nabla \mathcal{L}\left(\textit{\textbf{w}}_{m,t}^{(e-1)}; \mathcal{D}_m^{(e)}\right)$} is an unbiased estimator of the local gradient for each client.
Suppose its expectation
\begin{equation}\scalebox{0.8}{$
\mathbb{E}_{\mathcal{D}_m^{(e)} \sim
 \mathcal{D}_m}\left[\textit{\textbf{g}}_{m, t}^{(e)}\right]=\nabla \mathcal{L}\left(\textit{\textbf{w}}_{m, t}^{(e)};  \mathcal{D}_m^{(e)}\right)=\nabla \mathcal{L}_{m,t}^{(e)}, \nonumber$}
\end{equation}
and its variance is bounded by $\sigma^2$:
\scalebox{0.8}{$
\mathbb{E}[\|\textit{\textbf{g}}_{m,t}^{(e)}-\nabla\mathcal{L}(\textit{\textbf{w}}_{m,t}^{(e)})\|_2^2]\leq\sigma^2$}. 
\end{assumption}

Based on the above assumptions, we formulate the following theorem  and corollaries.
Notably, we add ``\scalebox{0.8}{$\frac{1}{2}$}" into the  local iteration, denoted as \scalebox{0.8}{$\{\frac{1}{2},1,\cdots,E\}$} in our analysis.
For example,  $tE$ denotes the time step before local prototype aggregation, and \scalebox{0.8}{$tE + \frac{1}{2}$} denotes the time step between local prototype aggregation and the first local iteration in the $t$-th round.

\begin{thm} \label{14785414}
  In PPFPL,  under any proportion of malicious clients, for the $t$-th communication round, the variation of  loss function for each benign client  can be bounded as:
\begin{equation}\scalebox{1}{$
\begin{aligned}
\mathbb{E}&\left[\mathcal{L}_{m,t+1}^{\frac{1}{2}}\right] -\mathcal{L}_{m,t}^E \leq G(\lambda, \eta, E),\nonumber
\end{aligned}$}
\end{equation}
where 
\scalebox{0.7}{$
\begin{aligned}
 G(\lambda, \eta, E)= -\left(\eta-\frac{ \eta^2L_1}{2}\right)\sum_{e=\frac{1}{2}}^{E}\left\Vert\nabla \mathcal{L}_{m,t}^{(e)}\right\Vert_2^2+\frac{ E \eta^2L_1}{2} \sigma^2 +  2\lambda.\end{aligned}$}
\end{thm}
We can observe that each variable in this upper bound is independent of the malicious clients, indicating  that the impact of malicious clients is  constrained. 
This constraint arises from incorporating  the cosine similarity  in the auxiliary term of our local optimization function, which restricts the update direction of the local model.
\begin{proof}
Assuming that \textbf{Assumption 1} holds, we can get
\begin{equation}
\scalebox{0.8}{$
\begin{aligned}
\mathcal{L}_{m,t}^{(1)} &\leq \mathcal{L}_{m,t}^{(\frac{1}{2})}+\left \langle\nabla \mathcal{L}_{m,t}^{(\frac{1}{2})},\left(\textit{\textbf{w}}_{m,t}^{(1)}-\textit{\textbf{w}}_{m,t}^{(\frac{1}{2})}\right)\right\rangle+\frac{L_1}{2}\left\|\textit{\textbf{w}}_{m,t}^{(1)}-\textit{\textbf{w}}_{m,t}^{(\frac{1}{2})}\right\|_2^2\\ & \overset{\text{(a)}}{=} \mathcal{L}_{m,t}^{(\frac{1}{2})} - \eta\left\langle\nabla \mathcal{L}_{m,t}^{(\frac{1}{2})},\textit{\textbf{g}}_{m,t}^{(\frac{1}{2})}\right\rangle+\frac{\eta^2 L_1}{2}\left\|\textit{\textbf{g}}_{m,t}^{(\frac{1}{2})}\right\|_2^2,\\\label{102218311}
\end{aligned}$}
\end{equation}
where (a) follows from \scalebox{0.7}{$\textit{\textbf{w}}_{m,t}^{(1)}=\textit{\textbf{w}}_{m,t}^{(\frac{1}{2})} -\eta\textit{\textbf{g}}_{m,t}^{(\frac{1}{2})}$}. Taking expectation on both sides of formula (\ref{102218311}), we can get
\begin{equation}
\scalebox{0.7}{$
\begin{aligned}
&\mathbb{E}[\mathcal{L}_{m,t}^{(1)}] \leq\mathcal{L}_{m,t}^{(\frac{1}{2})}- \eta\mathbb{E}\left[\left\langle\nabla \mathcal{L}_{m,t}^{(\frac{1}{2})},\textit{\textbf{g}}_{m,t}^{(\frac{1}{2})}\right\rangle\right]+\frac{\eta^2 L_1}{2}\mathbb{E}\left[\left\|\textit{\textbf{g}}_{m,t}^{(\frac{1}{2})}\right\|_2^2\right]\\
&\overset{\text{(b)}}{=}\mathcal{L}_{m,t}^{(\frac{1}{2})}-\eta \Vert\nabla \mathcal{L}_{m,t}^{(\frac{1}{2})}\Vert_2^2+\frac{\eta^2 L_1}{2}\mathbb{E}\left[\left\|\textit{\textbf{g}}_{m,t}^{(\frac{1}{2})}\right\|_2^2\right]\\
&\overset{\text{(c)}}{=}\mathcal{L}_{m,t}^{(\frac{1}{2})}-\eta \Vert\nabla \mathcal{L}_{m,t}^{(\frac{1}{2})}\Vert_2^2+\frac{\eta^2 L_1}{2}\left(\left\|\mathbb{E}[\textit{\textbf{g}}_{m,t}^{(\frac{1}{2})}]\right\|^2_2+Var(\textit{\textbf{g}}_{m,t}^{(\frac{1}{2})})\right)\\
&\overset{\text{(d)}}{\leq}\mathcal{L}_{m,t}^{(\frac{1}{2})}-\eta \Vert\nabla \mathcal{L}_{m,t}^{(\frac{1}{2})}\Vert_2^2+\frac{\eta^2 L_1}{2}\left(\left\|\nabla\mathcal{L}_{m,t}^{(\frac{1}{2})}\right\|_2^2+Var(\textit{\textbf{g}}_{m,t}^{(\frac{1}{2})})\right)\\
&=\mathcal{L}_{m,t}^{(\frac{1}{2})}-(\eta-\frac{\eta^2 L_1}{2}) \Vert\nabla \mathcal{L}_{m,t}^{(\frac{1}{2})}\Vert_2^2+\frac{\eta^2 L_1}{2}Var(\textit{\textbf{g}}_{m,t}^{(\frac{1}{2})})\\
&\overset{\text{(e)}}{\leq}\mathcal{L}_{m,t}^{(\frac{1}{2})}-(\eta-\frac{\eta^2 L_1}{2}) \Vert\nabla \mathcal{L}_{m,t}^{(\frac{1}{2})}\Vert_2^2+\frac{\eta^2 L_1}{2}\sigma^2,\\
\label{10221832}\nonumber
\end{aligned}$}
\end{equation}
where (b), (d) and (e)  are derived from \textbf{Assumption 2}, (c) follows from \scalebox{0.8}{$Var(x) = \mathbb{E}[x^2]-(\mathbb{E}[x])^2$}.
During the local computation step, the loss function is updated  $E$ times, and it can be  bounded as:
\begin{equation}
\scalebox{0.8}{$
\mathbb{E}[\mathcal{L}_{m,t}^{(1)}] \leq \mathcal{L}_{m,t}^{(\frac{1}{2})}-(\eta-\frac{\eta^2 L_1}{2}) \sum_{e=\frac{1}{2}}^{E}\left\Vert\nabla \mathcal{L}_{m,t}^{(e)}\right\Vert_2^2+\frac{E\eta^2 L_1}{2}\sigma^2.\\
\label{10221831}$}
\end{equation}
Additionally, since a single communication round involves both local computation and secure aggregation,  we  need to compute the impact of the aggregation result for   loss function of each benign client.
Specifically, the  loss function of each benign client at the $((t + 1)E+\frac{1}{2})$ time step is represented as follows:
\begin{equation}
\scalebox{0.8}{$
\begin{aligned}
&\mathcal{L}_{m,t+1}^{\frac{1}{2}}  =\mathcal{L}_{m,t}^{E}+\mathcal{L}_{m,t+1}^{\frac{1}{2}}-\mathcal{L}_{m,t}^{E} \\
& \stackrel{}{=} \mathcal{L}_{m,t}^{E}+ \lambda \mathcal{L_R}(\textit{\textbf{c}}^k_{m,t+1},\textit{\textbf{C}}^k_{t+2})-\lambda \mathcal{L_R}(\textit{\textbf{c}}^k_{m,t+1},\textit{\textbf{C}}^k_{t+1})\\
& \overset{}{=} \mathcal{L}_{m,t}^{E}- \frac{\lambda}{|\mathcal{I}|}\sum_{k\in\mathcal{I}} \textit{sim}(\textit{\textbf{c}}_{m,t+1}^k,\textit{\textbf{C}}^k_{t+2})+\frac{\lambda}{|\mathcal{I}|}\sum_{k\in\mathcal{I}} \textit{sim}(\textit{\textbf{c}}_{m,t+1}^k,\textit{\textbf{C}}^k_{t+1})\\
&\overset{\text{(f)}}{\leq} \mathcal{L}_{m,t}^{E}+ 2\lambda , \label{10222250}
\end{aligned}$}
\end{equation}
where (f) follows from  \scalebox{0.8}{$-1 \leq \textit{sim}(\cdot,\cdot)\leq 1$}.
Notably, although different proportions of malicious clients can influence the direction of $\textit{\textbf{C}}^k_{t+1}$, the contribution of malicious prototypes is inherently bounded by the cosine similarity  $\textit{sim}(\cdot,\cdot)$.
In other words, regardless of the proportion of the malicious clients, PPFPL suppresses the impact of malicious clients, thereby maintaining the robustness of the aggregated results.

 Taking expectation on both sides of formula (\ref{10222250}), we can get:
 \begin{equation}\scalebox{0.8}{$
\begin{aligned}
\mathbb{E}\left[\mathcal{L}_{m,t+1}^{\frac{1}{2}}\right]
&\leq \mathcal{L}_{m,t}^{E}+ 2\lambda. \label{10222300}
\end{aligned}$}
\end{equation}
Thus,  during the $t$-th communication  round, according to the formula (\ref{10221831}) and formula (\ref{10222300}),  the variation of  loss function for each benign client  can be bounded as:
\begin{equation}\scalebox{1}{$
\begin{aligned}
\mathbb{E}&\left[\mathcal{L}_{m,t+1}^{\frac{1}{2}}\right] -\mathcal{L}_{m,t}^E \leq G(\lambda, \eta, E), \nonumber
\end{aligned}$}
\end{equation}
where  
\scalebox{0.75}{$
G(\lambda, \eta, E)= -\left(\eta-\frac{ \eta^2L_1}{2}\right)\sum_{e=\frac{1}{2}}^{E}\left\Vert\nabla \mathcal{L}_{m,t}^{(e)}\right\Vert_2^2+\frac{ E \eta^2L_1}{2} \sigma^2 +  2\lambda$}.
Thus, \textbf{Theorem 1} is proved.
\end{proof}

\begin{corollary}
Given any fixed $\lambda$ and $E$, the $G(\eta)$ is  convex   with respect to $\eta$.
\end{corollary}
\begin{proof}
To prove that $G(\eta)$ is   convex   for $\eta$, we need to prove that the second order  derivative of $G(\eta)$ with respect to $\eta$ is always nonnegative.
Thus, we have
\begin{equation}
\scalebox{0.8}{$
\begin{aligned}\frac{dG(\eta)}{d \eta} = -\left(1-L_1 \eta\right) \sum_{e=\frac{1}{2}}^{E}\left\Vert\nabla \mathcal{L}_{m,t}^{(e)}\right\Vert_2^2+L_1 E \eta \sigma^2,\end{aligned}$}\nonumber
\end{equation}
and
\begin{equation}
\scalebox{0.7}{$
\begin{aligned}
\frac{d^2G(\eta)}{d \eta^2} = L_1 \sum_{e=\frac{1}{2}}^{E}\left\Vert\nabla \mathcal{L}_{m,t}^{(e)}\right\Vert_2^2+L_1 E  \sigma^2.
\end{aligned}$}\nonumber
\end{equation}
Since $L_1$, $E$, and $\sigma^2$   are all greater than 0, we have \scalebox{0.8}{$\frac{d^2G(\eta)}{d \eta^2} \textgreater 0.$}
Thus, $G(\eta)$ is proved to be a convex function and there exists a minimum value of $G(\eta)$.
\end{proof}

\begin{corollary}
Given any fixed $\lambda$ and $E$, the variation of the  loss function  compared to the previous round exists a minimum bound when $\eta = \eta^ \star $, where \scalebox{0.8}{$\eta^ \star = \frac{ \sum_{e=\frac{1}{2}}^{E}\left\Vert\nabla \mathcal{L}_{m,t}^{(e)}\right\Vert_2^2}{L_1E\sigma^2+L_1 \sum_{e=\frac{1}{2}}^{E}\left\Vert\nabla \mathcal{L}_{m,t}^{(e)}\right\Vert_2^2}$.}
\end{corollary}
\begin{proof}
When  $G^\prime(\eta)$ equals 0, then $G(\eta)$ obtains a extremum value. Let \scalebox{0.8}{$(\frac{dG(\eta)}{d \eta}|\eta = \eta^ \star) = 0$},   we can get:
\begin{equation}\scalebox{0.85}{$
(\frac{dG(\eta)}{d \eta}|\eta = \eta^ \star)= -\left(1-L_1 \eta\right) \sum_{e=\frac{1}{2}}^{E}\left\Vert\nabla \mathcal{L}_{m,t}^{(e)}\right\Vert_2^2+L_1 E \eta \sigma^2  = 0.$} \nonumber
\end{equation}
Thus, $G(\eta)$ exists  a extremum value  when \scalebox{0.8}{$\eta = \eta^ \star = \frac{ \sum_{e=\frac{1}{2}}^{E}\left\Vert\nabla \mathcal{L}_{m,t}^{(e)}\right\Vert_2^2}{L_1E\sigma^2+L_1 \sum_{e=\frac{1}{2}}^{E}\left\Vert\nabla \mathcal{L}_{m,t}^{(e)}\right\Vert_2^2}$}.
Because corollary 1 proves  the second order  derivative of $G(\eta)$ with respect to $\eta$ is always nonnegative, the extremum value is the minimal value.
Therefore, we can understand that it  exists a minimum bound when $\eta = \eta^ \star $.
\end{proof}

\begin{corollary}
Given any fixed $\eta$ and $E$, the loss function of  arbitrary client monotonously decreases  in each communication round  when  \scalebox{0.8}{ $\lambda \textless \left(\frac{1}{2}\eta-\frac{L_1 \eta^2}{4}\right) \sum_{e=\frac{1}{2}}^{E}\left\Vert\nabla \mathcal{L}_{m,t}^{(e)}\right\Vert_2^2-\frac{E \eta^2L_1}{4} \sigma^2.$}
\end{corollary}
\begin{proof}
To guarantee that the local loss function decreases after each communication, it is necessary to make sure that $G(\lambda) \textless 0$.  We have:
\begin{equation}
\scalebox{0.8}{$
\begin{aligned}
 -\left(\eta-\frac{L_1 \eta^2}{2}\right) \sum_{e=\frac{1}{2}}^{E}\left\Vert\nabla \mathcal{L}_{m,t}^{(e)}\right\Vert_2^2+\frac{E \eta^2L_1}{2} \sigma^2 +  2\lambda \textless 0.
\end{aligned}$}\nonumber
\end{equation}
After simplification, we can get:
\begin{equation}
\scalebox{0.8}{$
\begin{aligned}\lambda \textless \left(\frac{1}{2}\eta-\frac{L_1 \eta^2}{4}\right) \sum_{e=\frac{1}{2}}^{E}\left\Vert\nabla \mathcal{L}_{m,t}^{(e)}\right\Vert_2^2-\frac{E \eta^2L_1}{4} \sigma^2.\end{aligned}$} \nonumber
\end{equation}
\end{proof}

\begin{thm} \label{992139}
Let the proportions of benign and malicious clients be $(1-\kappa)$ and $\kappa$, respectively, with $\kappa \in [0,1]$.  
In the $t$-th round, let $\textit{\textbf{u}}_{ben}^k$ denote the average prototype of benign clients for class $k$, and $\textit{\textbf{u}}_{mal}^k$ denote the average prototype of malicious clients for class $k$. Then, the aggregated global prototype for class $k$ is given by:
$$\textit{\textbf{C}}^k_{t+1} = (1-\kappa) \textit{\textbf{u}}_{ben}^k + \kappa \textit{\textbf{u}}_{mal}^k
$$
and its deviation  in magnitude from the benign prototypes satisfies:
$$
\Vert \textit{\textbf{C}}^k_{t+1} -  \textit{\textbf{u}}_{ben}^k\Vert_2 = \kappa \Vert\textit{\textbf{u}}_{mal}^k- \textit{\textbf{u}}_{ben}^k \Vert_2 \leq 2 \kappa.
$$
\end{thm}
From the above formula, we can observe that the deviation of the global prototype from the benign prototype increases   linearly with the malicious proportion $\kappa$, thus ensuring that the influence of malicious clients always remains within a controllable range.
\begin{proof}
By definition of the malicious client  proportion $\kappa$,  the global prototype can be expressed as:
\begin{equation}
\begin{aligned}\scalebox{0.9}{$
\textit{\textbf{C}}^k_{t+1} = (1-\kappa) \textit{\textbf{u}}_{ben}^k + \kappa \textit{\textbf{u}}_{mal}^k.$}
\end{aligned} \nonumber
\end{equation}
Subtracting $\textit{\textbf{u}}_{ben}^k$ from both sides yields:
\begin{equation}
\scalebox{1}{$
\begin{aligned}
\textit{\textbf{C}}^k_{t+1}  - \textit{\textbf{u}}_{ben}^k = \kappa(\textit{\textbf{u}}_{mal}^k - \textit{\textbf{u}}_{ben}^k).
\end{aligned}$}\nonumber
\end{equation}
Taking the $\ell_2$-norm on both sides, we have:
\begin{equation}
\scalebox{1}{$
\begin{aligned}
\Vert\textit{\textbf{C}}^k_{t+1}  - \textit{\textbf{u}}_{ben}^k\Vert_2 = \kappa\Vert(\textit{\textbf{u}}_{mal}^k - \textit{\textbf{u}}_{ben}^k)\Vert_2.
\end{aligned}$}\nonumber
\end{equation}
Since $\textit{\textbf{u}}_{mal}^k $ and $\textit{\textbf{u}}_{ben}^k$ are unit vectors, the triangle inequality gives:
\begin{equation}
\scalebox{1}{$
\begin{aligned}
\Vert\textit{\textbf{u}}_{mal}^k -  \textit{\textbf{u}}_{ben}^k\Vert_2 \leq \Vert \textit{\textbf{u}}_{mal}^k \Vert + \Vert \textit{\textbf{u}}_{ben}^k \Vert = 2.
\end{aligned}$}\nonumber
\end{equation}
Consequently, we obtain the relation
$ 
\Vert\textit{\textbf{C}}^k_{t+1}  -  \textit{\textbf{u}}_{ben}^k \Vert_2 \leq 2\kappa.\
$
This completes the proof of Theorem \ref{992139}.
\end{proof}

\begin{thm}
The \textit{Aggregator}, \textit{Verifier}, and malicious clients cannot access any sensitive information about   benign  clients.
\end{thm}

\begin{proof}
During the secure aggregation, the two servers can obtain plaintext values $\Vert\textit{\textbf{C}}^{\prime k}_{t+1}\Vert$ and $\textit{Sum}^{k}_{t}$, where $\Vert\textit{\textbf{C}}^{\prime k}_{t+1}\Vert$ denotes the norm  of  trusted prototype, $\textit{Sum}^{k}_{t}$ denotes  the sum of aggregation weight  of class $k$.
For non-colluding \textit{Aggregator} and \textit{Verifier}, they cannot  get any sensitive information  from these   plaintext information.
In addition, for colluding malicious clients, if $(|\mathcal{S}|-1)$ clients are compromised, they can theoretically infer the local prototype  of the remaining benign from the encrypted global prototype. However, the real scenario does not exist when there are $(|\mathcal{S}|-1)$ malicious clients.
Since benign clients receive only global prototypes distributed by \textit{Aggregator}, they cannot deduce  information about others.
Therefore, neither third-party entities nor malicious clients can derive sensitive information about benign clients.
\end{proof}

\section{Experiments}
\label{sec:Experiments}
 In this section, we evaluate the performance of PPFPL in the presence of data poisoning attacks on Non-IID data.
\subsection{Experimental Settings}
\subsubsection{Datasets and Models}
Similar to previous works \cite{9849010}\cite{mu2024feddmc}\cite{10491118}, we utilize  three public available datasets, namely MNIST \cite{lecun2002gradient}, FMNIST \cite{xiao2017fashion}, and CIFAR10 \cite{krizhevsky2009learning}, to evaluate performance of our PPFPL.
Furthermore, we apply Convolution Neural Network (CNN)  as local model to both MNIST and FMNIST.
For CIFAR10, we employ ResNet18  as the local model,  initialized with pre-trained parameters. These initial parameters have an accuracy of 27.5\% on  CIFAR10's test set.

\subsubsection{Hyper-parameters Settings of FL}
We employ a cross-silo configuration in our FL experiments.  
Specifically, we set up 20 clients, each of which uploads  local prototypes at each communication round. The number of  rounds is set to 100, 150, and 150 for MNIST, FMNIST, and CIFAR10, respectively.
The local learning rate is set to $\eta = 0.01$, the importance weight to $\lambda = 1$, the batch size to 64, and the number of local iterations to $E = 5$. The detection threshold $\chi$ is configured as 0.
Except for special declared, the above default settings are applied.
These  hyper-parameters are consistent across all clients.
For encryption, we instantiate the CKKS in our experiments using TenSEAL \cite{benaissa2021tenseal} with the polynomial modulus degree of 8192 and a coefficient modulus  of  200-bits, and a  scale of 40-bits.
This configuration offers  128-bits security level \cite{10609598}\cite{rahman2024benchmarking}, and all encryption parameters follow the recommended settings provided by TenSEAL.

\subsubsection{Non-IID  Settings}
To simulate the Non-IID data in cross-silo FL, we create  class-space heterogeneity among clients, which is common in   cross-silo scenarios.
Specifically,  large organizations (e.g., hospitals, companies)  hold  datasets with different classes, and their class distributions may differ significantly, or  even be missing some classes altogether.
When these organizations participate in federated training, the union of their data classes defines  the entire FL classification task.
This phenomenon leads to Non-IID data across  organizations.
However, the  Dirichlet distribution \cite{lin2016dirichlet} assumes that each client's data is sampled from all classes, meaning  each client  usually has  samples of every class.
However, in  cross-silo  scenarios,  some clients may have no samples from certain classes,  which  contradicts  the assumptions of  Dirichlet distribution.

To model the data distribution in cross-silo scenarios, similar to previous works \cite{tan2022fedproto}\cite{10398600}, we define \textit{Avg} as the mean number of data classes per client, and \textit{Std} as the  standard deviation of these class counts.
In our experiments, we fix \textit{Avg} to be 3 , 4 or 5, and fix \textit{Std} to be 1 or 2,  to create the  class-space heterogeneity.
Clients are randomly assigned to different classes, with partial class overlap among them.
 To visualize the different  data distributions, we plot heat maps  as shown in Fig. \ref{110811}.
\begin{figure}[t]
  \subfloat[\textit{Avg}=3, \textit{Std}=2]
  {
      \label{102411}  \includegraphics[width=0.32\linewidth]{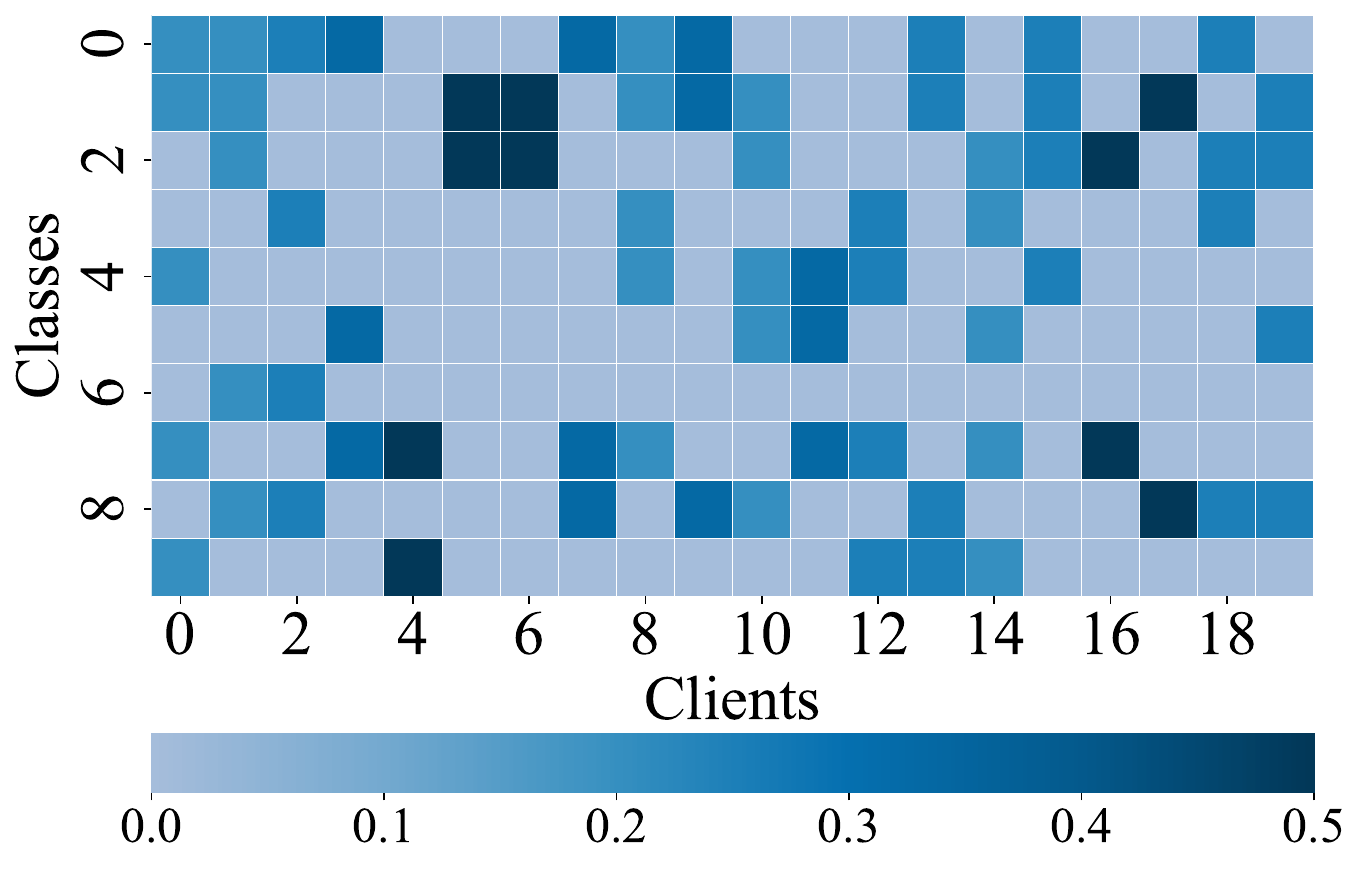}
  }
    \subfloat[\textit{Avg}=4, \textit{Std}=2]
  {
      \label{102415}  \includegraphics[width=0.32\linewidth]{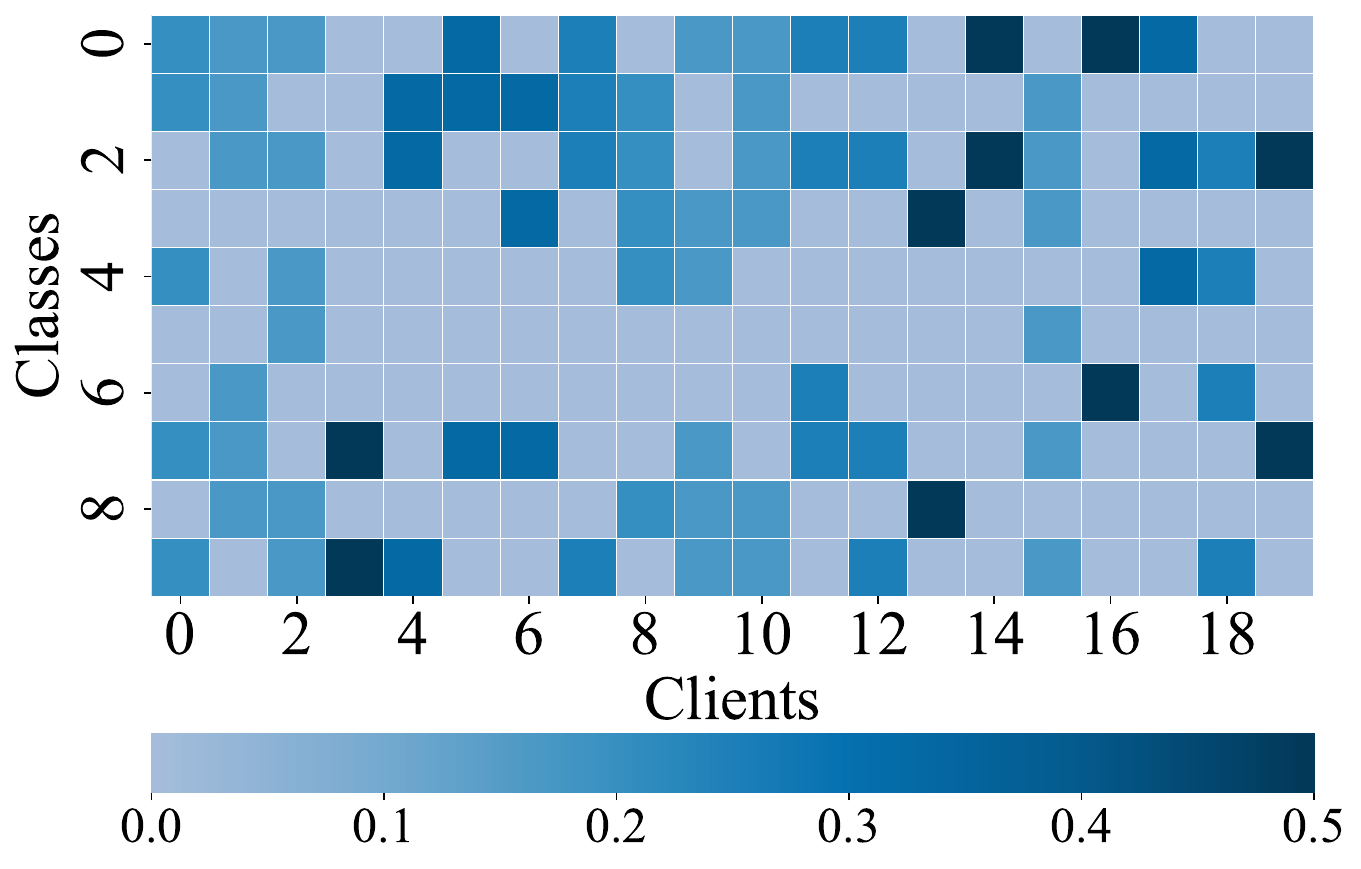}
  }
  \subfloat[\textit{Avg}=5, \textit{Std}=2]
  {
      \label{102414}  \includegraphics[width=0.31\linewidth]{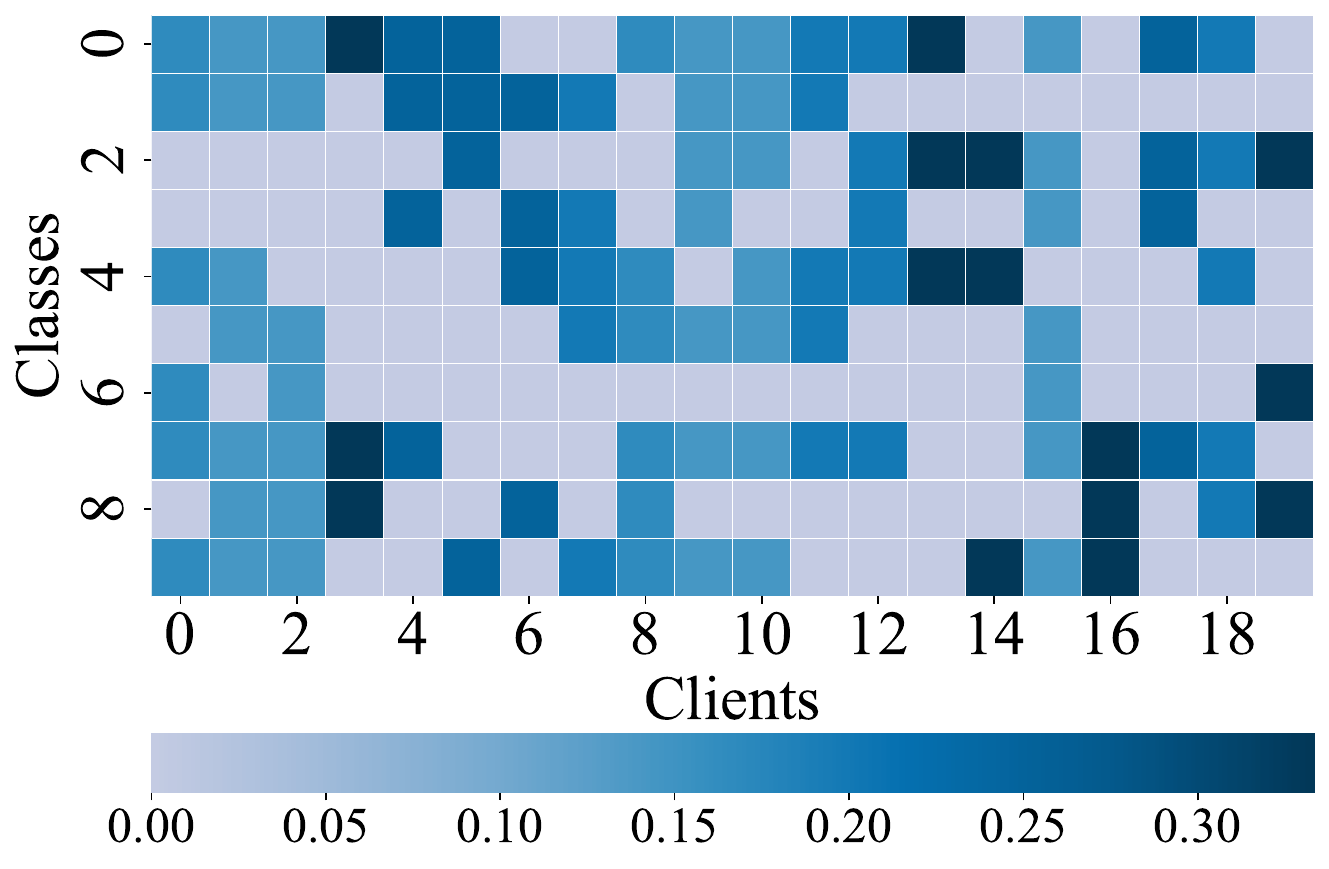}
  }
\caption{Heat maps  under different heterogeneity distributions.}
\label{110811}
\end{figure}

\subsubsection{Setting of Data Poisoning Attacks}
In our experiments, we evaluate two types of data poisoning attacks: feature attacks and label attacks.
For feature attacks, malicious clients randomly  alter  their own training data features in a completely randomized manner without following any   specific rule.
For label attacks, malicious clients modify  the labels  of their  training data  to incorrect  labels.
All training data  of  malicious clients are modified.
Unless otherwise specified, malicious clients keep   their behaviors fixed during federated training, and the attack  remains consistent throughout the training period.
Since the number of malicious clients affects FL performance differently, we set different proportions  of malicious  clients, denotes as \textit{Att}.

\subsubsection{Evaluation Measure}
Our goal is to improve the performance of cross-silo PPFL under   poisoned Non-IID data  while maintaining robustness against data poisoning attacks.
Therefore, we evaluate the FL performance  by measuring the average accuracy of  benign clients. For experiment comparison, we use FedAvg as the baseline, assuming it is not subject to data poisoning attacks. In addition, we  compare  PPFPL with robust schemes (i.e., Krum \cite{blanchard2017machine}, Foolsgold \cite{fung2018mitigating}, and FedDMC \cite{mu2024feddmc}) and privacy-preserving robust schemes (i.e., ShieldFL \cite{9762272}, PBFL \cite{9849010},   AntidoteFL \cite{liu2025antidotefl}, and FLDP-DM \cite{10966462}) under the same experiment conditions.

\subsection{Experimental Results}
\subsubsection{Visualization of Prototypes}

We visualize client-submitted prototypes in PPFPL using t-SNE, as shown in Fig. \ref{11181612}.
Specifically, we display  the distribution of prototypes from clients  under feature and label attacks with  $\textit{Att} =20\%$ on  CIFAR10  within a single communication round.
We observe that the prototypes of malicious clients deviate significantly from those of benign clients, regardless of   features attacks or labels attacks.
In addition, we  notice that the prototypes submitted by malicious clients are concentrated in a ring area.
This is because when  malicious clients perform local model training, the features or labels are disrupted, causing the generated prototypes to spread from the center point to the surrounding areas.
Most importantly, we observe a clear difference  between the prototypes submitted by benign and malicious clients. 
This difference  is prominent because it is independent of the data distribution.
The finding  validates the effectiveness of leveraging prototypes to defend against data poisoning attacks.
To effectively filter out malicious prototypes, it is essential to employ a secure aggregation protocol that ensures Byzantine-robust aggregation results.
The following is the security evaluation of PPFPL.

\begin{figure}
  \centering
  \subfloat[Feature attacks, \textit{Avg}=3,\textit{Std}=2.]
  {
      \label{111816121}  \includegraphics[width=0.44\linewidth]{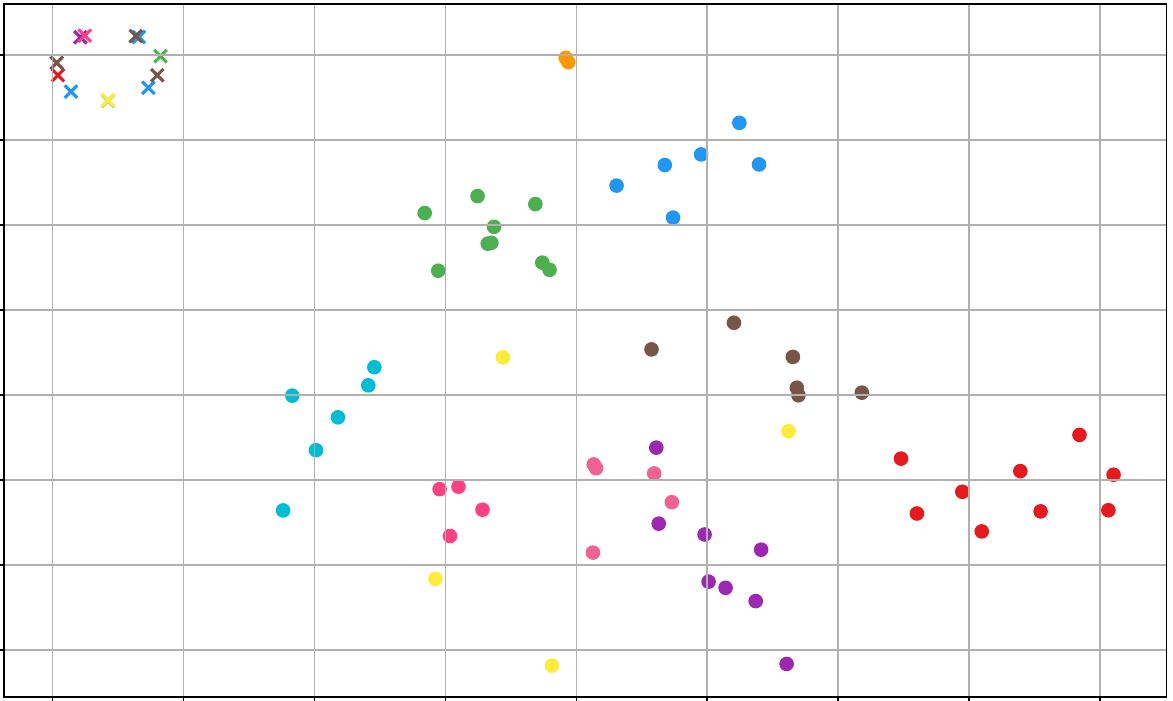}
  }
  \subfloat[Feature attacks, \textit{Avg}=5,\textit{Std}=1.]
  {
      \label{111816122}  \includegraphics[width=0.44\linewidth]{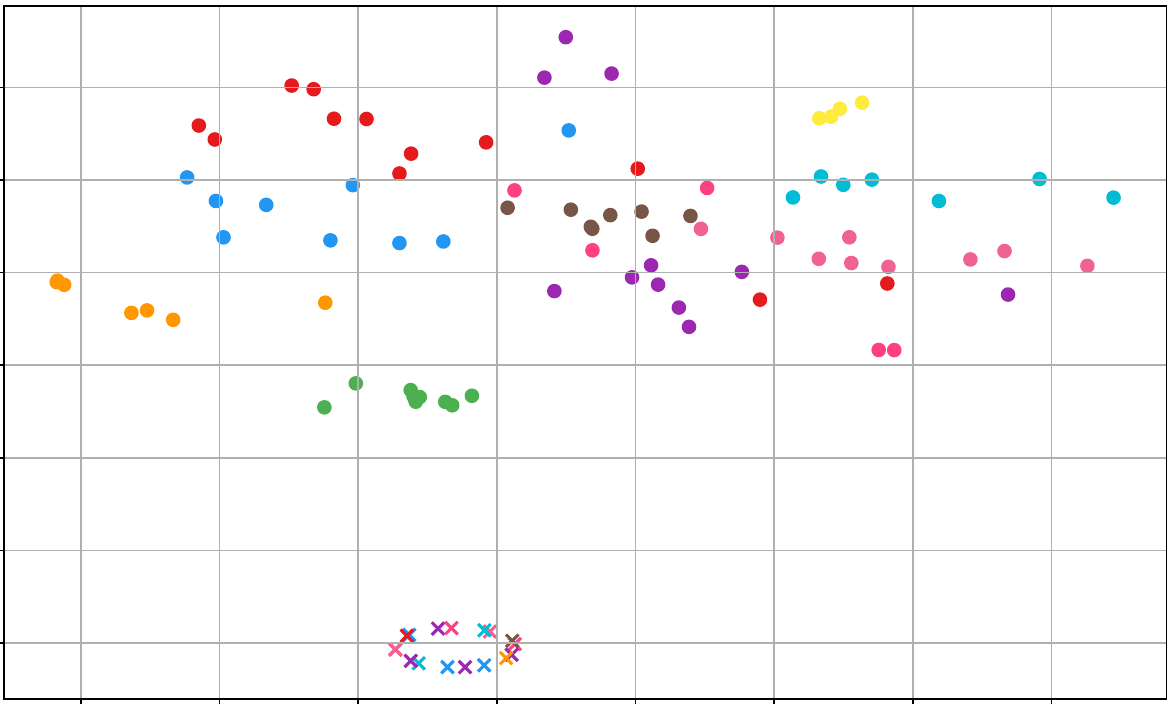}
  }
  
  \subfloat[Label attacks, \textit{Avg}=3$, $\textit{Std}=2.]
  {
      \label{111816123}  \includegraphics[width=0.44\linewidth]{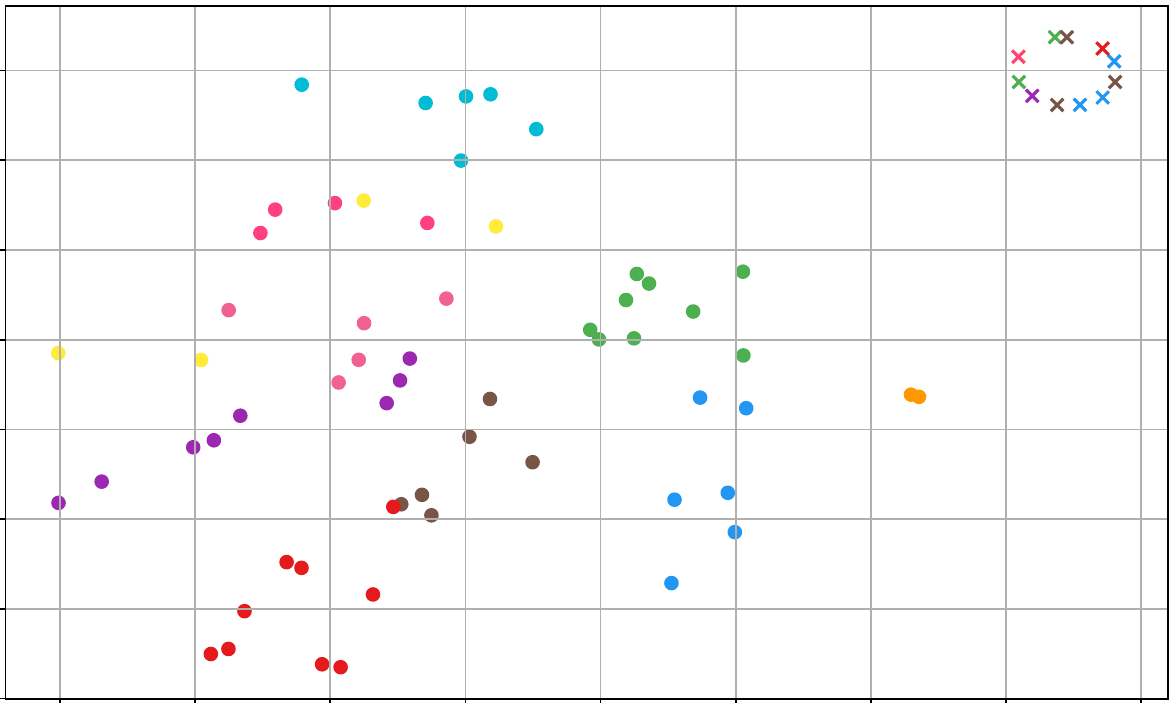}
  }
    \subfloat[Label attacks,  \textit{Avg}=5$, $\textit{Std}=1.]
  {
      \label{111816124}  \includegraphics[width=0.44\linewidth]{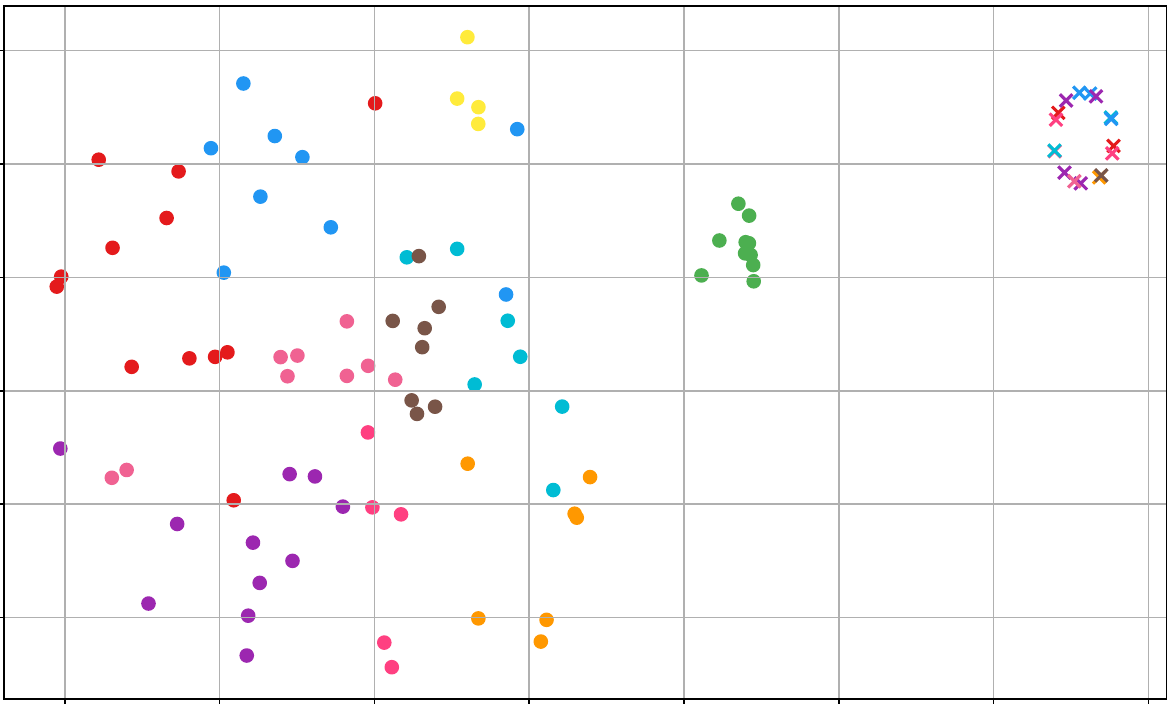}
  }
  \caption{
 The t-SNE visualization of client-submitted prototypes. 
 Different colors indicate distinct classes.
The symbol ``$\cdot$'' denotes prototypes from benign clients, and ``$\times$'' represents those from malicious clients.}
  \label{11181612}
\end{figure}

\subsubsection{Security Evaluation}
\begin{figure}[t]
  \centering
  \subfloat[Feature attacks, FMNIST.]
  {
      \label{102412}  \includegraphics[width=0.44\linewidth]{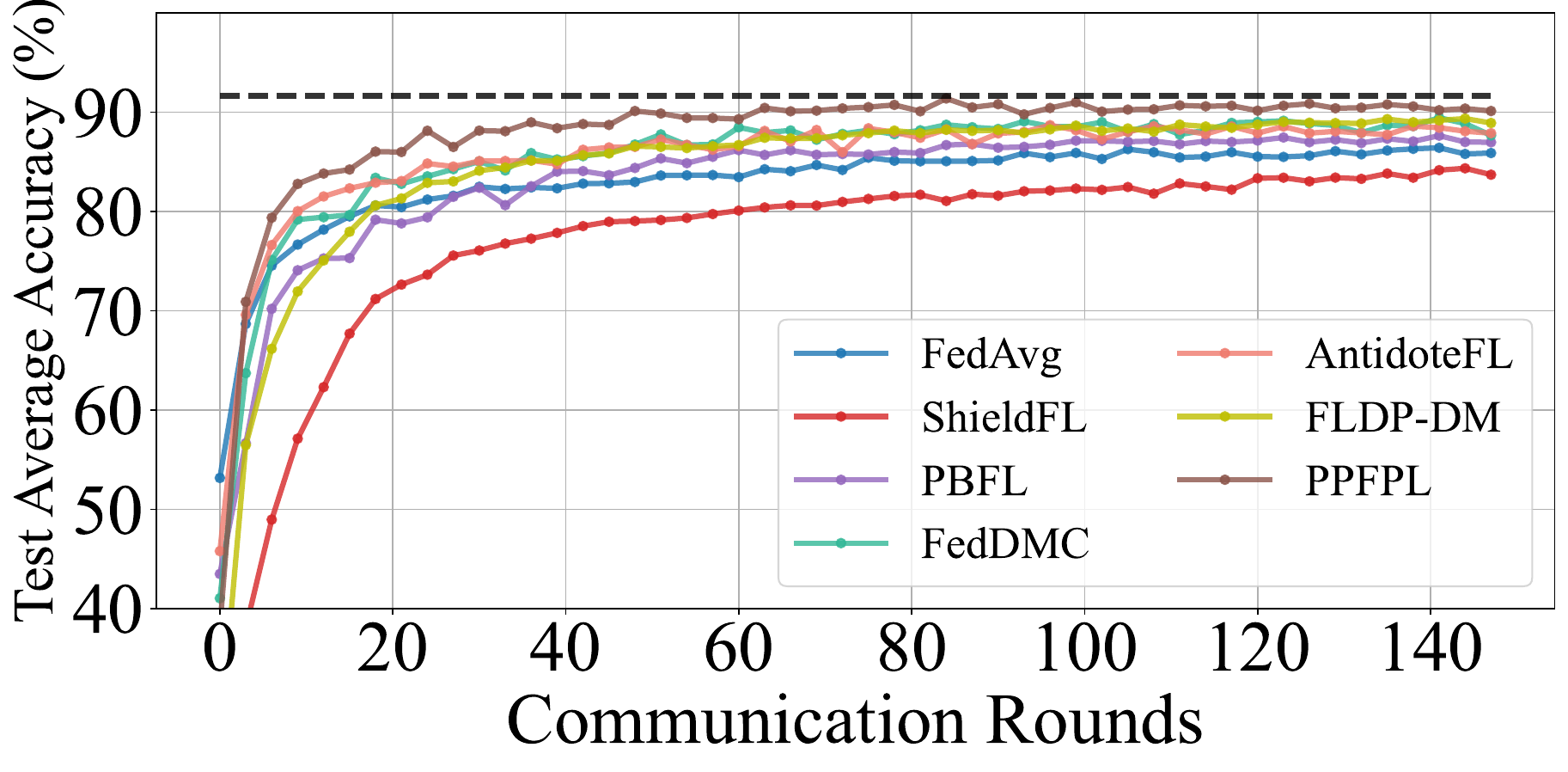}
  }
    \subfloat[Feature attacks,  CIFAR10.]
  {
      \label{102413}  \includegraphics[width=0.44\linewidth]{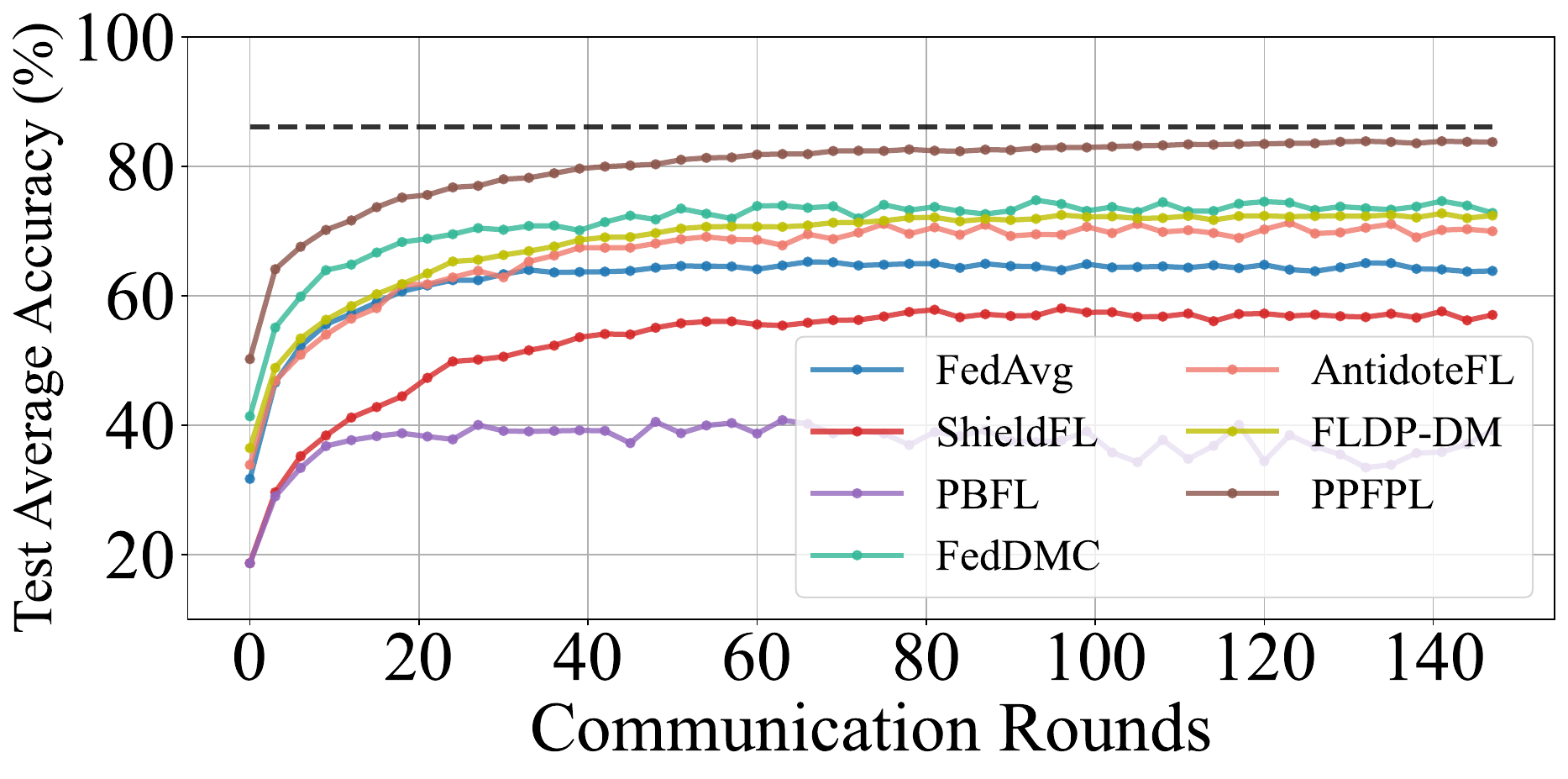}
  }
  
  \subfloat[Label attacks, FMNIST.]
  {
      \label{102415}  \includegraphics[width=0.44\linewidth]{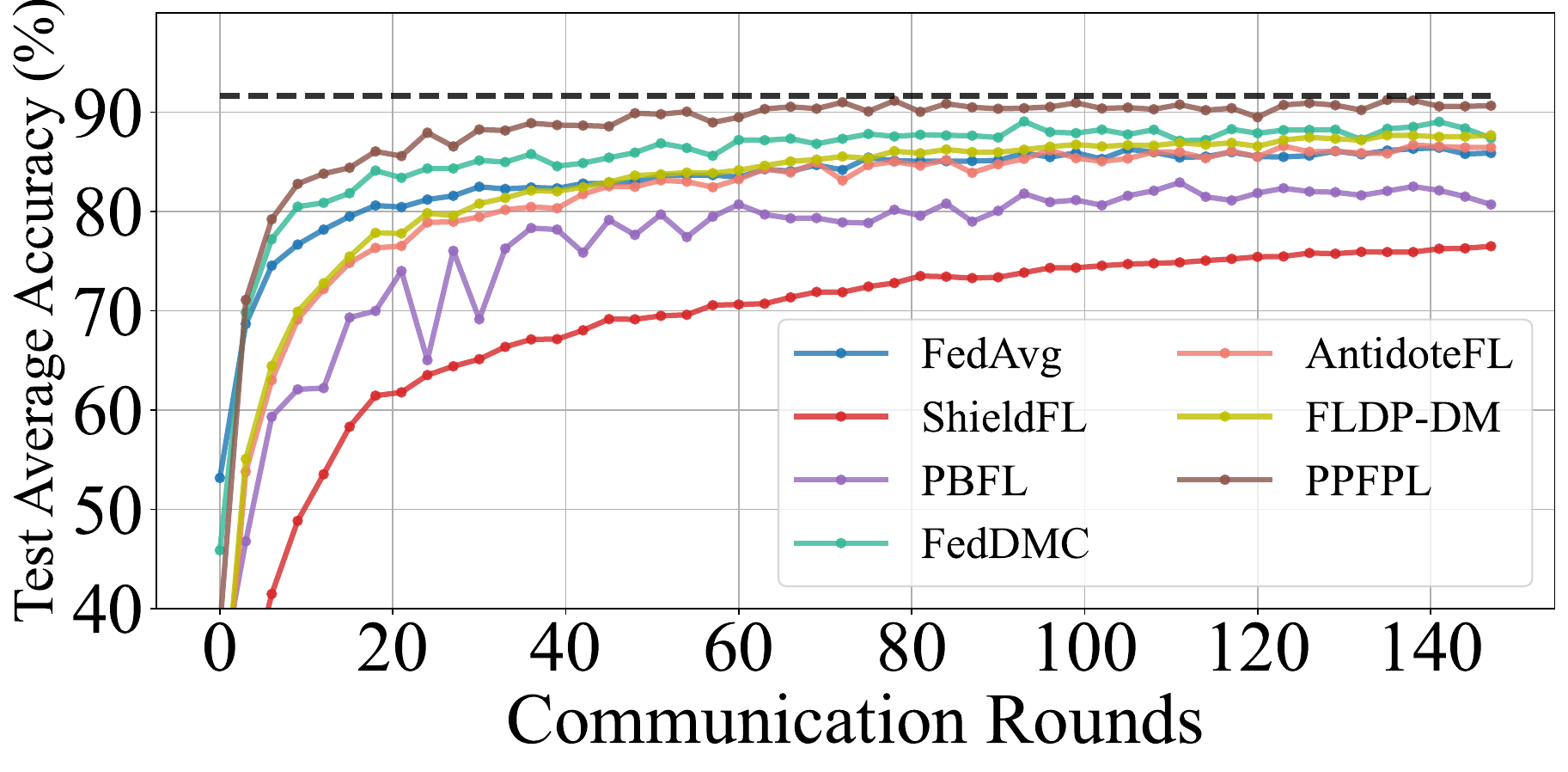}
  }
    \subfloat[Label attacks, CIFAR10.]
  {
      \label{102416}  \includegraphics[width=0.44\linewidth]{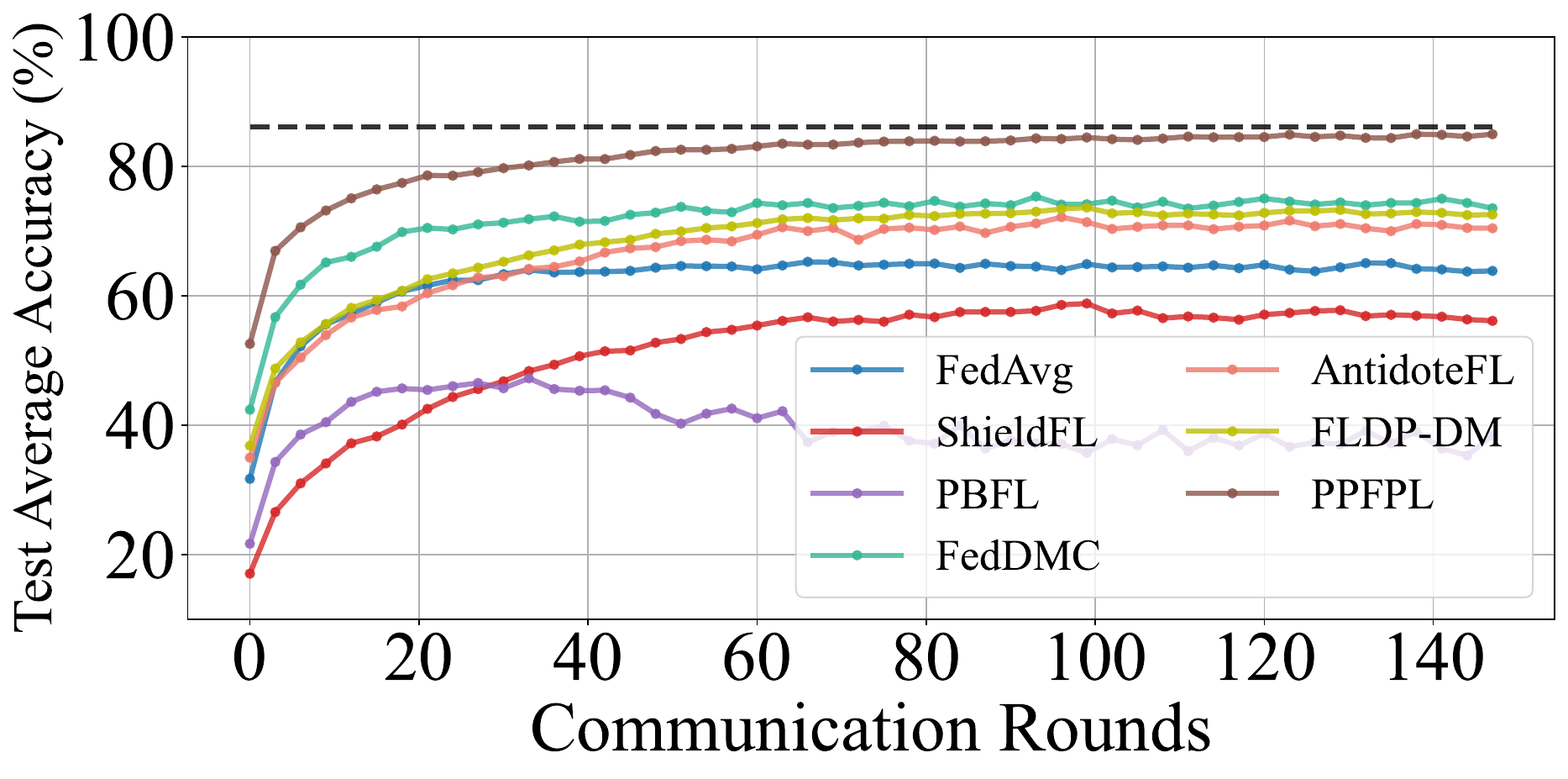}
  }
  \caption{Test average accuracy comparison between PPFPL and existing schemes  in  poisoned Non-IID data.}
  \label{10241}
\end{figure}

To evaluate the security  of our framework, we  evaluate PPFPL's performance against feature attacks and  label attacks on three datasets   under Non-IID setting with \textit{Avg}=3 and \textit{Std}=2.
Specifically, the proportion of malicious clients is set to \textit{Att}=20\%.
 Furthermore, we compare PPFPL with existing defense schemes  (i.e., FedAvg, ShieldFL, PBFL, FedDMC,  AntidoteFL, and FLDP-DM), where the global model learned by FedAvg is not subject to data poisoning attacks.

The experiment results are shown in Fig. \ref{10241}, where  the black dashed line represents the highest accuracy of PPFPL without attacks.
We  observe that  ShieldFL and PBFL suffer significant accuracy degradation under attacks compared to FedAvg without attacks on the CIFAR10 dataset.
This indicates that as   datasets become more complex, the negative impact of data poisoning attacks on model accuracy becomes more obvious. This is mainly because   complex datasets increase the difficulty of defense,  and  poisoned Non-IID data further degrades the defense performance.
Notably, FedDMC and AntidoteFL achieve higher performance compared to FedAvg because they adjust the aggregation weights of  model updates, which mitigates the bias caused by Non-IID data and improves the performance of federated learning.
Moreover, PPFPL (under attacks) outperforms FedAvg (no attacks), with its accuracy approaching the highest accuracy of PPFPL without attacks  when it nears convergence.
This demonstrates  that PPFPL guarantees the high accuracy of  the learned  model  in the presence of poisoned Non-IID data, which is benefited that the prototype is not affected by tampered data distribution, while the secure aggregation protocol resists malicious prototypes submitted by clients.
Thus, our PPFPL overcomes the difficulty confronted by  these defense schemes, and  improves the FL performance  in poisoned Non-IID data.

Additionally, we evaluate the change of loss during training for PPFPL and the compared schemes, as shown in Fig. \ref{10252355}.
We observe that PBFL's loss fluctuates significantly throughout the training process, while PPFPL demonstrates relatively smooth convergence.
This indicates that  data poisoning attacks severely disrupt PBFL's convergence, while the convergence of PPFPL is not affected by  data poisoning attacks,  thereby preserving model accuracy.
This is consistent with our Corollary 3. Specifically, PPFPL still converges in the presence of data poisoning attacks as long as $\lambda$, $E$, and $\eta$ satisfy a specific relationship among them.
Moreover, PPFPL's loss at convergence is lower than other defense schemes under CIFAR10, which can satisfy the design goal of security.
Notably,  our experiments show that  feature attacks and label attacks have  similar  impacts on the performance of PPFPL and other schemes.
This indicates that the impact on model performance is similar under poisoned Non-IID data, regardless of the type of data poisoning attacks.
\begin{figure}[t]
  \centering
  \subfloat[Feature attacks, CIFAR10.]
  {
      \label{102523555}  \includegraphics[width=0.44\linewidth]{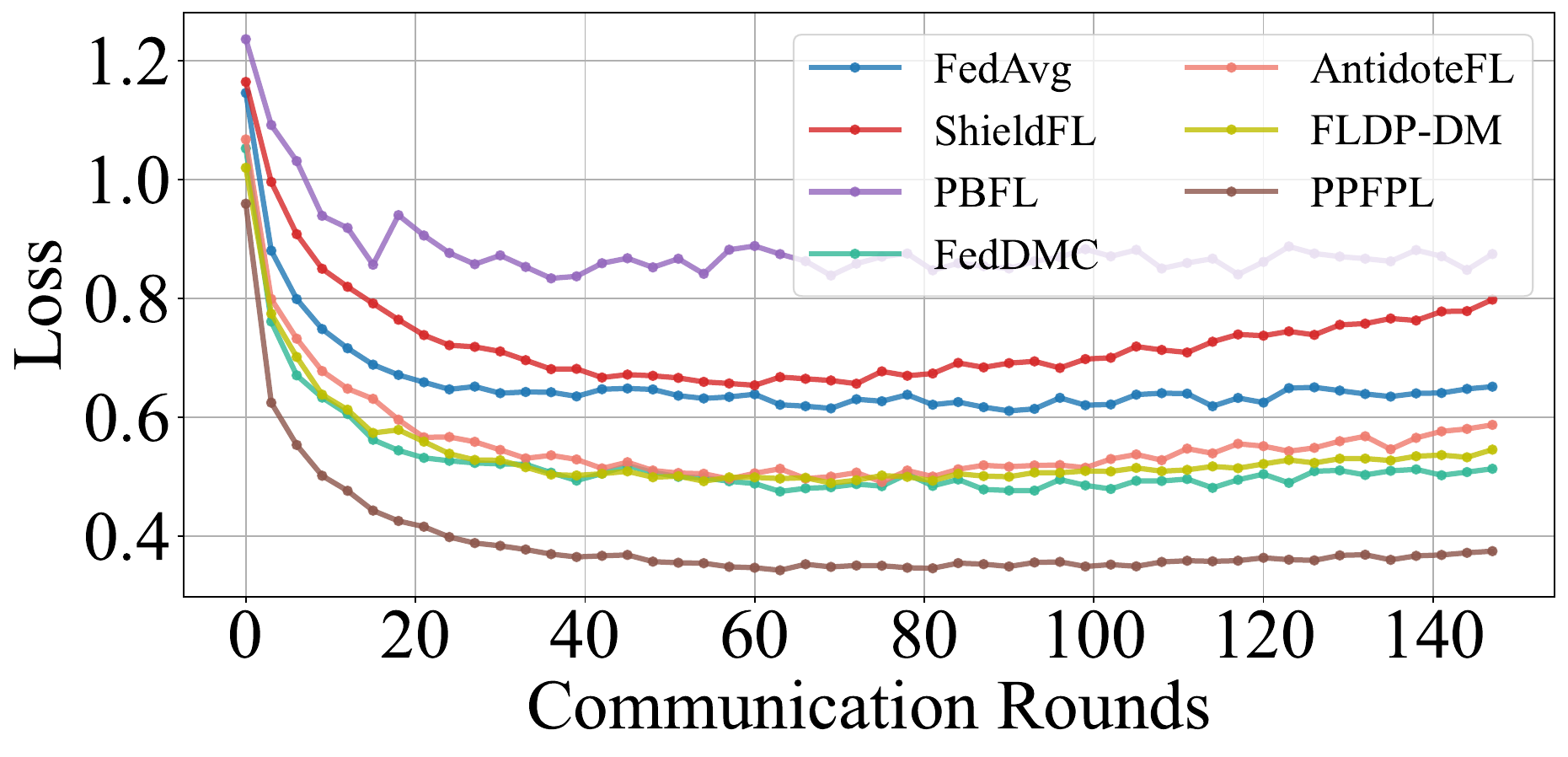}
  }
  \subfloat[Label attacks, CIFAR10.]
  {
      \label{102523556}  \includegraphics[width=0.44\linewidth]{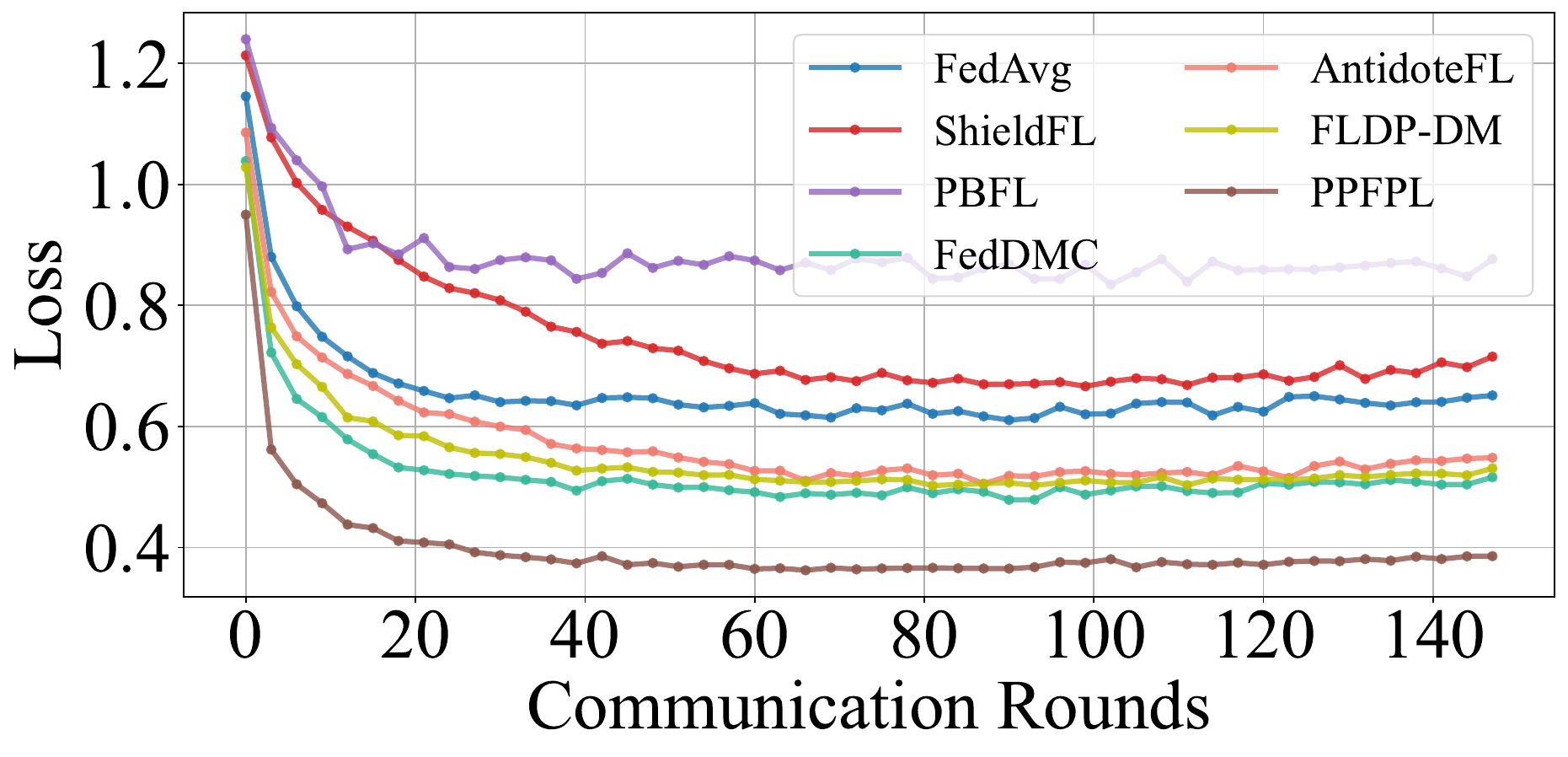}
  }
  \caption{Loss comparison between PPFPL and existing schemes  in poisoned Non-IID data.}
  \label{10252355}
\end{figure}

\subsubsection{Different Data Distributions for PPFPL}
To further test the security of PPFPL under different data distributions, we evaluate the accuracy of  PPFPL against feature attacks under different Non-IID  conditions  in MNIST, FMNIST and CIFAR10.
In addition, we compare its performance with  Krum, Foolsgold, ShieldFL, PBFL, FedDMC,   AntidoteFL, and FLDP-DM.
The specific experimental results are reported in TABLE \ref{112102}.
Notably, since test average accuracy has fluctuated after each communication round,  we  select the average of five highest test average accuracies  across all communication rounds.
From TABLE  \ref{112102}, we can observe that the performance of PPFPL is not easily affected by changes in data distribution, while other schemes are very susceptible to data distribution.
This is attributed to that the client-submitted prototype does not change due to the change in data distribution.
In addition, we observe a slight degradation in the performance of PPFPL with $Att =30\%$ attacks. 
This is because the tampered Non-IID data  reduces the contribution of effective samples to the model, resulting in a slightly lower performance, which is a reasonable phenomenon.
Overall, our PPFPL improves performance in cross-silo PPFL under poisoned Non-IID data while resisting data poisoning attacks, making it suitable for real-world deployments.

\begin{table}[t]
\centering
  \caption{Test Average Accuracy (\%) on MNIST, FMNIST and CIFAR10 with feature attacks.}
\tabcolsep=0.25cm
\scalebox{0.65}{
\begin{tabular}{c|c|c|c|ccccc}
\hline\hline
\multirow{2}{*}{\textbf{Dataset}} & \multirow{2}{*}{\textbf{Method}} & \multirow{2}{*}{\textit{Att} \%} & \multirow{2}{*}{\textit{Std}} & \multicolumn{5}{c}{\textbf{Test Average Accuracy (\%)}}                             \\ \cline{5-9}
                      &                         &                     &                     & \multicolumn{1}{c}{\textit{Avg = }3   } & \multicolumn{1}{c}{\textit{Avg} = 4 } &\textit{ Avg }= 5          &  \textit{ Avg }= 6     &  \textit{ Avg }= 7               \\ \hline
                      \multirow{32}{*}{\textbf{MNIST}}     & \multirow{4}{*}{Krum \cite{blanchard2017machine}}       & \multirow{2}{*}{20}    &   1                   & \multicolumn{1}{c}{95.96} & \multicolumn{1}{c}{96.95} &97.01     &97.04  &\textbf{97.12}            \\
                   &                         &                      &    2                  & \multicolumn{1}{c}{95.67} & \multicolumn{1}{c}{94.28} &  94.88       &95.51  &96.20    \\\cline{3-9}
                      &                         & \multirow{2}{*}{30}    &  1                    & \multicolumn{1}{c}{95.74} & \multicolumn{1}{c}{96.61} &96.82         & \textbf{96.78}  &\textbf{96.92}   \\
                   &                         &                      &    2                  & \multicolumn{1}{c}{95.67} & \multicolumn{1}{c}{94.28} &  94.88    &94.77  &95.80      \\ \cline{2-9}
& \multirow{4}{*}{Foolsgold \cite{fung2018mitigating}}       & \multirow{2}{*}{20}    &   1                   & \multicolumn{1}{c}{96.69} & \multicolumn{1}{c}{94.03} &95.13      &  95.76       & 95.39 \\
                   &                         &                      &    2                  & \multicolumn{1}{c}{96.19} & \multicolumn{1}{c}{94.82} &  96.73  &   95.27       & 95.75    \\\cline{3-9}
                      &                         & \multirow{2}{*}{30}    &  1                    & \multicolumn{1}{c}{96.47} & \multicolumn{1}{c}{93.76} &94.98       &  95.21       & 95.00           \\
                   &                         &                      &    2                  & \multicolumn{1}{c}{96.02} & \multicolumn{1}{c}{94.74} &  96.58     &   95.01       & 95.49    \\\cline{2-9}
 & \multirow{4}{*}{PBFL \cite{9849010}}       & \multirow{2}{*}{20}    &   1                   & \multicolumn{1}{c}{97.78} & \multicolumn{1}{c}{\textbf{97.57}} &96.34    &    \textbf{97.10}   &  96.98   \\
                   &                         &                      &    2                  & \multicolumn{1}{c}{97.36} & \multicolumn{1}{c}{96.17} &  96.26    &94.64   &  95.41     \\\cline{3-9}
                      &                         & \multirow{2}{*}{30}    &  1  & \multicolumn{1}{c}{95.78} & \multicolumn{1}{c}{94.15} & 95.95           &    96.50   &  96.75    \\
                      &                         &                      &       2               & \multicolumn{1}{c}{95.40} & \multicolumn{1}{c}{93.62} &93.48   &94.70   &  94.62    \\\cline{2-9}
 & \multirow{4}{*}{ShieldFL \cite{9762272}}       & \multirow{2}{*}{20}    &   1                   & \multicolumn{1}{c}{96.91} & \multicolumn{1}{c}{94.05} &95.10    &    95.64   & 95.25   \\
                      &                         &                      &                2      & \multicolumn{1}{c}{97.25} & \multicolumn{1}{c}{94.88} & 95.69  & 95.44   & 95.76        \\\cline{3-9}
                      &                         & \multirow{2}{*}{30}    &             1          & \multicolumn{1}{c}{96.77} & \multicolumn{1}{c}{93.79} &94.93     & 95.60      & 95.81   \\
                      &                         &                      &                2      & \multicolumn{1}{c}{97.08} & \multicolumn{1}{c}{94.69} & 95.46     & 95.73      & 94.46    \\\cline{2-9}
 & \multirow{4}{*}{FedDMC \cite{mu2024feddmc}}       & \multirow{2}{*}{20}    &   1                   & \multicolumn{1}{c}{97.62} & \multicolumn{1}{c}{97.46} &97.16    &    96.86   &  96.53   \\
                   &                         &                      &    2                  & \multicolumn{1}{c}{97.53} & \multicolumn{1}{c}{97.32} &  97.04    &96.73   &  96.39     \\\cline{3-9}
                      &                         & \multirow{2}{*}{30}    &  1  & \multicolumn{1}{c}{97.62} & \multicolumn{1}{c}{96.98} & 96.88           &    96.46   &  96.17    \\
                      &                         &                      &       2               & \multicolumn{1}{c}{97.43} & \multicolumn{1}{c}{96.58} &96.34   &96.26   &  94.33    \\\cline{2-9}
 & \multirow{4}{*}{AntidoteFL \cite{liu2025antidotefl}}       & \multirow{2}{*}{20}    &   1                   & \multicolumn{1}{c}{97.58} & \multicolumn{1}{c}{97.40} &97.09    &    96.73   & 96.34   \\
                      &                         &                      &                2      & \multicolumn{1}{c}{97.31} & \multicolumn{1}{c}{97.18} & 96.85  & 95.45   & 96.30        \\\cline{3-9}
                      &                         & \multirow{2}{*}{30}    &             1          & \multicolumn{1}{c}{97.35} & \multicolumn{1}{c}{96.39} &96.42     & 96.06      & 96.20   \\
                      &                         &                      &                2      & \multicolumn{1}{c}{97.08} & \multicolumn{1}{c}{94.69} & 95.46     & 95.73      & 94.46    \\\cline{2-9}
 & \multirow{4}{*}{FLDP-DM \cite{10966462}}       & \multirow{2}{*}{20}    &   1                   & \multicolumn{1}{c}{97.60} & \multicolumn{1}{c}{97.22} &97.15    &    96.78   & 96.37   \\
                      &                         &                      &                2      & \multicolumn{1}{c}{97.48} & \multicolumn{1}{c}{97.16} & 96.93  & 96.40   & 96.08        \\\cline{3-9}
                      &                         & \multirow{2}{*}{30}    &             1          & \multicolumn{1}{c}{97.14} & \multicolumn{1}{c}{96.85} &96.77     & 96.76      & 96.42   \\
                      &                         &                      &                2      & \multicolumn{1}{c}{96.95} & \multicolumn{1}{c}{96.36} & 96.12     & 95.91      & 94.87    \\\cline{2-9}
& \multirow{4}{*}{PPFPL}       & \multirow{2}{*}{20}    &    1                  & \multicolumn{1}{c}{\textbf{97.87}} & \multicolumn{1}{c}{97.53} &\textbf{97.35}    &   97.08    &  96.89      \\
                      &                         &                      &             2         & \multicolumn{1}{c}{\textbf{97.74}} & \multicolumn{1}{c}{\textbf{97.49}} & \textbf{97.43}   & \textbf{96.82}    &  \textbf{96.56}        \\\cline{3-9}
                      &                         & \multirow{2}{*}{30}    &       1               & \multicolumn{1}{c}{\textbf{97.85}} & \multicolumn{1}{c}{\textbf{97.04}} &\textbf{97.02}     & 96.52  & 96.37        \\
                      &                         &                      &               2       & \multicolumn{1}{c}{\textbf{97.52}} & \multicolumn{1}{c}{\textbf{96.92}} & \textbf{96.70}    & \textbf{96.50}& \textbf{96.48}    \\ \hline
\multirow{32}{*}{\textbf{FMNIST}}
 & \multirow{4}{*}{Krum \cite{blanchard2017machine}}       & \multirow{2}{*}{20}    &                     1 & \multicolumn{1}{c}{84.24} & \multicolumn{1}{c}{84.80} & 86.03& 86.52 &     87.35    \\
                      &                         &                      &           2           & \multicolumn{1}{c}{85.02} & \multicolumn{1}{c}{83.64} &  86.82    & 84.75 &     86.08              \\\cline{3-9}
                      &                         & \multirow{2}{*}{30}    &    1                  & \multicolumn{1}{c}{83.72} & \multicolumn{1}{c}{84.38} & 85.49    &  85.12 &     86.79   \\
                      &                         &                      &           2           & \multicolumn{1}{c}{84.62} & \multicolumn{1}{c}{82.73} &  85.73         & 82.86 &     84.09          \\\cline{2-9}
  & \multirow{4}{*}{Foolsgold \cite{fung2018mitigating}}       & \multirow{2}{*}{20}    &                     1 & \multicolumn{1}{c}{84.88} & \multicolumn{1}{c}{84.73} & 84.04  &  83.55  & 82.62 \\
                      &                         &                      &           2           & \multicolumn{1}{c}{82.49} & \multicolumn{1}{c}{83.88} &  85.31     &  81.75  & 82.24           \\\cline{3-9}
                      &                         & \multirow{2}{*}{30}    &    1                  &  \multicolumn{1}{c}{84.06} & \multicolumn{1}{c}{84.01} & 83.50    &  82.70  & 82.04    \\
                      &                         &                      &           2           & \multicolumn{1}{c}{81.86} & \multicolumn{1}{c}{83.43} &  84.98       &  81.19  & 82.68             \\\cline{2-9}
   & \multirow{4}{*}{PBFL \cite{9849010}}       & \multirow{2}{*}{20}    &                     1 & \multicolumn{1}{c}{87.01} & \multicolumn{1}{c}{80.37} & 76.93   & 74.26   & 72.10 \\
                      &                         &                      &           2           & \multicolumn{1}{c}{87.42} & \multicolumn{1}{c}{76.31} &  74.47        & 72.99   & 72.43        \\\cline{3-9}
                      &                         & \multirow{2}{*}{30}    &    1                  & \multicolumn{1}{c}{86.92} & \multicolumn{1}{c}{75.45} & 73.06        & 73.70  & 72.84        \\
                      &                         &                      &                 2     & \multicolumn{1}{c}{79.34} & \multicolumn{1}{c}{72.68} & 72.81        & 70.57   & 70.49            \\\cline{2-9}
 & \multirow{4}{*}{ShieldFL \cite{9762272}}       & \multirow{2}{*}{20}    &                     1 & \multicolumn{1}{c}{87.01} & \multicolumn{1}{c}{78.90} & 78.06  & 70.90 &  70.26\\
                      &                         &                      &           2           & \multicolumn{1}{c}{76.47} & \multicolumn{1}{c}{77.89} &  70.96   &    67.47 &  66.80         \\\cline{3-9}
                      &                         & \multirow{2}{*}{30}    &    1                  & \multicolumn{1}{c}{86.24} & \multicolumn{1}{c}{78.46} & 77.51  &69.33 &  68.81    \\
                      &                         &                      &           2           & \multicolumn{1}{c}{75.81} & \multicolumn{1}{c}{77.35} &  70.57    &    66.42 &  65.46    \\\cline{2-9}
 & \multirow{4}{*}{FedDMC \cite{mu2024feddmc}}       & \multirow{2}{*}{20}    &   1                   & \multicolumn{1}{c}{90.60} & \multicolumn{1}{c}{88.21} &87.03    &    86.40   &  86.11   \\
                   &                         &                      &    2                  & \multicolumn{1}{c}{89.53} & \multicolumn{1}{c}{88.70} &  87.27    &87.20   &  87.13     \\\cline{3-9}
                      &                         & \multirow{2}{*}{30}    &  1  & \multicolumn{1}{c}{89.73} & \multicolumn{1}{c}{88.03} & 86.38           &    86.70   &  86.14    \\
                      &                         &                      &       2               & \multicolumn{1}{c}{89.40} & \multicolumn{1}{c}{87.52} &86.23   &86.40   &  86.12    \\\cline{2-9}
 & \multirow{4}{*}{AntidoteFL \cite{liu2025antidotefl}}       & \multirow{2}{*}{20}    &   1                   & \multicolumn{1}{c}{89.13} & \multicolumn{1}{c}{87.20} &86.58    &    85.36   & 85.25   \\
                      &                         &                      &                2      & \multicolumn{1}{c}{88.72} & \multicolumn{1}{c}{86.13} & 85.34  & 85.16   & 86.48        \\\cline{3-9}
                      &                         & \multirow{2}{*}{30}    &             1          & \multicolumn{1}{c}{87.25} & \multicolumn{1}{c}{86.24} &85.13     & 84.10      & 84.25   \\
                      &                         &                      &                2      & \multicolumn{1}{c}{87.10} & \multicolumn{1}{c}{86.61} & 85.10     & 84.39      & 84.90    \\\cline{2-9}
 & \multirow{4}{*}{FLDP-DM \cite{10966462}}       & \multirow{2}{*}{20}    &   1                   & \multicolumn{1}{c}{90.48} & \multicolumn{1}{c}{88.06} &86.53    &    86.17   & 85.86   \\
                      &                         &                      &                2      & \multicolumn{1}{c}{89.32} & \multicolumn{1}{c}{87.94} & 86.95  & 85.09   & 84.77        \\\cline{3-9}
                      &                         & \multirow{2}{*}{30}    &             1          & \multicolumn{1}{c}{88.64} & \multicolumn{1}{c}{87.35} &85.90     & 85.47      & 85.30   \\
                      &                         &                      &                2      & \multicolumn{1}{c}{88.79} & \multicolumn{1}{c}{86.92} & 85.71     & 85.06      & 84.55    \\\cline{2-9}
& \multirow{4}{*}{PPFPL}       & \multirow{2}{*}{20}    &     1                & \multicolumn{1}{c}{\textbf{91.38}} & \multicolumn{1}{c}{\textbf{90.18}} &    \textbf{89.27} &    \textbf{88.56}&    \textbf{88.13}    \\
                      &                         &                      &                 2     & \multicolumn{1}{c}{\textbf{90.48}} & \multicolumn{1}{c}{\textbf{89.63}} &\textbf{88.45}    &    \textbf{88.79}&    \textbf{88.20}    \\\cline{3-9}
                      &                         & \multirow{2}{*}{30}    &      1                & \multicolumn{1}{c}{\textbf{90.97}} & \multicolumn{1}{c}{\textbf{90.08}} &\textbf{88.93}    & \textbf{88.33}&    \textbf{88.14}   \\
                      &                         &                      &                 2     & \multicolumn{1}{c}{\textbf{90.40}} & \multicolumn{1}{c}{\textbf{88.62}} &\textbf{87.74}   & \textbf{88.41}&    \textbf{87.97}  \\ \hline
\multirow{32}{*}{\textbf{CIFAR10}}
& \multirow{4}{*}{Krum \cite{blanchard2017machine}}       & \multirow{2}{*}{20}    &   1                   & \multicolumn{1}{c}{62.91} & \multicolumn{1}{c}{61.24} &68.19  &68.87 & 69.22       \\
                      &                         &                      &                2      & \multicolumn{1}{c}{57.90} & \multicolumn{1}{c}{57.94} & 68.33   &68.65 &   68.79       \\\cline{3-9}
                      &                         & \multirow{2}{*}{30}    &             1         & \multicolumn{1}{c}{60.10} & \multicolumn{1}{c}{60.73} &66.36  &  68.70 & 69.01          \\
                      &                         &                      &                  2    & \multicolumn{1}{c}{56.67} & \multicolumn{1}{c}{57.29} & 67.03    &68.49 &   68.51      \\\cline{2-9}
& \multirow{4}{*}{Foolsgold \cite{fung2018mitigating}}       & \multirow{2}{*}{20}    &   1                   & \multicolumn{1}{c}{60.91} & \multicolumn{1}{c}{59.24} &66.19    &66.45  &   66.52 \\
                      &                         &                      &                2      & \multicolumn{1}{c}{56.90} & \multicolumn{1}{c}{55.94} & 66.33     & 66.62  &   66.70      \\\cline{3-9}
                      &                         & \multirow{2}{*}{30}    &             1        & \multicolumn{1}{c}{60.37} & \multicolumn{1}{c}{58.68} &65.64  &66.26  &   66.44    \\
                      &                         &                      &                2      & \multicolumn{1}{c}{56.17} & \multicolumn{1}{c}{55.09} & 65.85    &   66.40  & 66.41  \\\cline{2-9}
& \multirow{4}{*}{PBFL \cite{9849010}}       & \multirow{2}{*}{20}    &                     1 & \multicolumn{1}{c}{38.06} & \multicolumn{1}{c}{39.18} &  41.81 &  42.42    &   42.75   \\
                      &                         &                      &                   2   & \multicolumn{1}{c}{39.46} & \multicolumn{1}{c}{40.54} &  41.30   & 42.12    &   42.76     \\\cline{3-9}
                      &                         & \multirow{2}{*}{30}    &              1        & \multicolumn{1}{c}{37.18} & \multicolumn{1}{c}{39.02} & 41.39   &   42.32    &   42.41      \\
                      &                         &                      &                 2     & \multicolumn{1}{c}{37.80} & \multicolumn{1}{c}{37.52} &41.08    & 41.96    &   42.67    \\\cline{2-9}
 & \multirow{4}{*}{ShieldFL \cite{9762272}}       & \multirow{2}{*}{20}    &   1                   & \multicolumn{1}{c}{61.62} & \multicolumn{1}{c}{61.94} &68.31  & 69.53   & 68.66  \\
                      &                         &                      &                2      & \multicolumn{1}{c}{57.84} & \multicolumn{1}{c}{57.26} & 66.14     & 67.43   & 68.62    \\\cline{3-9}
                      &                         & \multirow{2}{*}{30}    &             1         & \multicolumn{1}{c}{60.94} & \multicolumn{1}{c}{61.25} &67.50   & 67.01   & 68.24      \\
                      &                         &                      &                2      & \multicolumn{1}{c}{57.10} & \multicolumn{1}{c}{56.76} & 65.47    & 66.42   & 67.69       \\\cline{2-9}
 & \multirow{4}{*}{FedDMC \cite{mu2024feddmc}}       & \multirow{2}{*}{20}    &   1                   & \multicolumn{1}{c}{75.28} & \multicolumn{1}{c}{74.16} &74.37    &    73.86   &  \textbf{73.50}   \\
                   &                         &                      &    2                  & \multicolumn{1}{c}{73.52} & \multicolumn{1}{c}{72.40} &  72.37    &72.06   &  71.58     \\\cline{3-9}
                      &                         & \multirow{2}{*}{30}    &  1  & \multicolumn{1}{c}{74.60} & \multicolumn{1}{c}{74.28} & 73.97           &    73.62   &   \textbf{73.39}    \\
                      &                         &                      &       2               & \multicolumn{1}{c}{72.67} & \multicolumn{1}{c}{72.48} &72.11   &71.85   &  70.78    \\\cline{2-9}
 & \multirow{4}{*}{AntidoteFL \cite{liu2025antidotefl}}       & \multirow{2}{*}{20}    &   1                   & \multicolumn{1}{c}{72.66} & \multicolumn{1}{c}{72.27} &71.19    &    72.19   & 70.65   \\
                      &                         &                      &                2      & \multicolumn{1}{c}{70.86} & \multicolumn{1}{c}{70.40} & 70.19  & 69.84   & 69.75        \\\cline{3-9}
                      &                         & \multirow{2}{*}{30}    &             1          & \multicolumn{1}{c}{71.49} & \multicolumn{1}{c}{71.61} &70.43     & 71.24      & 69.20   \\
                      &                         &                      &                2      & \multicolumn{1}{c}{69.82} & \multicolumn{1}{c}{68.70} & 68.42     & 68.83      & 68.17    \\\cline{2-9}
 & \multirow{4}{*}{FLDP-DM \cite{10966462}}       & \multirow{2}{*}{20}    &   1                   & \multicolumn{1}{c}{74.89} & \multicolumn{1}{c}{72.48} &72.07    &    70.63   & 68.30   \\
                      &                         &                      &                2      & \multicolumn{1}{c}{72.30} & \multicolumn{1}{c}{71.55} & 70.50  & 69.38   & 69.60        \\\cline{3-9}
                      &                         & \multirow{2}{*}{30}    &             1          & \multicolumn{1}{c}{73.86} & \multicolumn{1}{c}{71.79} &70.49     & 70.43      & 68.02   \\
                      &                         &                      &                2      & \multicolumn{1}{c}{72.05} & \multicolumn{1}{c}{71.33} & 69.54     & 69.26      & 67.40    \\\cline{2-9}
& \multirow{4}{*}{PPFPL}       & \multirow{2}{*}{20}    &   1                   & \multicolumn{1}{c}{\textbf{83.41}} & \multicolumn{1}{c}{\textbf{80.95}} &\textbf{77.82}     &\textbf{74.76} & 71.35     \\
                      &                         &                      &                2      & \multicolumn{1}{c}{\textbf{83.44}} & \multicolumn{1}{c}{\textbf{83.24}} & \textbf{77.19}         &\textbf{74.64} & \textbf{71.60}   \\\cline{3-9}
                      &                         & \multirow{2}{*}{30}    &             1         & \multicolumn{1}{c}{\textbf{82.31}} & \multicolumn{1}{c}{\textbf{81.03}} &\textbf{77.18}       &\textbf{74.53} & 71.15     \\
                      &                         &                      &                  2    & \multicolumn{1}{c}{\textbf{83.16}} & \multicolumn{1}{c}{\textbf{82.37}} & \textbf{77.47}     &\textbf{74.62} & \textbf{71.19}     \\ \hline\hline
\end{tabular}}
\label{112102}
\end{table}

\subsubsection{Secure Aggregation Protocol  for PPFPL}

To test the effectiveness of the aggregation protocol, we  evaluate the performance of PPFPL with $\chi$ of -1, 0, 0.2, and 0.5, respectively.
 Here, $\chi = -1$ means that  the protocol  performs  normalization verification and average aggregation.
The data distribution is set to Non-IID (i.e., \textit{Avg}=3, \textit{Std}=2).
In addition,    $\boldsymbol{X}$ in Table \ref{112101} denotes that client-submitted prototypes are unnormalized  and aggregated using simple averaging.
As shown in TABLE \ref{112101}, we can observe from $\boldsymbol{X}$ that the performance of  PPFPL  suffers degradation to some extent compared to other settings.
The performance degradation is because malicious prototype poisons the global prototype.
However, the poisoned global prototype can only affect part of  local model training, and cannot affect the minimization of local classification loss.
Thus, the model accuracy among benign clients in PPFPL remains relatively robust against the influence of malicious clients.
Furthermore, we observe that the presence of  detection threshold improves the performance of PPFPL, which indicates  that the secure aggregation protocol computes Byzantine-robust aggregation results.
The above observation raises the question: whether PPFPL can maintain high performance in high proportion of attacks? We address this question in the subsequent experiments.
\begin{table}[t]
\centering
  \caption{Test Average Accuracy (\%) on MNIST, FMNIST and CIFAR10 with diffierent $\chi$ on feature attacks.}
\begin{threeparttable}
\tabcolsep=0.43cm
\scalebox{0.75}{
\begin{tabular}{c|c|ccccc}
\hline\hline
\textbf{Datasets}   &\textit{Att} &\textbf{$\chi$} = \textbf{-1}& \textbf{$\chi$} = \textbf{0} & \textbf{$\chi$} = \textbf{0.2} & \textbf{$\chi$} = \textbf{0.5} & $\boldsymbol{X}$ \\ \hline
\multirow{4}{*}{\textbf{MNIST} }
&10\%&96.23 &97.86   &97.62     &98.16     &95.73     \\
					&20\%&96.10 &97.20   &97.55     &97.03     &94.35   \\
				&30\%&95.21 &97.10   &97.47     &97.82     &94.20     \\
 				&40\%&94.54 &96.75   &96.51     &96.48     &93.14    \\ \hline
\multirow{4}{*}{\textbf{FMNIST} }    &10\%&88.37 &90.68   &90.23     &90.78     &86.14    \\
&20\% &88.20 &90.28   & 90.36    &90.48     &86.12     \\
    &30\%&88.03  &90.14   &90.06     &90.16     &85.71    \\
   	&40\%&87.62  &90.23   &90.41     &90.04     &85.29   \\ \hline
\multirow{4}{*}{\textbf{CIFAR10} }  &10\% &81.66 &83.59   & 83.10    &83.24    &80.97    \\
&20\% &81.61 &83.57   & 83.67    &83.31     &79.53     \\
   &30\% &81.38 &83.45   & 83.35    &83.21     &78.08     \\
  &40\% &80.74 &83.05   & 83.27    &82.94     &78.00    \\ \hline\hline

\end{tabular}}
\end{threeparttable}
\label{112101}
\end{table}

\begin{figure}[t]
  \centering
  \subfloat[\textit{Att} = 20\%]
  {
      \label{110911}  \includegraphics[width=0.44\linewidth]{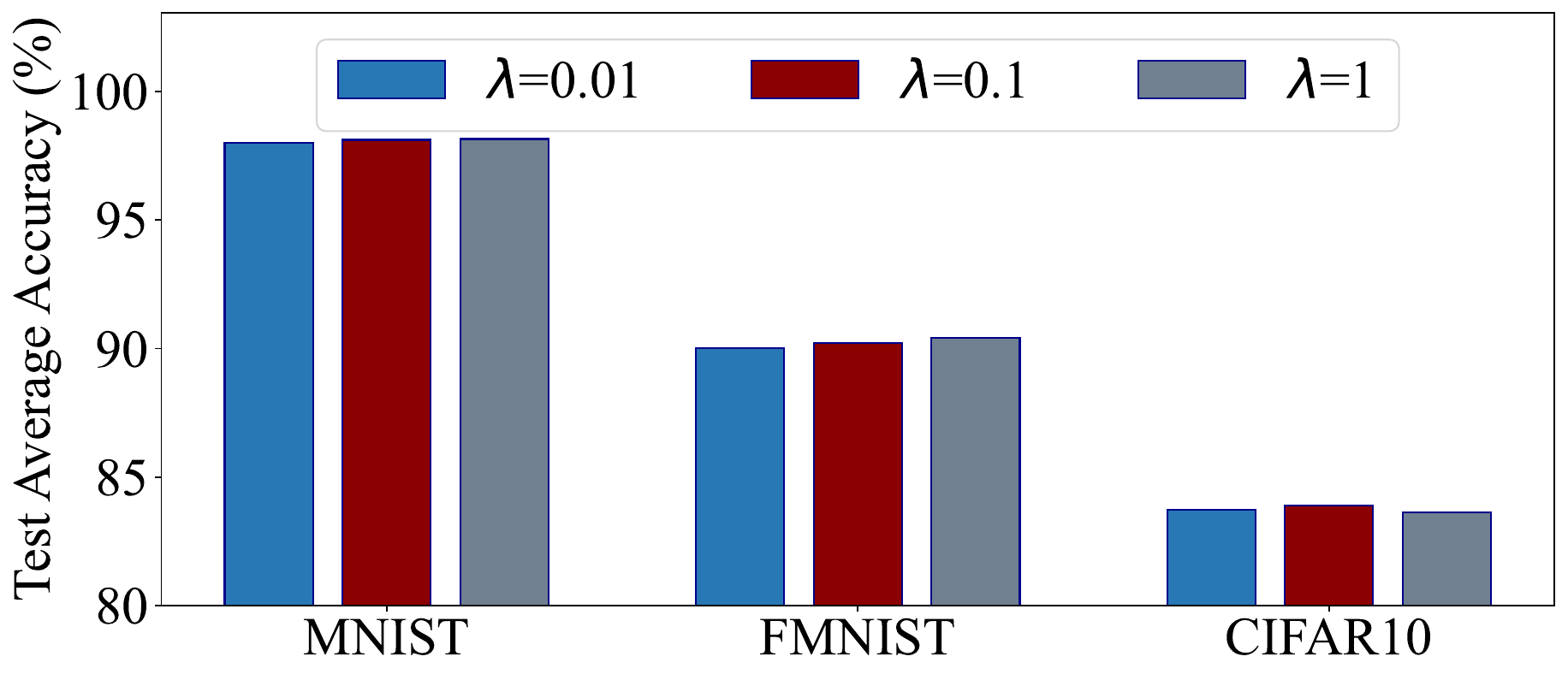}

  }
  \subfloat[\textit{Att} = 40\%]
  {
      \label{110912}  \includegraphics[width=0.44\linewidth]{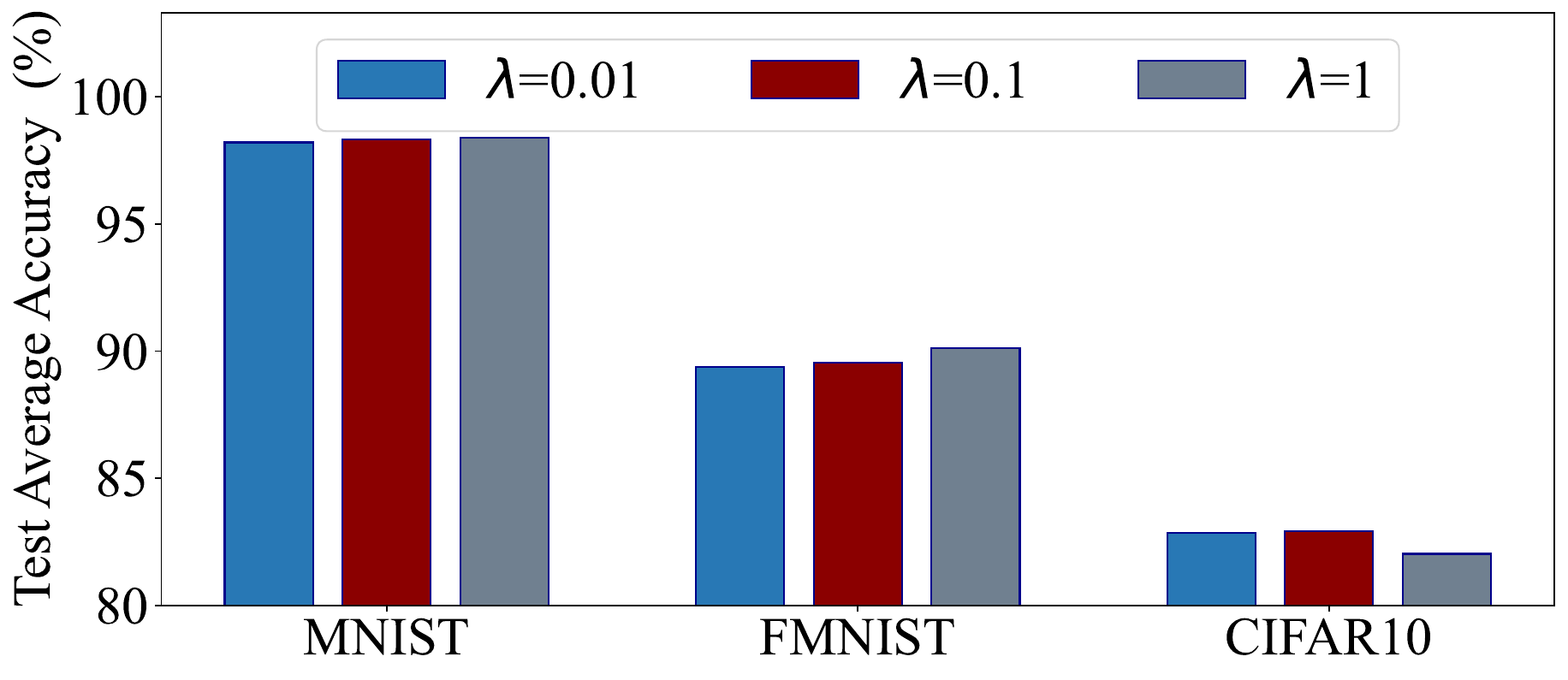}
  }
  
  \subfloat[ \textit{Att} = 60\%]
  {
      \label{110913}  \includegraphics[width=0.44\linewidth]{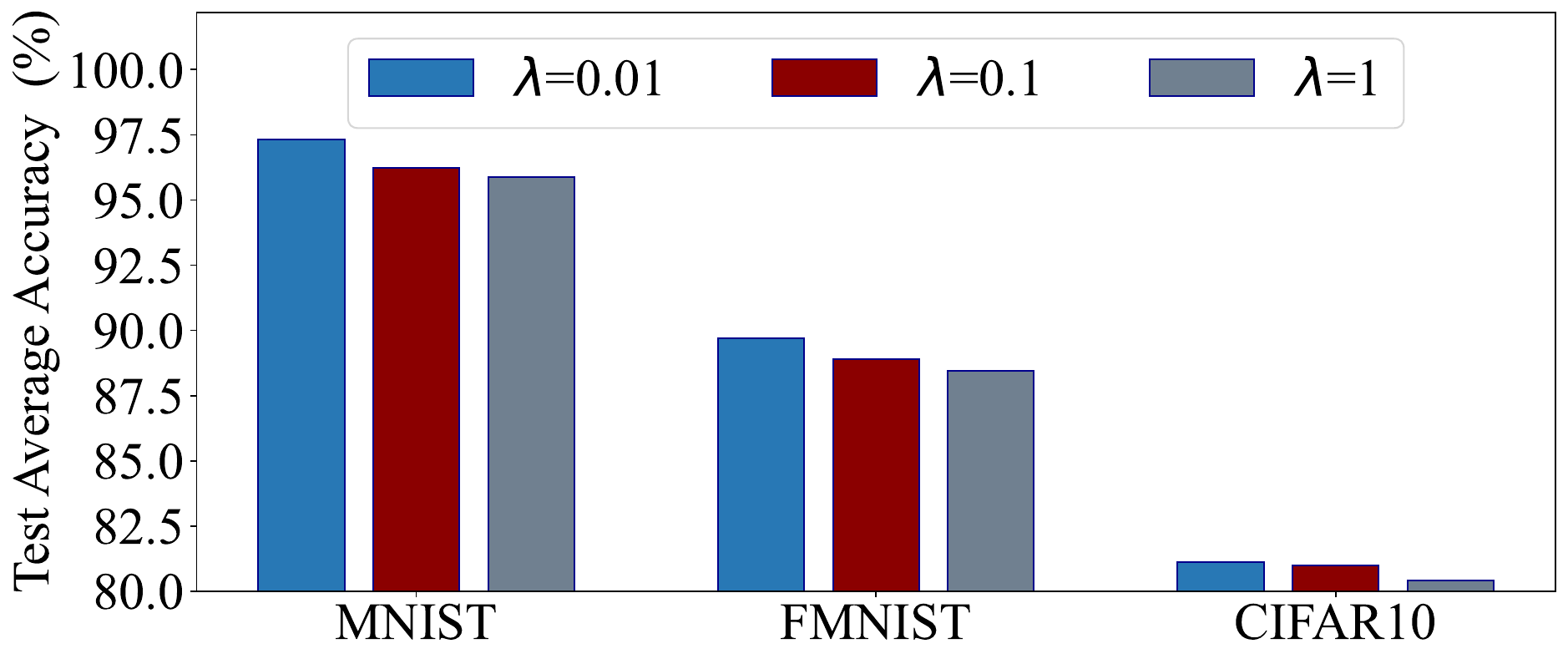}
  }
  \subfloat[ \textit{Att} = 80\%]
  {
      \label{110914}\includegraphics[width=0.44\linewidth]{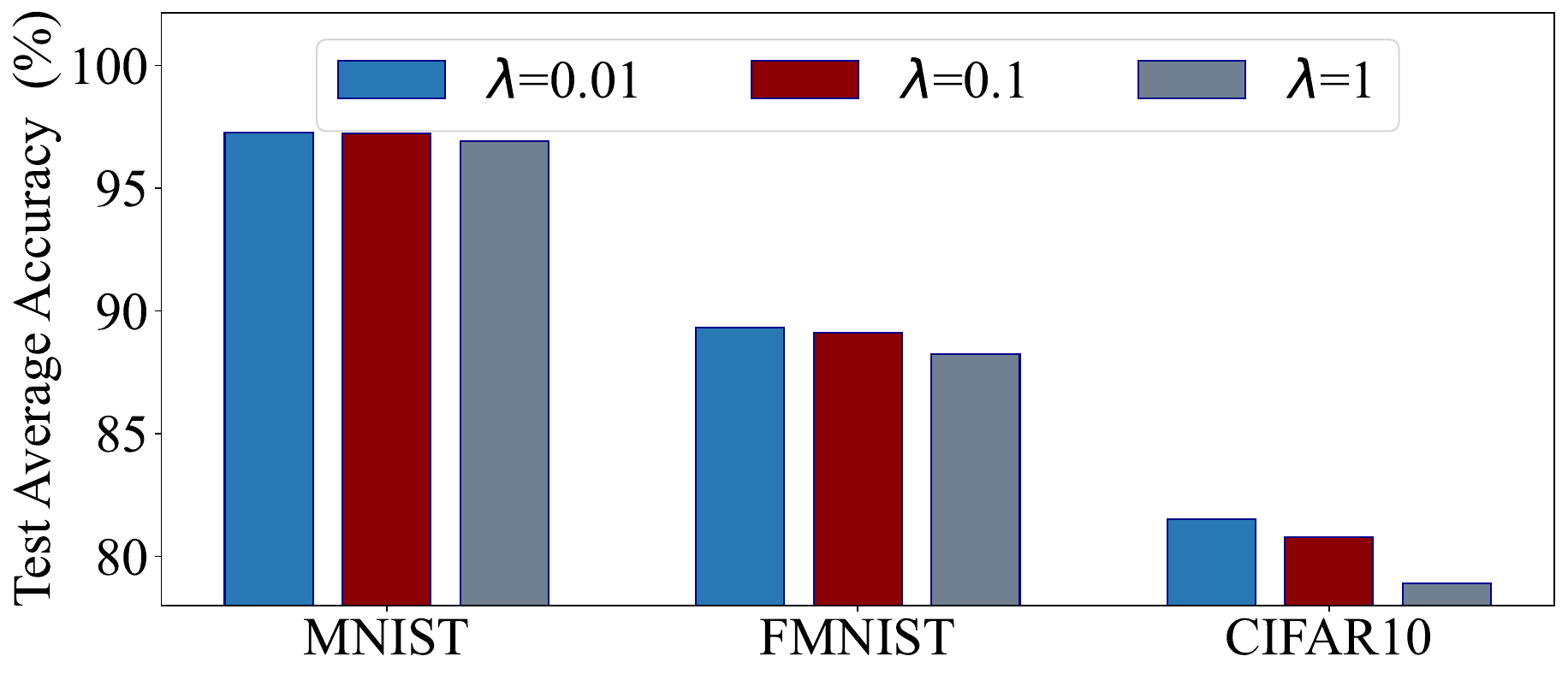}
  }
\caption{Test average accuracy of different $\lambda$ with different  proportions \textit{Att} with feature attacks.}
\label{11091}
\end{figure}


\subsubsection{High Proportion   Attacks for PPFPL}
To   evaluate the performance of PPFPL under high proportion attacks, we test it with  \textit{Att} = 20\%, 40\%, 60\% and 80\%.
In addition, we set   $\lambda$ as 0.01, 0.1 and 1 under Non-IID data distribution (\textit{Avg}=3, \textit{Std}=2).
The $\lambda$ denotes the  importance weight of the auxiliary term  in the loss function,  which can be considered as the degree of influence among clients.
The experimental results are shown in Fig. \ref{11091}.
Surprisingly, We  observe that PPFPL still has high performance when \textit{Att} is 60\% or even 80\%, which indicates that the training of benign clients is not interfered by malicious clients.
This is because the benign clients do not rely solely on information distributed by  two servers in federated prototype learning, but relies heavily on its local  training data.
Furthermore, we observe that when \textit{Att} is 20\% or 40\%, the accuracy  increases slightly with the growth  of $\lambda$  in the three datasets. On the contrary, when \textit{Att} is 60\% or 80\%, the accuracy tends to decrease with the increase of  $\lambda$.
This is because   the  larger $\lambda$ strengthens the collaboration among clients, but makes  PPFPL  more vulnerable to data poisoning attacks from malicious clients.
Conversely, the smaller value of $\lambda$ weakens  collaboration among clients, but  enhances resistance to data poisoning attacks.
Therefore, PPFPL can appropriately adjust the size of $\lambda$ according to actual conditions.

 \subsubsection{Dynamic Data Poisoning Attacks for PPFPL}

In practical cross-silo scenarios, the behavior of malicious clients may vary across different rounds. 
To emulate this scenario, we configure each malicious client to alternate between feature attacks and label attacks in consecutive rounds, termed dynamic poisoning attacks. 
We evaluate the performance of PPFPL under the dynamic attacks.
The experimental results (shown in TABLE \ref{9111951}) demonstrate that the performance of PPFPL only slightly decreases with the increase in the number of malicious clients, mainly due to the reduction in benign samples.
Thus, PPFPL maintains superior performance under dynamic poisoning attacks.
This is because, regardless of how  malicious clients manipulate  their submitted prototypes, our secure aggregation protocol can limit their impact to the directional perturbation.
The protocol filters out prototypes with excessive directional deviation, ensuring that PPFPL is resilient to dynamic poisoning attacks.

\begin{table}[t]
\centering
  \caption{Test Average Accuracy (\%) of PPFPL against dynamic  poisoning attacks.}
\begin{threeparttable}
\tabcolsep=0.42cm
\scalebox{0.7}{
\begin{tabular}{c|c|ccccc}
\hline\hline
\multirow{2}{*}{\textbf{Datasets}}   &\multirow{2}{*}{\makecell{\textit{Data}\\ \textit{Distributions}}} &\multicolumn{5}{c}{Proportion of Malicious Clients
 (\%)}   \\ 
  &&0\% &10\%&20\%  &30\% & 40\%  \\ \hline
\multirow{6}{*}{\textbf{CIFAR10} }
&\textit{Avg}=3, \textit{Avg}=1&83.05&83.20 &82.79  &82.27  & 81.55       \\
					&\textit{Avg}=3, \textit{Avg}=2&83.29&83.13 &82.83  &82.04    &81.26  \\
				&\textit{Avg}=4, \textit{Avg}=1&82.04&81.87 &80.71  &80.93   &80.60 \\
 				&\textit{Avg}=4, \textit{Avg}=2&82.11&82.06 &82.45   &81.75  & 81.21   \\ 
  				&\textit{Avg}=5, \textit{Avg}=1&77.82&77.31 &77.24  &76.92  & 76.49   \\ 
   				&\textit{Avg}=5, \textit{Avg}=2&77.95&78.12 &78.03   &77.80  & 76.22   \\ \hline\hline
\end{tabular}}
\end{threeparttable}
\label{9111951}
\end{table}

 \subsubsection{Stability of PPFPL}
\begin{figure}
  \centering
\begin{scriptsize}
  \subfloat[Feature attacks]
  {
      \label{22523161}  \includegraphics[width=0.44\linewidth]{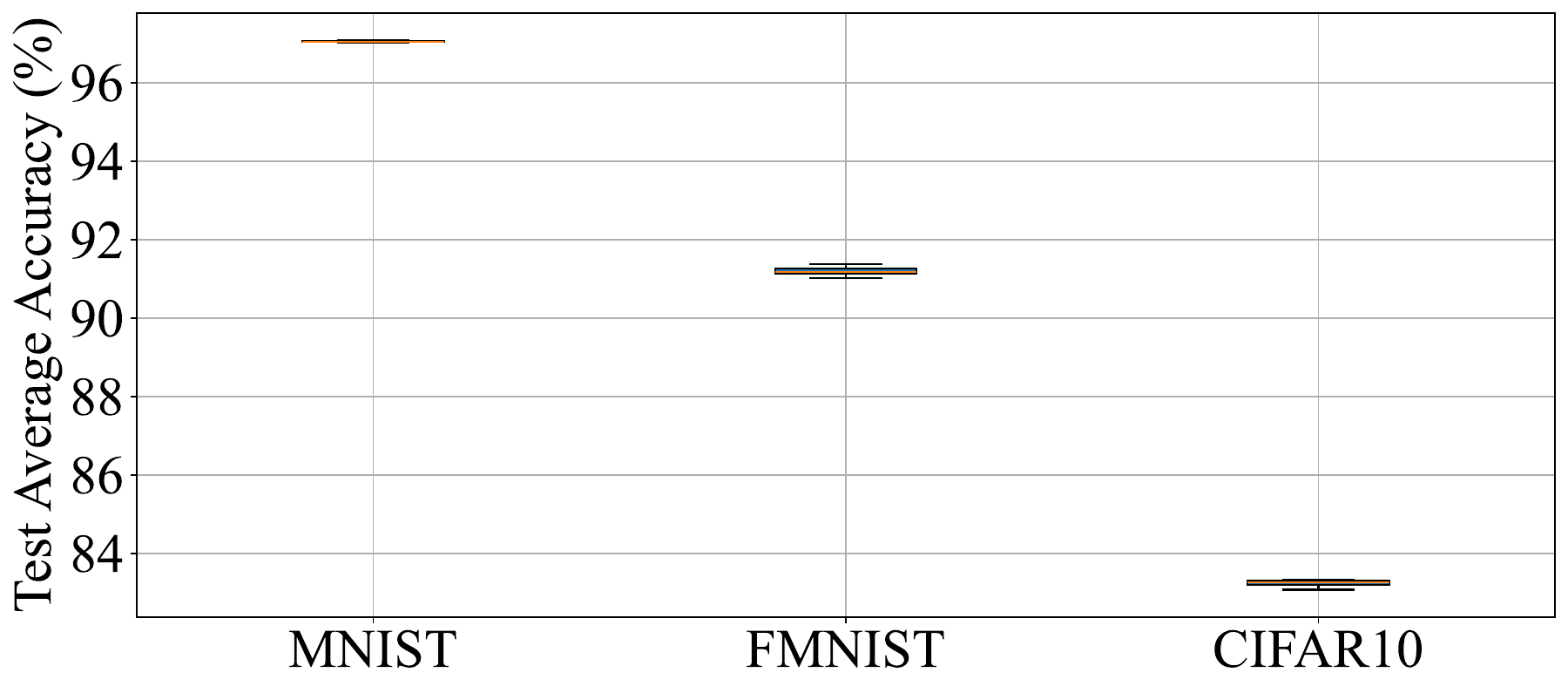}

  }
  \subfloat[Label attacks]
  {
      \label{22523162}  \includegraphics[width=0.44\linewidth]{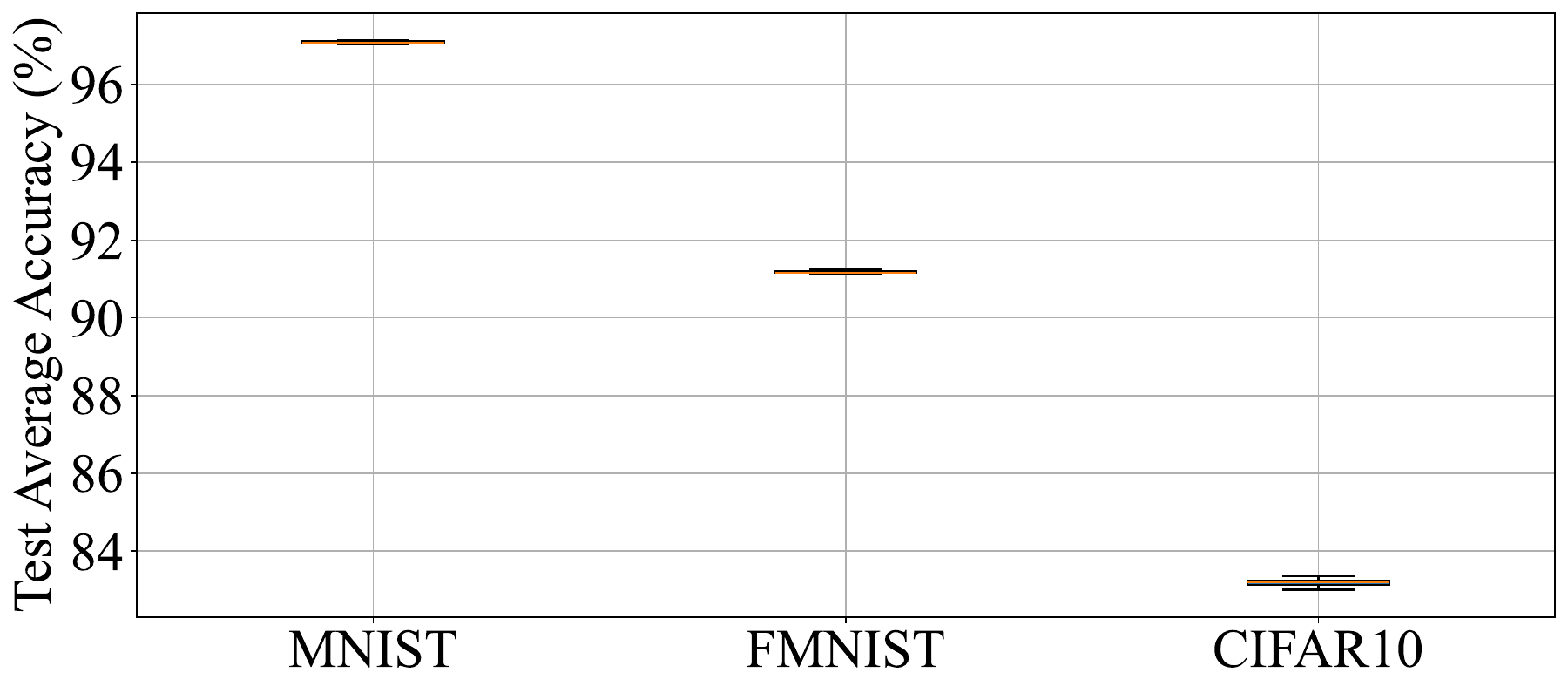}
  }
\caption{Stability of PPFPL under MNIST, FMNIST, and CIFAR10.}
\label{2252316}
\end{scriptsize}
\end{figure}

It is observed from Fig. \ref{10241} and \ref{10252355} that the performance of PPFPL exhibits some fluctuations. Therefore, it is important to evaluate  the stability of PPFPL over multiple randomized experiments. 
To verify the stability of PPFPL, we conduct experiments on MNIST, FMNIST, and CIFAR10 datasets, with 20\% of clients being malicious.
The data distribution follows the Non-IID setting (\textit{Avg}=3, \textit{Std}=2).
Each experiment is repeated 10 times.
The  boxplots in Fig. \ref{2252316} provides a clear visualization of the stability of PPFPL.
The results show that PPFPL's performance fluctuates less and remains stable against  data poisoning attacks.
The stability of PPFPL proves its reliability in real-world scenarios.

\subsubsection{Scalability  of PPFPL}

Federated learning is often deployed on an unknown number of clients.
Scalability directly determines whether PPFPL can maintain stable performance under different client scales.
To evaluate the scalability of PPFPL, we conduct  experiments on  MNIST, FMNIST, and CIFAR10 datasets with 20\% malicious clients,  varying the number of clients as 20,   30, 40,   80, and 100.
The data distribution is Non-IID (i.e., $\textit{Avg}=3,\textit{Std}=1$) setting.
Fig. \ref{312231} shows the scalability of PPFPL against data poisoning attacks.
It is observed that PPFPL maintains stable performance when the number of clients is 40.
However, the performance of PPFPL   degrades significantly when the number increases to 100. 
For example, the accuracy drops by 1.72\% on MNIST and 14.11\% on CIFAR10.
The main reason is that PPFPL follows a personalized training paradigm in which each client relies on its own data for training.
As the number of clients increases, the amount of data allocated to each client decreases.
The small amount of data in each client makes it difficult to obtain a high-quality model through local training.
This shows that the performance of PPFPL is closely related to the number of clients and ultimately depends on on the data available to each client.
Therefore, in large-scale networks, each client needs to have sufficient data to fully exploit the advantages of PPFPL.
Under our experimental conditions, the data volumes of  MNIST, FMNIST, and CIFAR10 datasets can reasonably support up to about 40 clients, within which PPFPL demonstrates good scalability.

\begin{figure}
  \centering
\begin{scriptsize}
  \subfloat[MNIST]
  {
      \label{3122311}  \includegraphics[width=0.32\linewidth]{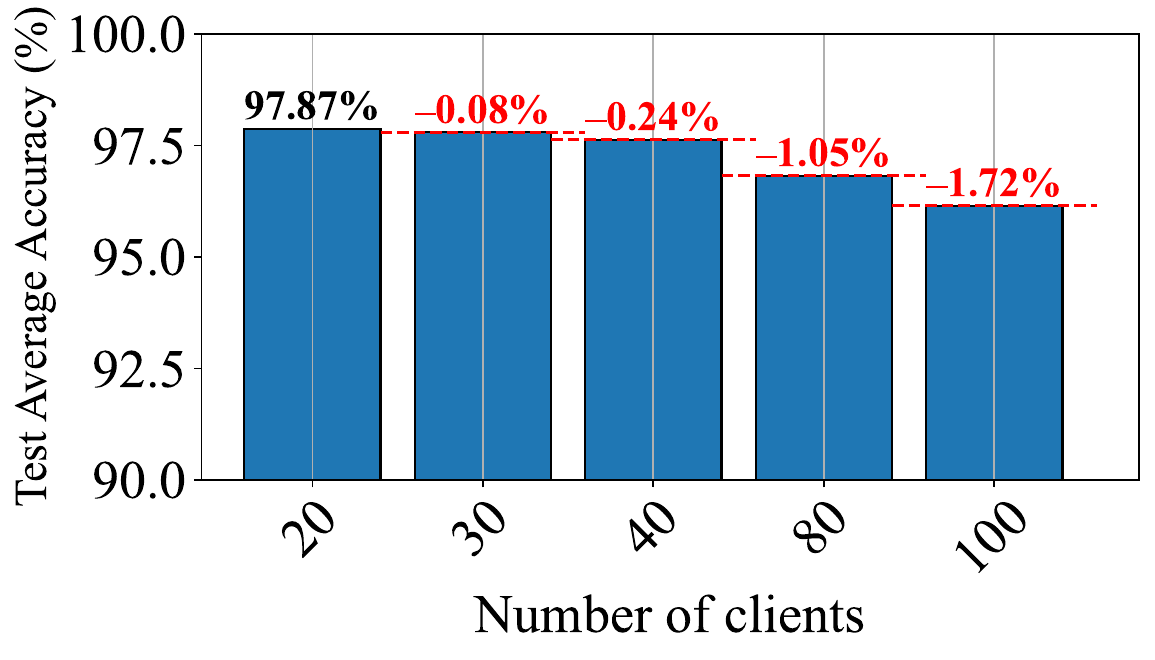}

  }
  \subfloat[FMNIST]
  {
      \label{3122312}  \includegraphics[width=0.32\linewidth]{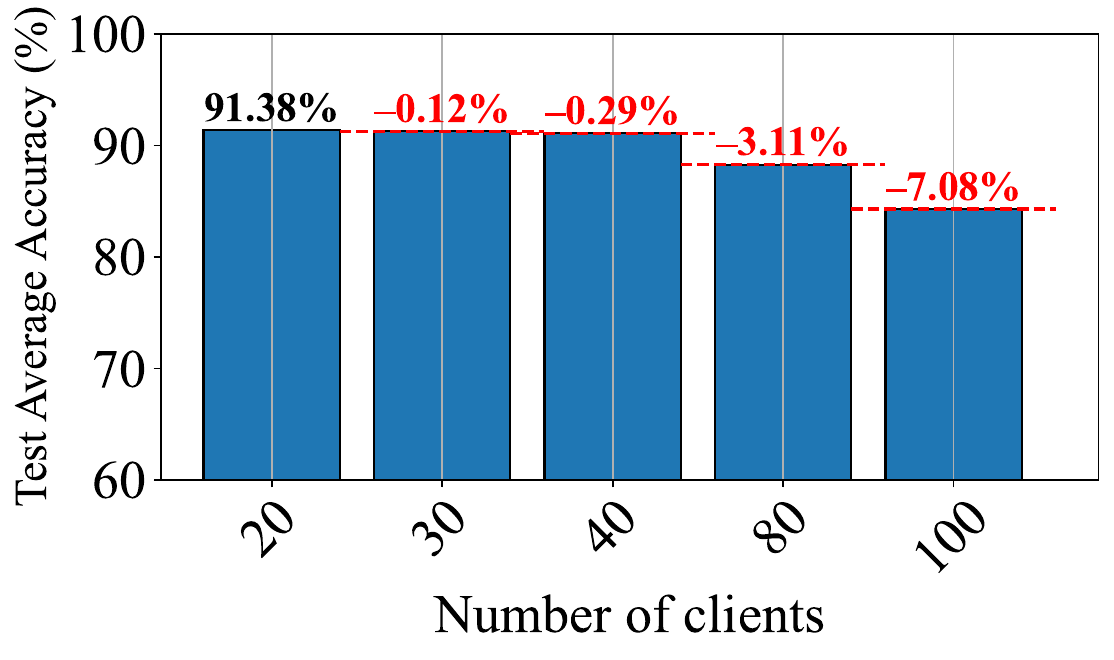}
  }
    \subfloat[CIFAR10]
  {
      \label{3122313}  \includegraphics[width=0.32\linewidth]{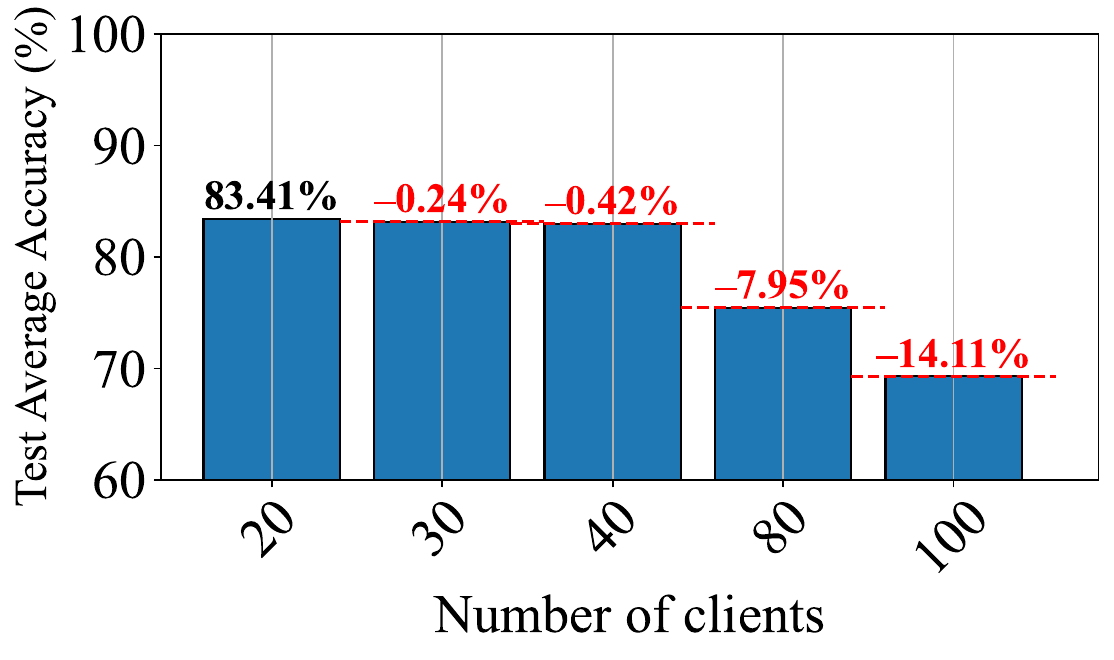}
  }
\caption{Scalability of PPFPL against  data poisoning   attacks over the MNIST, FMNIST, and CIFAR10.}
\label{312231}
\end{scriptsize}
\end{figure}

\subsubsection{Ablation Study}
To comprehensively evaluate the effectiveness of the secure aggregation protocol, we conduct  ablation experiments. Within the PPFPL framework, we remove  the normalization verification module (denoted as ``w/o NV'', where ``w/o'' stands for ``without'') and the secure aggregation protocol (denoted as ``w/o SAP''). 
We conduct experiments on CIFAR10 dataset, where 30\% of malicious clients launched attacks.
The experimental results are shown in TABLE \ref{3159191351}. 
Under both attack settings, PPFPL's performance degrades  regardless of which component is removed, demonstrating the effectiveness of normalization verification and the secure aggregation protocol. 
Notably, in the absence of malicious clients, removing the secure aggregation protocol actually led to optimal model performance. 
This phenomenon is attributed to the increase in the number of benign training samples, which enhances the overall FL performance.

\begin{table}[t]
\caption{Performance impact of different components in PPFPL.}
\tabcolsep= 0.54cm
\centering
\scalebox{0.85}{
\begin{tabular}{ccccc}
\hline\hline 
 \multirow{2}{*}{Schemes} &  \multirow{2}{*}{\makecell{Data \\Distribution}}  &\multicolumn{3}{c}{Test Average Accuracy (\%)}  \\
&&Feature   &Label  &No \\\hline 
 \multirow{2}{*}{PPFPL} &\textit{Avg}=3, \textit{Avg}=1&\textbf{82.31}  &\textbf{82.50}  &83.07 \\
&\textit{Avg}=3, \textit{Avg}=2&\textbf{83.16}  &\textbf{83.42}  &83.29 \\\hline 
 \multirow{2}{*}{w/o NV} &\textit{Avg}=3, \textit{Avg}=1&80.26  &80.29  &83.20 \\
&\textit{Avg}=3, \textit{Avg}=2&80.74  &80.61  &83.34\\\hline 
 \multirow{2}{*}{w/o SAP} &\textit{Avg}=3, \textit{Avg}=1&77.69  &77.24  &\textbf{83.53}\\
&\textit{Avg}=3, \textit{Avg}=2&78.03  &78.53 &\textbf{83.64} \\\hline 
\hline
\end{tabular}}
\justifying
Note:  ``Feature'' indicates feature attacks, ``Label'' indicates label attacks, and ``No''   indicates the absence of malicious clients. In addition, ``NV'' indicates the normalization verification module, ``SAP'' indicates the secure aggregation  protocol. Best rusults are in \textbf{blod}.
\label{3159191351}
\end{table}





\subsubsection{Efficiency Evaluation}

\begin{figure}[t]
  \centering
  \subfloat[]
  {
      \label{103116591}  \includegraphics[width=0.42\linewidth]{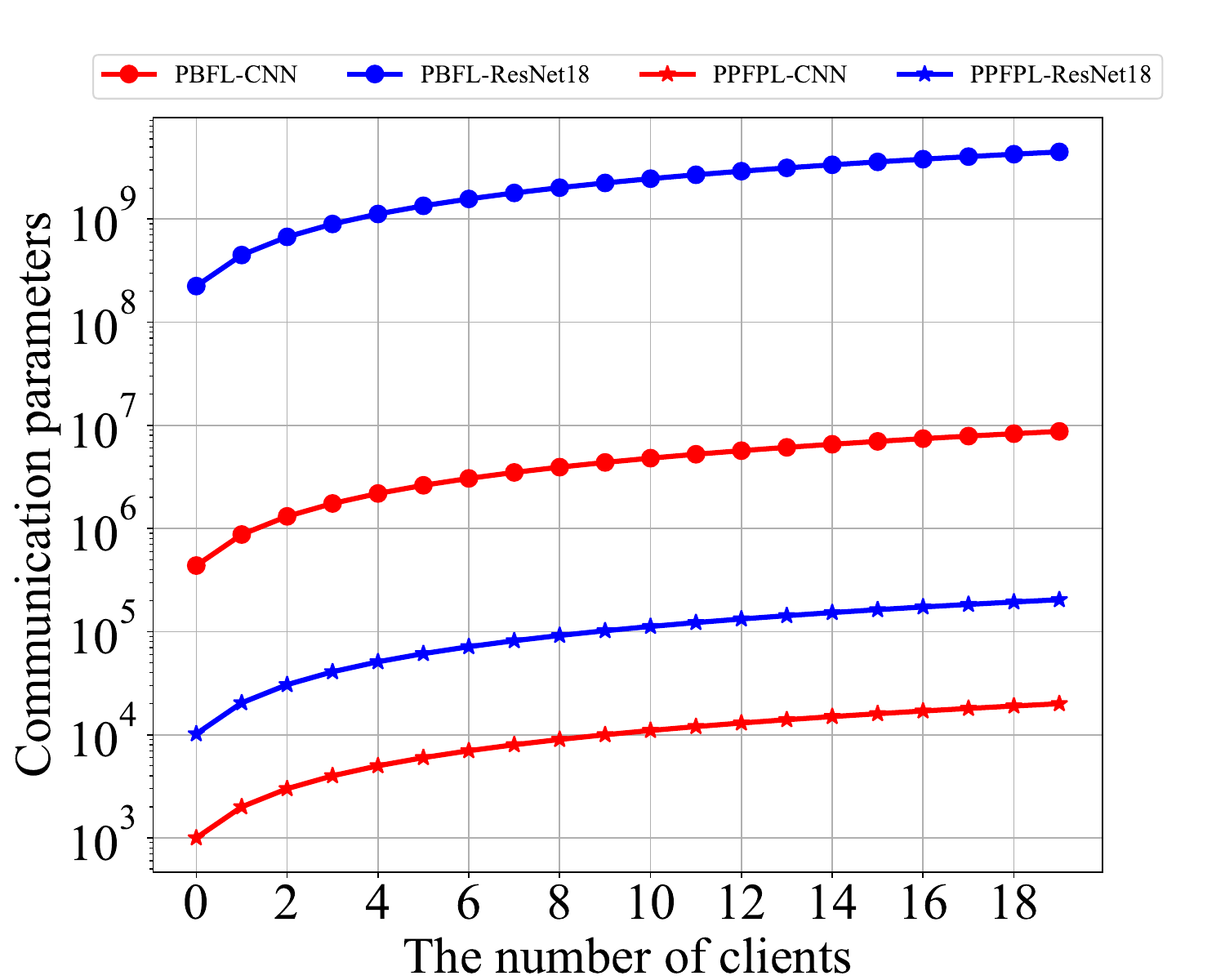}

  }
  \subfloat[]
  {
      \label{103116592}  \includegraphics[width=0.48\linewidth]{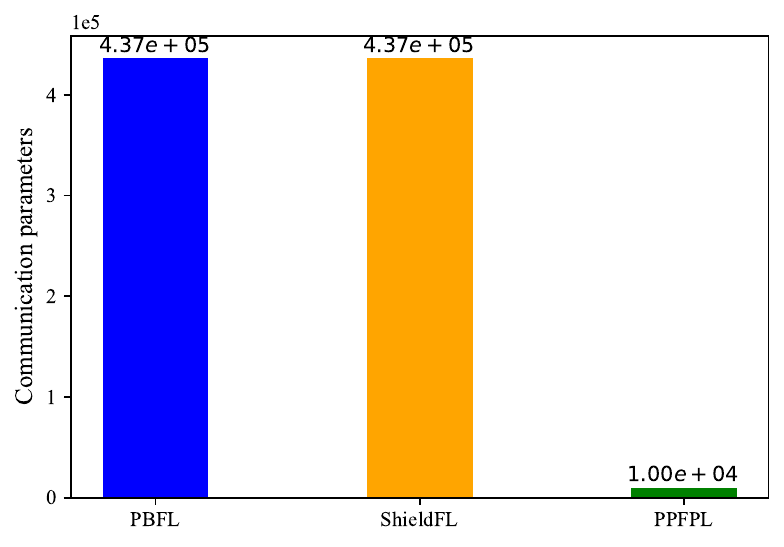}
  }
\caption{Communication parameters. (a) The number of parameters submitted by  clients under different local models. (b) The number of parameters submitted by  client using CNN model in one round under different scheme.}
\label{10311659}
\end{figure}

To evaluate the efficiency of PPFPL, we measure  the number of parameters submitted by clients, as reported in Fig. \ref{10311659}.
We   observe  that  both PBFL and ShieldFL require clients to submit a fixed number of parameters determined by the model architecture, resulting in identical parameter volumes  across all clients.
In contrast, the number of parameters submitted by clients in PPFPL is significantly lower  under the same model architecture.
This is because  the number of prototype parameters in PPFPL is determined by the output of the feature extractor rather than the model architecture itself.
Consequently, PPFPL effectively reduces the parameter volume transmitted to the servers, thereby decreasing both computational and communication overhead.

\begin{table}
\centering
\caption{The cost time of ciphertext operation using CNN model under PBFL, ShieldFL, and PPFPL.}
\tabcolsep=0.8cm
\scalebox{0.9}{
\begin{tabular}{lccc}
\hline
Operation & PBFL  & ShieldFL  & PPFPL \\
\hline
Encrypt  & 1.72s &1.42s &0.62s \\
Decrypt  &1.35s &1.19s &0.40s\\
\hline
\end{tabular}}
\label{2131403}
\end{table}

Furthermore, ciphertext operations are the primary  factor affecting the efficiency of our framework.
We evaluate  the time cost of ciphertext operations for PPFPL, PBFL, and ShieldFL. 
Notably, PBFL employs the CKKS, while ShieldFL is built upon two-door HE.
We measure  the time spent on encryption and decryption per round for each client as our evaluation metric.
As shown in Table \ref{2131403}, PPFPL's performance in both encryption and decryption is lower than that of PBFL and ShieldFL.
This advantage stems from the fact that PPFPL requires fewer prototype parameters submitted by each client, significantly reducing the computational cost of the ciphertext operations.

\subsection{Complexity Analysis}
\begin{table}
\centering
\caption{Computation  overload  and communication overload of client.}
\scalebox{0.9}{
\begin{tabular}{lcc}
\hline
Scheme & Computation  Overload  & Communication Overload   \\
\hline
PEFL \cite{9524709} & $\mathcal{O}(T_{tr})+\mathcal{O}(gT_{ch})$ &$\mathcal{O}(gP_{h})$ \\
ShieldFL \cite{9762272} &$\mathcal{O}(T_{tr})+\mathcal{O}(gT_{mul})+\mathcal{O}(gT_{ch})$  &$\mathcal{O}(gP_{h})$\\
PBFL \cite{9849010}  &$\mathcal{O}(T_{tr})+\mathcal{O}(gT_{ch})$ &$\mathcal{O}(gP_{h})$\\
PPFPL & $\mathcal{O}(T_{tr})+\mathcal{O}(T_{pro})+ \mathcal{O}(pT_{ch})$ & $\mathcal{O}(nP_{h})$ \\
\hline
\end{tabular}}
\label{3161102}
\end{table}

We analyze the computational   and communication overload for each client in PPFPL  and compared it with similar schemes  \cite{9849010}\cite{9762272}\cite{9524709}, as shown in TABLE \ref{3161102}.
The computational overload for  clients in PPFPL consists of   three components:  local model training,
prototype generation, and encryption.
Formally, the computational overhead  can be expressed as  $\mathcal{O}(T_{tr})+\mathcal{O}(T_{pro})+ \mathcal{O}(pT_{ch})$,  where  $p$ is the number of prototype parameters,   $T_{ch}$ denotes the  encryption time overhead, $T_{tr}$ denotes the time overhead for  local model training, and $T_{pro}$ denotes  the time overhead of prototype generation.
Notably, the overhead of prototype generation  is significantly  lower than that of local model training, i.e., $T_{tr} \gg T_{pro}$.
Other comparable  schemes \cite{9849010}\cite{9762272}\cite{9524709} adopt homomorphic encryption on gradients, resulting in an encryption overhead $\mathcal{O}(gT_{ch})$, where $g$ denotes  the number of gradient parameters.
Since the relation $g \textgreater  p$ holds,  the  overload  relation satisfies  $\mathcal{O}(gT_{ch})\textgreater \mathcal{O}(pT_{ch})$.
Hence,  the  relation of ciphertext operations is $\mathcal{O}(gT_{mul})+\mathcal{O}(gT_{ch}) \textgreater \mathcal{O}(gT_{ch}) \textgreater  \mathcal{O}(pT_{ch}) $.
In addition,  the client communication overhead   in PPFPL is $ \mathcal{O}(nP_{h})$,  where
$P_{h}$ denotes the communication complexity per parameter.
Compared with the schemes in  \cite{9849010}\cite{9762272}\cite{9524709}, the inequality $\mathcal{O}(gP_{h})\textgreater \mathcal{O}(pP_{h})$ holds, demonstrating the efficiency advantage of PPFPL.

\section{Conclusion}
\label{sec:concludesq}
This paper proposes PPFPL, a privacy-preserving federated prototype learning framework that effectively enhances federated learning performance in poisoned Non-IID data environments.
Through comprehensively theoretical analysis and experimental validation, we  demonstrate the significant advantages of PPFPL compared existing FL defense schemes.
The framework provides robust security guarantees for distributed computing scenarios, with important practical applications in privacy-sensitive domains such as finance and healthcare.

However, our  current framework has certain limitations. 
First,   frequent interactions in the two-server architecture increase system failures caused by communication latency. 
Second, existing experimental datasets  do not adequately capture  the complexity of real-world cross-silo scenarios, making it difficult to effectively evaluate prototypes generated by benign clients under varying data quality conditions.
This limitation may result in unfair evaluation outcomes in practical settings.
 Therefore, our future work will focus on two directions: (i) designing more efficient two-server interaction protocols while ensuring privacy and security.
(ii) constructing benchmark datasets that capture the characteristics of cross-silo scenarios to  evaluate  FL system performance.

\bibliographystyle{IEEEtran}
\bibliography{samplebase}

\end{document}